\documentclass{lmcs}
\pdfoutput=1
\usepackage[utf8]{inputenc}

\usepackage{lmodern}
\usepackage[T1]{fontenc}
\usepackage{stmaryrd}
\usepackage[pdfencoding=auto]{hyperref}

\usepackage{aliascnt}
\newaliascnt{definition}{thm}
\newaliascnt{proposition}{thm}
\newaliascnt{lemma}{thm}
\newaliascnt{corollary}{thm}
\newaliascnt{conjecture}{thm}
\newaliascnt{remark}{thm}
\newaliascnt{example}{thm}

\theoremstyle{plain}
\newtheorem{theorem}[thm]{Theorem}

\aliascntresetthe{proposition}
\newtheorem{lemma}[lemma]{Lemma}
\aliascntresetthe{lemma}
\newtheorem{corollary}[corollary]{Corollary}
\aliascntresetthe{corollary}

\theoremstyle{definition}
\newtheorem{definition}[definition]{Definition}
\aliascntresetthe{definition}
\newtheorem{example}[example]{Example}
\aliascntresetthe{example}

\usepackage{cleveref}

\usepackage{tikz}
\usetikzlibrary{matrix,arrows,positioning,automata}
\newenvironment{diagram}{\begin{tikzpicture}[->,auto,>=latex,line width=.5pt]}{\end{tikzpicture}}%

\usepackage{wrapfig}

\setlist[itemize]{%
	labelindent=*,
	leftmargin=*,
	label=$\triangleright$%
}%

\setlist[enumerate]{%
	font=\normalfont,
	labelindent=*,
	leftmargin=*,
	label=\arabic*.,%
	ref=\arabic*%
}

\newlist{wbconditions}{enumerate}{1}
\setlist[wbconditions,1]{%
	font=\normalfont,
	labelindent=*,
	leftmargin=*,
	label=W\arabic*.,%
	ref=W\arabic*%
}

\crefname{figure}{figure}{figures}
\Crefname{figure}{Figure}{Figures}

\crefname{equation}{equation}{equations}
\Crefname{equation}{Equation}{Equations}

\crefname{page}{page}{pages}
\Crefname{page}{Page}{Pages}

\crefname{section}{section}{sections}
\Crefname{section}{Section}{Sections}

\crefname{definition}{definition}{definitions}
\Crefname{definition}{Definition}{Definitions}

\crefname{example}{example}{examples}
\Crefname{example}{Example}{Examples}

\crefname{thm}{theorem}{theorems}
\Crefname{thm}{Theorem}{Theorems}

\crefname{theorem}{theorem}{theorems}
\Crefname{theorem}{Theorem}{Theorems}

\crefname{lemma}{lemma}{lemmas}
\Crefname{lemma}{Lemma}{Lemmas}

\crefname{corollary}{corollary}{corollaries}
\Crefname{corollary}{Corollary}{Corollaries}

\crefname{wbconditionsi}{condition}{conditions}
\Crefname{wbconditionsi}{Condition}{Conditions}

\renewcommand{\phi}{\varphi}
\newcommand{\set}[1]{\left\{#1\right\}}
\newcommand{\powerset}[1]{\mathcal{P}\left({#1}\right)}
\newcommand{\cat}[1]{\ensuremath{\mathbb{#1}}}
\newcommand{\Set}{\ensuremath{\mathnormal{Set}}}%
\newcommand{\PMet}{\mathnormal{PMet}}
\newcommand{\R}{\ensuremath{\mathbb{R}}}%
\newcommand{\N}{\ensuremath{\mathbb{N}}}%
\newcommand{\lbbd}{\mathopen{[\![}}
\newcommand{\rbbd}{\mathclose{]\!]}}
\newcommand{\final}[1]{\lbbd #1 \rbbd}
\newcommand{\CoAlg}[1]{\ensuremath{\mathnormal{CoAlg}\left(#1\right)}}
\newcommand{\Kl}{\ensuremath{\mathcal{K}\ell}}
\newcommand{\EM}{\ensuremath{\mathcal{EM}}}
\DeclareSymbolFont{bbold}{U}{bbold}{m}{n}
\DeclareSymbolFontAlphabet{\mathbbold}{bbold}
\newcommand{\one}{\ensuremath{\mathbbold{1}}}
\newcommand{\two}{\ensuremath{\mathbbold{2}}}
\newcommand{\nonexpansiveTo}{\ensuremath{\,\mathop{\to}\,}}
\newcommand{\reals}{{[0,\top]}}
\newcommand{\prealinf} {\ensuremath{[0,\infty]}}
\newcommand{\preal}{\mathbb{R}_0^+}
\newcommand{\bd}{\mathnormal{bd}}
\newcommand{\td}{\mathnormal{td}}
\newcommand{\Id}{\ensuremath{\mathrm{Id}}}
\newcommand{\id}{\ensuremath{\mathrm{id}}}
 
\newcommand{\ev}{\ensuremath{ev}}
\newcommand{\Kantorovich}[2]{\ensuremath{#2^{\,\uparrow {#1}}}}

\newcommand{\LiftedFunctor}[1]{\ensuremath{\overline{#1}}}
\newcommand{\EvaluationFunctor}[1]{\ensuremath{\widetilde{#1}}}
\newcommand{\Wasserstein}[2]{\ensuremath{#2^{\,\downarrow {#1}}}}

\newcommand{\LiftedMetric}[2]{\ensuremath{#2^{#1}}}
\newcommand{\Couplings}[1]{\Gamma_{#1}}
\newcommand{\Powerset}{\ensuremath{\mathcal{P}}}
\newcommand{\PowersetFinite}{\ensuremath{\mathcal{P}_{\!f}}}
\newcommand{\Distributions}{\ensuremath{\mathcal{D}}}
\newcommand{\DistributionsFinite}{\ensuremath{\mathcal{D}_{\!f}}}
\newcommand{\supp}{\ensuremath{\mathrm{supp}}}

\newcommand{\textlsc}[1]{\textsc{\MakeLowercase{#1}}}
\newcommand{\name}[2][]{{#2}}

\newcommand{\new}[1]{#1}

\begin{document}

% Metadata
\title[Coalgebraic Behavioral Metrics]{Coalgebraic Behavioral Metrics}

\author[P. Baldan]{Paolo Baldan}
\address{Dipartimento di Matematica, Università di Padova, Italy}
\email{baldan@math.unipd.it}

\author[F. Bonchi]{Filippo Bonchi}
\address{Dipartimento di Informatica, Università di Pisa, Italy}
\email{bonchi@di.unipi.it}

\author[H. Kerstan]{Henning Kerstan}
\address{}
\email{mail@henningkerstan.de}

\author[B. König]{Barbara König}
\address{Universität Duisburg-Essen, Germany}
\email{barbara\_koenig@uni-due.de}



% PDF properties
\hypersetup{%
pdftitle={Coalgebraic Behavioral Metrics},%
pdfauthor={Paolo Baldan, Filippo Bonchi, Henning Kerstan, Barbara König},%
pdfkeywords={behavioral metric, bisimilarity metric, trace metric, functor lifting, monad lifting, pseudometric, coalgebra}
}

\begin{abstract}
We study different behavioral metrics, such as those arising from both branching and linear-time semantics, in a coalgebraic setting. Given a coalgebra $\alpha\colon X \to HX$ for a functor $H \colon \Set\to \Set$, we define a framework for deriving pseudometrics on $X$ which measure the behavioral distance of states.

A crucial step is the lifting of the functor $H$ on $\Set$ to a functor $\LiftedFunctor{H}$ on the category $\PMet$ of pseudometric spaces. We present two different approaches which can be viewed as generalizations of the Kantorovich and Wasserstein pseudometrics for probability measures. We show that the pseudometrics provided by the two approaches coincide on several natural examples, but in general they differ.

If $H$ has a final coalgebra, every lifting $\LiftedFunctor{H}$ yields in a canonical way a behavioral distance which is usually branching-time, i.e., it generalizes bisimilarity. In order to model linear-time metrics (generalizing trace equivalences), we show sufficient conditions for lifting distributive laws and monads. These results enable us to employ the generalized powerset construction.
\end{abstract}

\maketitle

\section{Introduction}
When considering the behavior of state-based system models embodying
quantitative information, such as probabilities, time or cost, the interest normally shifts from behavioral equivalences to behavioral distances. In fact, in a quantitative setting, it is often quite unnatural to ask that two systems exhibit exactly the same behavior, while it can be more reasonable to require that the distance between their behaviors is sufficiently small \cite{GJS90,DGJP04,vBW05,bblm:total-variation-markov,afs:linear-branching-metrics-quant,dAFS09,flt:quantitative-spectrum}.

Coalgebras \cite{Rut00} are a well-established abstract framework where a canonical notion of behavioral equivalence can be uniformly derived. In a nutshell, a coalgebra for a functor $H\colon \cat{C} \to \cat{C}$ is an arrow $\alpha \colon X \to H X$
in some category $\cat{C}$. When $\cat{C}$ is $\Set$, the category of sets and functions, coalgebras represent state machines: $X$ is a set of states, $\alpha$ the transition function, and $H$ describes the type of the transitions performed. For instance, deterministic automata are coalgebras $\langle o,t\rangle \colon X \to \two \times  X^A$: for any state $x\in X$, $o$ specifies whether it is final ($o(x)=1$) or not ($o(x)=0$), and $t$ the successor state for any given input in $A$. Similarly, nondeterministic automata can be seen as coalgebras $\langle o,t\rangle\colon  X \to \two \times  \mathcal{P}(X)^A$ where now $t$ assigns, for every state and input, a set of possible successors. 
Under certain conditions on $H$, there exists a final $H$-coalgebra,
namely a coalgebra $\zeta\colon Z \to HZ$ such that for every
coalgebra $\alpha\colon X \to HX$ there is a unique coalgebra
homomorphism $\final{\cdot} \colon X \to Z$. When the base category is
$\Set$, $Z$ can be thought of as the set of all $H$-behaviors and
$\final{\cdot} $ as the function mapping every state in $X$ into its
behavior. Then two states $x,y\in X$ are said to be behaviorally
equivalent when  $\final{x}  = \final{y} $. For deterministic automata this equivalence coincides with language equivalence while, for nondeterministic automata, it coincides with bisimilarity.

In this paper, we investigate how to exploit coalgebras to derive canonical notions of behavioral distances. The key step in our approach is to lift a functor $H$ on $\Set$ to a functor $\LiftedFunctor{H}$ on $\PMet$, the category of pseudometric spaces and nonexpansive maps. Given a pseudometric space $(X,d_X)$, the goal is to define a suitable pseudometric on $HX$.
For the (discrete) probability distribution functor $\Distributions$, there are two 
liftings, known as Wasserstein and Kantorovich liftings,
that have been extensively studied in transportation
theory~\cite{Vil09}. The Kantorovich-Rubinstein duality states that these two liftings coincide.

It is our observation that these notions of liftings can analogously be defined for arbitrary functors $H$, leading to a rich general
theory (\Cref{sec:lifting}). As concrete examples, besides $\Distributions$, we study the (finite)
powerset functor (resulting in the Hausdorff metric), the input functor, the coproduct and product bifunctors. The Kantorovich-Rubinstein duality holds for these, but it does not hold in general (we provide a counterexample).

Once a lifting $\LiftedFunctor{H}$ has been defined, it is easy to
derive a behavioral distance on the states of an $H$-coalgebra
$\alpha\colon X\to HX$. First, we turn $\alpha$ into a coalgebra for
$\LiftedFunctor{H}$ by endowing $X$ with the trivial discrete distance. Then,
for any two states $x,y\in X$, their behavioral distance can be
defined as the distance of their images in a final
$\LiftedFunctor{H}$-coalgebra
$\zeta\colon (Z,d_Z)\to \LiftedFunctor{H}(Z,d_Z)$, namely as
$d_Z(\final{x} ,\final{y} )$. This notion of distance is well behaved
with respect to the underlying behavioral equivalence: we show
(\Cref{thm:final-coalgebra}) that a final $H$-coalgebra can be enriched with a pseudometric $d_Z$ on $Z$ yielding the final $\LiftedFunctor{H}$-coalgebra. This immediately implies that if
two states are behaviorally equivalent, then they are at distance $0$
(\Cref{thm:bisim-pseudometric}). Moreover, we show how to compute
distances on the final coalgebra as well as on arbitrary coalgebras
via fixed-point iteration and we prove that the pseudometric obtained
on the final coalgebra is indeed a metric (\Cref{thm:comp-dist,thm:d-omega-is-metric}). We recover behavioral metrics in the
setting of probabilistic systems \cite{DGJP04,vBW06} and of metric
transition systems \cite{dAFS09}.

The canonical notion of equivalence for coalgebras, in a sense, fully captures the behavior of the system as expressed by the functor $H$. Sometimes one is interested in coarser equivalences, ignoring some aspects of a computation, a notable example being language equivalence for nondeterministic automata. In the coalgebraic setting this can be achieved by means of the generalized powerset construction~\cite{SBBR13,JSS12,JSS15}. The starting observation is that the distinction between the behavior to be observed and the computational effects that are intended to be hidden, is sometimes formally captured by splitting the functor $H$ in two components, a functor $F$ for the observable behavior and a monad $T$ describing the computational effects, e.g.,  $\one + \_$, $\mathcal{P}$ or $\mathcal{D}$ provides partial, nondeterministic or probabilistic computations, respectively. For instance, the functor for nondeterministic automata $\two \times  \Powerset(X)^A$ can be seen as the composition of the functor $F X = \two \times X^A$ of deterministic automata, with the powerset monad $T = \mathcal{P}$, capturing nondeterminism. Trace semantics can be derived by viewing a coalgebra $X \to \two \times \mathcal{P}(X)^A$ as a coalgebra $\mathcal{P}(X) \to \two \times \mathcal{P}(X)^A$, via a determinization construction.  Similarly probabilistic automata can be seen as coalgebras of the form $X \to [0,1] \times \mathcal{D}(X)^A$, yielding coalgebras $\mathcal{D}(X) \to [0,1] \times \mathcal{D}(X)^A$ via determinization. In general terms, the determinization of a coalgebra $X\to FTX$ leads to an $F$-coalgebra $TX\to FTX$. Formally this  can be done whenever there is a distributive law between $F$ and $T$. In this case, $F$ lifts to  a functor $\widehat{F}$ in $\EM(T)$, the Eilenberg-Moore category of $T$, and the determinized coalgebra can be regarded as an $\widehat{F}$-coalgebra.

In the second part of this paper, we exploit the aforementioned
approach \cite{JSS15} for systematically deriving metric trace
semantics for coalgebras. The situation is summarized by the diagram in \Cref{fig:gen-pow-trace} (on \cpageref{fig:gen-pow-trace}). As a first step, building on our technique for lifting functors from $\Set$ to $\PMet$, we identify conditions under which also natural transformations, monads and distributive laws can be lifted (\Cref{prop:nt-lifting}).
In this way we obtain an adjunction between $\PMet$ and
$\EM(\LiftedFunctor{T})$, where $\LiftedFunctor{T}$ is the lifted
monad. Via the lifted distributive law we can transfer a functor
$\LiftedFunctor{F}\colon \PMet \to \PMet$ to an endofunctor
$\widehat{\LiftedFunctor{F}}$ on $\EM(\LiftedFunctor{T})$. By using the
discrete distance, coalgebras of the form $TX\to FTX$ can now
live in $\EM(\LiftedFunctor{T})$ and can be equipped with a trace
distance via a map into a final coalgebra. 

We show that this notion covers known or meaningful trace distances
such as a metric on languages for nondeterministic automata or a
variant of the total variation distance on distributions for
probabilistic automata
(\Cref{exa:final-coalgebra-tracedistance-nfa,exa:final-coalgebra-tracedistance-pa}).

\paragraph{\bf{Synopsis}}
Motivated by an example of probabilistic transition systems, we will explain the Wasserstein and Kantorovich liftings for $\Distributions$  in terms of transportation theory (\Cref{sec:motivation}). 
After some preliminaries on pseudometrics (\Cref{sec:pseudometrics}) and a quick look at further motivating examples (\Cref{sec:motivating-examples}) we will generalize the Kantorovich/Wasserstein liftings to arbitrary endofunctors on $\Set$ (\Cref{sec:lifting}). 
This leads to the definition of behavioral distance and the proof that it enjoys several desirable properties (\Cref{sec:final-coalgebra}). 
For trace pseudometrics, we need to study first compositionality of our liftings (\Cref{sec:compositionality}) and then the lifting of natural transformations and monads (\Cref{sec:monadlifting}). Using these results, we can then employ the generalized powerset construction \cite{JSS15} to obtain trace pseudometrics and show how this applies to nondeterministic and probabilistic automata (\Cref{sec:tracemetrics}). Finally, we will discuss related and future work (\Cref{sec:conclusion}).

\paragraph{\bf{Notation}}
We assume that the reader is familiar with the basic notions of coalgebras and category theory but will present some of the definitions below to introduce our notation.
For a function $f\colon X \to Y$ and sets $A \subseteq X$, $B \subseteq Y$ we write $f[A] := \set{f(a) \mid a \in A}$ for the \emph{image} of $A$ and $f^{-1}[B]= \set{x \in X \mid f(x) \in B}$ for the \emph{preimage} of $B$. If $Y \subseteq [0,\infty]$ and $f,g\colon X \to Y$ are functions we write $f \leq g$ if $f(x)\leq g(x)$ for all $x \in X$.
Given a natural number $n \in \N$ and a family $(X_i)_{i = 1}^n$ of sets $X_i$ we denote the projections of the (cartesian) product of the $X_i$ by $\pi_i^n\colon \prod_{i=1}^n X_i \to X_i,$ or just by $\pi_i$ if $n$ is clear from the context. For a source $(f_i\colon X \to X_i)_{i = 1}^n$ we denote the unique mediating arrow to the product by $\langle f_1,\dots,f_n\rangle \colon X \to \prod_{i=1}^{n}X_i$. Similarly, given a family of arrows $(f_i\colon X_i \to Y_i)_{i = 1}^n$, we write $f_1 \times \dots \times f_n = \langle f_1 \circ \pi_1,\dots,f_n \circ \pi_n \rangle\colon \prod_{i=1}^n X_i \to \prod_{i=1}^n Y_i$.

We quickly recap the basic ideas of coalgebras. Let $F$ be an
endofunctor on the category $\Set$ of sets and functions. An
$F$-coalgebra is just a function $\alpha \colon X \to FX$. Given
another $F$-coalgebra $\beta\colon Y \to FY$ a coalgebra homomorphism
from $\alpha$ to $\beta$ is a function $f \colon A \to B$ such that
$\beta \circ f = Ff \circ \alpha$. We call an $F$-coalgebra
$\kappa\colon \Omega \to F\Omega$ \emph{final} if for any other
coalgebra $\alpha \colon X \to FX$ there is a unique coalgebra
homomorphism $\final{\cdot}_\alpha \colon X \to \Omega$. The final
coalgebra need not exist but if it does it is unique up to
isomorphism. It can be considered as the universe of all possible
behaviors. If we have an endofunctor $F$ such that a final coalgebra
$\kappa\colon \Omega \to F\Omega$ exists then for any coalgebra
$\alpha\colon X \to FX$ two states $x_1, x_2 \in X$ are said to be
\emph{behaviorally equivalent} if and only if
$\final{x_1}_\alpha = \final{x_2}_\alpha$. 

A \emph{probability distribution} on a given set $X$ is a function $P\colon X \to [0,1]$ satisfying $\sum_{x \in X}P(x) = 1$. For any set $B \subseteq X$ we define $P(B) = \sum_{x \in B}P(x)$. The \emph{support} of $P$ is the set $\supp(P) := \set{x \in X\mid P(x) > 0}$. We write $\Distributions(X)$ for the set of all probability distributions on a set $X$ and $\DistributionsFinite(X)$ for the subset containing the distributions on $X$ with finite support.

\section{Motivation}\label{sec:motivation}
The transition system in \Cref{fig:probsys2} (taken from \cite{vBW06})
is a purely probabilistic transition system (\textlsc{PTS})
\new{without labels} with state space  $X=\set{u,x,y,z}$ and an arbitrary $\epsilon \in \ ]0,1/2[$. An intuitive understanding of such a system is that in each state the system chooses a transition (indicated by the arrows) to another state using the probabilistic information which is given by the numbers on the arrows.

\begin{wrapfigure}{r}{3.6cm}
	\centering
	\begin{tikzpicture}[->,auto,node distance=1cm,semithick,>=latex]
	\tikzstyle{every state}=[fill=black!10,text=black,inner sep=0.5pt, minimum size=15pt]
	\tikzstyle{every node}=[text=black]
	% nodes
	\node[state] (x) {$x$};
	\node (center) [below=0.25cm of x] {};
	\node[state] (u) [left=0.5cm of center]{$u$};
	\node[state,accepting] (z) [right=0.5cm of center] {$z$};
	\node[state] (y) [below=0.25 cm of center] {$y$};
	%
	% paths
	\path (x) edge node[above left] {$\frac{1}{2}-\epsilon$} (u);
	\path (x) edge node[above right] {$\frac{1}{2}+\epsilon$} (z);
	\path (y) edge node[below left] {$\frac{1}{2}$} (u);
	\path (y) edge node[below right] {$\frac{1}{2}$} (z);
	\path (u) edge[loop left] node[left] {$1$} (u);
	\end{tikzpicture}
	\caption{A \textlsc{PTS}}
	\label{fig:probsys2}
\end{wrapfigure}

The state $z$ on the right hand side of \Cref{fig:probsys2} is a \emph{final state} so the system terminates with probability one (indicated by the double circle) when reaching that state. Contrary to that, state $u$ on the left hand side can be interpreted as a \emph{fail state} which -- once reached -- can never be left and the system loops indefinitely in $u$. Thus the behavior of these states is entirely different. 

Comparing the behavior of $x$ and $y$ is a bit more complicated -- they both have probabilistic transitions to $u$ and $z$ but in state $x$ there is always a bias towards the final state $z$ which is controlled only by the value of $\epsilon$. Due to this $x$ and $y$ are certainly not behaviorally equivalent but similar in the sense that their behavioral distance is $\epsilon$.

Let us now analyze how we can formally draw that conclusion. First
note that we can define a distance between the states $u$ and $z$
based solely on the fact that $z$ is final while $u$ is not. Though we
do not yet give it an explicit numerical value, we consider them to be
\emph{maximally apart}. Then, in order to compare $x$ and $y$ we need
to compare their transitions which (in this example) are probability
distributions\footnote{They are actually probability distributions on
  $X$ with support $\set{u,z}$. We will discuss this issue later.}
$P_x, P_y \colon \set{u,z} \to [0,1]$. \new{In particular
  $P_x(u) = \frac{1}{2}-\epsilon$, $P_x(z) = \frac{1}{2}+\epsilon$,
  $P_y(u) = \frac{1}{2}$, $P_y(z) = \frac{1}{2}$,
  $P_x(x) = P_x(y) = P_y(x) = P_y(y) = 0$.} Thus the underlying task is to
\emph{define a distance between probability distributions based on a
  distance on their common domain of definition}.

Let us now tackle this question with a separate, more illustrative example. We will come back to our probabilistic transition system afterwards (in fact, we will even discuss it more thoroughly throughout the whole paper).

\subsection{Wasserstein and Kantorovich Distance}
Suppose we are given the three element set $X = \set{A,B,C}$ along with a distance function (here: a metric) $d\colon X \times X \to \preal$ where
\begin{align*}
	d(A,A) = d(B,B) = d(C,C) = 0, \qquad &d(A,B) = d(B,A) = 3,\\
	d(A,C) = d(C,A) = 5, \qquad\text{ and} \qquad & d(B,C) = d(C,B) = 4\,.
\end{align*} 
Based on this function we now want to define a distance on probability distributions on $X$, i.e., a function $\LiftedMetric{\Distributions}{d}\colon \Distributions X \times \Distributions X \to [0,1]$, so we need to define its value $\LiftedMetric{\Distributions}{d}(P,Q)$ for all probability distributions $P,Q \colon X \to [0,1]$. As a concrete example, let us take the distributions 
\begin{align*}
	P(A) &= 0.7, & P(B) &= 0.1, & P(C) &= 0.2,\\
	Q(A) &= 0.2, & Q(B) &= 0.3, \text{ and} & Q(C) &= 0.5\,.
\end{align*}
In order to define their distance, we can interpret the three elements $A$, $B$, $C$ as places where a certain product is produced and consumed (imagine the places are coffee roasting plants, each with an adjacent café where one can taste the coffee). For the sake of simplicity, we assume that the total amount of production equals the amount of consumption\footnote{If this was not the case we could introduce a dummy place representing either a warehouse (overproduction) or the amount of the product that needs to be bought on the market (underproduction).}. The places are geographically distributed across an area and the distance function just describes their physical distance. Moreover, the above distributions model the supply ($P$) and demand ($Q$) of the product in \new{proportion} of the total supply or demand. We have illustrated this situation graphically in \Cref{fig:lifting}. The numbers on the edges indicate the distance between the places $A$, $B$, $C$ whereas the numbers on the nodes indicate supply $P$ (upper value) and demand $Q$ (lower value).

\begin{wrapfigure}{r}{4.5cm}
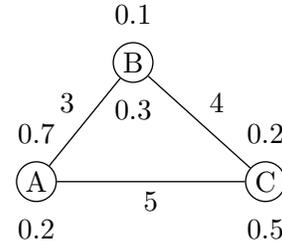

	\centering
	\begin{diagram}
		\tikzstyle{every state}=[draw=black,fill=white,text=black,inner sep=0.5pt, minimum size=15pt]
		\node[state] (A) {A};
		\node[state, above right=1.2cm and 0.9cm of A] (B) {B};
		\node[state, right=2.5cm of A] (C) {C};
		\path[-] (A) edge node[above left] {$3$} (B);
		\path[-] (A) edge node[below] {$5$} (C);
		\path[-] (B) edge node[above right] {$4$} (C);
		
		\node[text=black, above=.1cm of A] (sa) {$0.7$};
		\node[text=black, above=.1cm of B] (sb) {$0.1$};
		\node[text=black, above=.1cm of C] (sc) {$0.2$};
		
		\node[text=black, below=.1cm of A] (sa) {$0.2$};
		\node[text=black, below=.1cm of B] (sb) {$0.3$};
		\node[text=black, below=.1cm of C] (sc) {$0.5$};
	\end{diagram}
	\caption{Lifting example}
	\label{fig:lifting}
\end{wrapfigure}

The interpretation given above allows us to find two economically motivated views of defining a distance between $P$ and $Q$ based on the notion of \emph{transportation} which is studied extensively in \emph{transportation theory} \cite{Vil09}. The leading idea is that the product needs to be transported so as to avoid excess supply and meet all demands. As an owner of the three facilities we have two choices: do the transport on our own or find a logistics firm which will do it for us.

If we are organizing the transport ourselves, transporting one unit of our product from one place to another costs one monetary unit per unit of distance traveled. As an example, transporting ten units from $A$ to $B$ will cost $10 \cdot 3 = 30$ monetary units. Formally, we will have to define a function $t \colon X \times X \to [0,1]$ where $t(x,y)$ describes (in \%) the amount of goods to be transported from place $x$ to $y$ such that
\begin{itemize}
	\item supplies are used up: for all $x \in X$ we must have $\sum_{y \in X} t(x,y) = P(x) $, and
	\item demand is satisfied: for all $y \in X$ we must have $\sum_{x \in X} t(x,y) = Q(y)$.
\end{itemize}
In probability theory, such a function is as a joint probability distribution on $X \times X$ with marginals $P$ and $Q$ and is therefore called a \emph{coupling} of $P$ and $Q$. In our economic perspective we will just call it a \emph{transportation plan} for supply $P$ and demand $Q$ and write $T(P,Q)$ for the set of all such transportation plans.

If $M \in \R_+$ denotes the total supply (= total demand), the total transportation cost for any transportation plan $t \in T(P,Q)$ is given by the function $c_d$ which is parametrized by the distance function $d$ and defined as
\begin{align*}
	c_d(t) := \sum_{x,y \in X} \big(M \cdot t(x,y)\big) \cdot d(x,y) = M \cdot \sum_{x,y \in X} t(x,y) \cdot d(x,y)\,.
\end{align*}
Since we want to maximize our profits, we will see to it that the
total transportation cost is minimized, i.e., we will look for a
transportation plan $t^* \in T(P,Q)$ minimizing the value of $c_d$ (it
can be shown that such a $t^*$ always exists, by compactness). Using this, we can now define the distance of $P$ and $Q$ to be the relative costs $c_d(t^*)/M$ of such an optimal transportation plan, i.e., 
\begin{align*}
	\Wasserstein{\Distributions}{d}(P,Q) := \min \set{\sum_{x,y \in X} t(x,y) \cdot d(x,y)\ \bigg|\ t \in T(P,Q)} = c_d(t^*)/M\,.
\end{align*}
This distance between probability distributions is called the \emph{Wasserstein distance} and it can be shown that this distance is a pseudometric if $d$ is a pseudometric (and we will recover this result later). In our concrete example of \Cref{fig:lifting}, the best solution is apparently to first use the local production (zero costs) at each facility and then transport the remaining excess supply of $50\%$ in $A$ to the remaining demands in $B$ ($20\%$) and in $C$ ($30\%$). Thus we obtain
\begin{align*}
	\Wasserstein{\Distributions}{d}(P,Q) = 0.2 \cdot d(A,B) + 0.3 \cdot d(A,C) = 0.2 \cdot 3 + 0.3 \cdot 5 = 0.6+1.5 = 2.1
\end{align*}
yielding an optimal (absolute) transportation cost of $2.1 \cdot M$ monetary units.

If we decide to let a logistics firm do the transportation instead of doing it on our own, we assume that for each place it sets up a price for which it will buy a unit of our product (at a place with overproduction) or sell it (at a place with excess demand). Formally it will do so by giving us a price function $f \colon X \to \R_+$. We will only accept  \emph{competitive price functions} which satisfy $|f(x) - f(y)| \leq d(x,y)$ for all places $x,y \in X$. This amounts to saying that if we sell one unit of our product at place $x$ to the logistics firm and buy one at $y$ it does not cost more than transporting it ourselves from $x$ to $y$. If $d$ is a pseudometric, we will later call this requirement \emph{nonexpansiveness} of the function $f$. Here, we will denote the set of all these functions by $C(d)$.

The logistics firm is interested in its own profits which are given by the function $g_d$ which is again parametrized by $d$ and defined as
\begin{align*}
	g_d(f) := \sum_{x \in X} f(x) \cdot \Big(\big(M \cdot Q(x)\big)-\big(M \cdot P(x)\big)\Big) = M \cdot  \sum_{x \in X} f(x) \cdot \big(Q(x)-P(x)\big)
\end{align*}
for all competitive price functions $f \in C(d)$. If in this formula the value $Q(x)-P(x)$ is greater than $0$, there is underproduction so the logistics firm can sell goods whereas if $Q(x) - P(x) < 0$ it will have to buy them. Naturally, the logistics firm will want to maximize its profits so it will look for a competitive price function $f^* \in C(d)$ maximizing the value of $g_d$.  Based on this we can now define another distance between $P$ and $Q$ to be the relative profit $g_d(f^*) / M$, i.e.,
\begin{align*}
	\Kantorovich{\Distributions}{d}(P,Q) := \max\set{\sum_{x \in X} f(x) \cdot \big(Q(x) - P(x)\big) \,\Big|\,f \in C(d)} = g_d(f^*)/M\,.
\end{align*}
This distance is known as the \emph{Kantorovich distance} and is also a pseudometric if $d$ is. One can show that for our example it will be best if we give our product to the logistics firm for free\footnote{Apparently, this is only reasonable in presence of a contract that prohibits the logistics firm to use our product or sell it to anyone else.} in $A$, i.e., the logistic firm defines the price $f^*(A) =0$. Moreover, we need to buy it back at $B$ for three monetary units ($f^*(B) = 3$) and for five monetary units at $C$ ($f^*(C) = 5$). This yields as distance
\begin{align*}
	\Kantorovich{\Distributions}{d}(P,Q) &= \sum_{x \in X}f^*(x) \cdot \big(Q(x) - P(x)\big) \\\
	&= 0 \cdot (0.2-0.7) + 3 \cdot (0.3-0.1) + 5 \cdot(0.5-0.2) = 2.1
\end{align*}
which is exactly the same as the one obtained earlier. In fact, one
can prove that $\Kantorovich{\Distributions}{d}(P,Q) \leq
\Wasserstein{\Distributions}{d}(P,Q)$ so whenever we have a
transportation plan $t^*$ and a competitive price function $f^*$ so
that $c_d(t^*) = g_d(f^*)$ we know that both are optimal yielding
$\Kantorovich{\Distributions}{d}(P,Q) =
\Wasserstein{\Distributions}{d}(P,Q) = c_d(t^*)$. As final remark we
note that if $X$ is a finite set such a pair $(t^*,f^*)$ will always
exist and can be found e.g. using the simplex algorithm from linear
programming. 

Note that the Kantorovich-Rubinstein duality \cite{Vil09} states that
in the probabilistic setting the Wasserstein and the Kantorovich
distance coincide. However, this is not necessarily the case when we
lift other functors.

\subsection{Behavioral Distance as Fixed Point}
Now that we have finished our little excursion to transportation theory, let us come back to the original example where we wanted to define a distance between the two states $x$ and $y$ of the   probabilistic transition system in \Cref{fig:probsys2}. Since we consider $u$ and $z$ to be maximally apart, we could formally set $d(u,z) = d(z,u) = 1$ and $d(u,u) = d(z,z) = 0$ so we obtain as distance function $d \colon \set{u,z} \times \set{u,z} \to [0,1]$ the discrete $1$-bounded metric on the set $\set{u,z}$. Using this, we could then define the distance of $x$ and $y$ to be the distance of their transition distributions $P_x, P_y \colon \set{u,z} \to [0,1]$ yielding indeed a distance $d'(x,y) = \Wasserstein{\Distributions}{d}(P_x, P_y) = \Kantorovich{\Distributions}{d}(P_x,P_y) = \epsilon$ as claimed in the beginning.

However, the remaining question we need to answer is how the above
procedure gives rise to a proper behavioral distance in the sense that
we obtain a sound and complete definition of a distance function on
\emph{the whole set} $X$, i.e., a pseudometric
$d\colon X \times X \to [0,1]$. \new{Notice that the distance between
  distributions depends on the distances in $X$, and vice-versa.
  Thus, we are led to a fixed-point definition, in particular in the
  presence of cycles in the transition system.}  We just need to
observe that the definition of $d'(x,y)$ immediately yields the
following recursive formula
$d(x_1,x_2) =
\LiftedMetric{\Distributions}{d}\big(P_{x_1},P_{x_2}\big)$ for
\emph{all} $x_1,x_2 \in X$ where $\LiftedMetric{\Distributions}{d}$ is
one of the equivalent distances (Wasserstein or Kantorovich) defined
above. A known approach for probabilistic systems as the one above is
to define its behavioral distance to be a fixed point
$d^*\colon X \times X \to [0,1]$ of the above equation
\cite{vBW06,vBSW08}. It is not difficult to see that due to the
special structure of the above system, one obtains
$d^*(u,z) = d^*(z,u) = 1$ and $d^*(x,y) = d'(x,y) = \epsilon$ which
finally validates our initial claim that the distance of $x$ and $y$
is indeed $\epsilon$.

\section{Pseudometric Spaces}
\label{sec:pseudometrics}
Contrary to the usual definitions, where distances assume values in the half open real interval $[0,\infty[$, our distances assume values in a closed interval $[0,\top]$, where $\top \in\,]0,\infty]$ is a fixed maximal element (for our examples we will use $\top = 1$ or $\top=\infty$).

\begin{definition}[Pseudometric, Pseudometric Space]
	\label{def:pseudometric}
	\index{pseudometric}
	\index{pseudometric space}
	\index{metric}
	\index{metric space}
	Let $\top \in\,]0, \infty]$ be a fixed maximal distance and $X$ be a set. We call a function $d \colon X \times X \to \reals$ a \emph{$\top$-pseudometric} on $X$ (or a \emph{pseudometric} if $\top$ is clear from the context) if it satisfies $d(x,x) = 0$ (\emph{reflexivity}), \index{reflexive}\index{reflexivity} $d(x,y) = d(y,x)$ (\emph{symmetry}), \index{symmetric}\index{symmetry} and $d(x,z) \le d(x,y)+d(y,z)$ (\emph{triangle inequality})\index{triangle inequality} for all $x,y,z\in X$. If additionally $d(x,y)=0$ implies $x=y$, $d$ is called a \emph{$\top$-metric} (or a \emph{metric}). A \emph{(pseudo)metric space} is a pair $(X,d)$ where $X$ is a set and $d$ is a (pseudo)metric on $X$.
\end{definition}

A trivial example of a pseudometric is the constant $0$-function on any set whereas a simple example of a metric is the so-called \emph{discrete metric} which can be defined on any set $X$ as $d(x,x) = 0$ for all $x \in X$ and $d(x,y) = \top$ for all $x,y \in X$ with $x \not = y$.

The set of (pseudo)metrics over a fixed set $X$ is a complete lattice (since $\reals$ is) with respect to the pointwise order, i.e., for $d_1,d_2 \colon X \times X \to [0,\top]$ we define $d_1 \leq d_2$ if and only if $d_1(x, x') \leq d_2(x,x')$ holds for all $x,x' \in X$.

\begin{lemma}[Lattice of Pseudometrics]
  \label{le:lattice-pseudo}
  Let $X$ be a set. Then the set of pseudometrics over $X$, i.e.,
  $D_X = \{ d \mid d \colon X \times X \to \reals\ \land\ d\
  \text{pseudometric}\}$, endowed with the pointwise order, is a
  complete lattice. The join of a set of pseudometrics
  $D \subseteq D_X$ is
  $(\sup D)(x,y) = \sup \{ d(x,y) \mid d \in D \}$ for all $x, y \in X$. The meet of $D$ is $\inf D = \sup \{ d \mid d \in D_X \land\ \forall d' \in D.\ d \leq d'\}$.
\end{lemma}

\begin{proof}
  Let $D \subseteq D_X$. We first show that
  $d' \colon X \times X \to \reals$ defined by
  $d'(x,y) = \sup \{ d(x,y) \mid d \in D \}$ for all $x, y \in X$ is a
  pseudometric. Reflexivity and symmetry follows immediately from the fact that all $d \in D$ are pseudometrics. Concerning the triangle inequality, for all $x,y,z\in X$:
  \begin{align*}
    d'(x,y) + d'(y,z) &= \sup_{d\in D} d(x,y) + \sup_{d\in D} d (y,z)\\
                      &\geq\sup_{d\in D} \big(d(x,y) + d (y,z)\big) \geq \sup_{d\in D} d (x,z) = d'(x,z)\,.
  \end{align*}
  Since $d'$ is a pseudometric, it immediately follows that
  $\sup D = d$, as desired.
  
  The assertion about the meet is a completely general fact in
  complete lattices.
\end{proof}

Observe that while the join of a set of pseudometrics is given by the pointwise supremum, in general, the meet can be smaller than the pointwise infimum. Hereafter,  joins and meets of sets of pseudometrics will be implicitly taken in the corresponding lattice.

We will see later (\Cref{prop:PMetcomplete}) that the completeness of the lattice of  pseudometrics ensures that the category whose objects are pseudometric spaces is complete and cocomplete.

\subsection{Calculating with (Extended) Real Numbers}
By $d_e\colon \reals^2\to \reals$ we denote the ordinary Euclidean metric\index{Euclidean metric} on $\reals$, i.e., $d_e(x,y) = |x-y|$ for $x,y\in\reals\setminus\{\infty\}$, and -- where appropriate -- $d_e(x,\infty) = \infty$ if $x\neq \infty$ and $d_e(\infty,\infty) = 0$. Addition is defined in the usual way, in particular $x + \infty = \infty$ for $x\in\prealinf$. 

In the following lemma we rephrase the well-known fact that for $a,b,c \in [0,\infty)$ we have $|a-b| \leq c \iff a-b \leq c\, \land\, b-a \leq c$ to include the cases where $a,b,c$ might be $\infty$. The proof is a simple case distinction and hence omitted.

\begin{lemma}
	\label{lem:sum-vs-dist}
	For $a,b,c\in\prealinf$ we have $d_e(a,b)\leq c$ if and only if $(a \leq b+c)$ and $(b \leq a+c)$.
\end{lemma}

We continue with another simple calculation involving this extended Euclidean distance which will turn out to be useful. Again, we omit the straightforward proof.

\begin{lemma}
	\label{lem:max-sum}
	For a finite set $A$ and functions $f,g\colon A \to [0,\infty]$ we have
	\begin{enumerate}
		\item $d_e\big(\max_{a \in A}f(a),\max_{a \in A}g(a)\big) \leq \max_{a \in A} d_e\big(f(a), g(a)\big)$, and
		\item $d_e\left(\sum_{a \in A}f(a),\sum_{a \in A}g(a)\right) \leq \sum_{a \in A} d_e\big(f(a), g(a)\big)$.
	\end{enumerate}
\end{lemma}

\subsection{Pseudometrics Categorically}
Having established these two intermediary results, we recall that we want to work in a category whose objects are pseudometric spaces. In order to do so we need to define the arrows between them. While there are other, topologically motivated possibilities (taking the \emph{continuous} or \emph{uniformly continuous} functions with respect to the pseudometric topology), we just require that our functions do not increase distances.

\begin{definition}[Nonexpansive Function, Isometry]
	\label{def:nonexpansive-fct}
	\index{nonexpansive function}
	\index{function!nonexpansive}
	\index{isometry}
	\index{function!isometry}
	Let $\top \in \ ]0,\infty]$ be an extended real number and $(X,d_X)$ and $(Y,d_Y)$ be $\top$-pseudometric spaces. We call a function $f\colon X\to Y$ \emph{nonexpansive} if $d_Y \circ (f \times f) \le d_X$. In this case we write $f \colon (X, d_X) \nonexpansiveTo (Y, d_Y)$. If equality holds $f$ is called an \emph{isometry}. 
      \end{definition}

\new{Observe that that isometries are not required to be surjective. Moreover, since we work with pseudometrics, where distinct elements can have distance $0$, they are not necessarily injective.}
Note also that in this definition we have used a category theoretic mind-set and written the inequality in an ``element-free'' version as it will be easier to use in some of the subsequent results. Of course, this inequality is equivalent to requiring $d_Y\big(f(x),f(x')\big) \leq d_X(x,x')$ for all $x,x' \in X$. Simple examples of nonexpansive functions -- even isometries -- are of course the identity functions on a pseudometric space. 

Apparently, if these functions shall be the arrows of a category, we will have to check that nonexpansiveness is preserved by function composition.

\begin{lemma}[Composition of Nonexpansive Functions]
	Let $\top \in\ ]0, \infty]$, and $(X,d_X)$, $(Y,d_Y)$ and $(Z, d_Z)$ be $\top$-pseudometric spaces. If $f \colon (X, d_X) \nonexpansiveTo (Y, d_Y)$ and $g \colon (Y, d_Y) \nonexpansiveTo (Z, d_Z)$ are nonexpansive, then so is $g \circ f\colon (X, d_X) \nonexpansiveTo (Z, d_Z)$.
\end{lemma}
\begin{proof}
	Using nonexpansiveness of $g$ and $f$ we immediately conclude that $d_Z \circ \big((g \circ f) \times (g \circ f)\big) = d_Z \circ (g \times g) \circ (f \times f) \leq d_Y \circ (f \times f) \leq d_X$ which is the desired nonexpansiveness of $g \circ f$.
\end{proof}

With this result at hand we now give our category a name. Please note that the definition below actually defines a whole family of categories, parametrized by the chosen maximal element. Despite this, we will just speak of \emph{the} category of pseudometric spaces and keep in mind that there are (uncountably) many with the same properties. 

\begin{definition}[Category of Pseudometric Spaces]
	\label{def:category-metric}
	\index{category of pseudometric spaces}
	Let $\top \in\ ]0,\infty]$ be a fixed maximal element. The category $\PMet$ has as objects all pseudometric spaces whose pseudometrics have codomain $[0,\top]$. The arrows are the nonexpansive functions between these spaces. The identities are the (isometric) identity functions and composition of arrows is function composition. 
\end{definition}

This category is complete and cocomplete which in particular implies that it has products and coproducts. We will later see that the respective product and coproduct pseudometrics also arise as special instances of our lifting framework (see \Cref{lem:binary-prod-PMet,lem:coproduct-pseudometric}). 

\begin{theorem}\label{prop:PMetcomplete}
	$\PMet$ is bicomplete, i.e., it is complete and cocomplete.
\end{theorem}
\begin{proof}
	Let $U\colon \PMet \to \Set$ be the forgetful functor which maps every pseudometric space to its underlying set and each nonexpansive function to the underlying function. Moreover, let $D\colon I \to \PMet$ be a small diagram, and define $(X_i, d_i) := D(i)$ for each object $i \in I$. Obviously $UD\colon I \to \Set$ is also a small diagram. We show completeness and cocompleteness separately.
	\medskip\\
	\noindent\emph{Completeness}: Let $(f_i\colon X \to X_i)_{i \in I}$ be the limit cone to $UD$.
	
	Observe that for any $i \in I$, $d_i~ \circ (f_i \times f_i)\colon X^2 \to \reals$ is a pseudometric and define the pseudometric $d:=\sup_{i\in I} d_i~ \circ (f_i \times f_i)\colon X^2 \to \reals$ as provided by \Cref{le:lattice-pseudo}.

	With this pseudometric all $f_j$ are nonexpansive functions $(X, d) \nonexpansiveTo (X_j,d_j)$ because for all $j \in I$ and all $x,y \in X$ we have $d_j\big(f_j(x), f_j(y)\big) \leq \sup_{i \in I} d_i\big(f_i(x), f_i(y)\big) = d(x,y)$.
	
	Moreover, if $\big(f_i'\colon (X',d') \nonexpansiveTo (X_i,d_i)\big)_{i \in I}$ is a cone to $D$, $(f_i'\colon X' \to X_i)_{i \in I}$ is a cone to $UD$ and hence there is a unique function $g\colon X' \to X$ in $\Set$ satisfying $f_i \circ g = f_i'$ for all $i \in I$. We finish our proof by showing that this is a nonexpansive function $(X',d') \nonexpansiveTo (X,d)$. By nonexpansiveness of the $f_i'$ we have for all $i \in I$ and all $x,y \in X'$ that $d_i(f_i'(x),f_i'(y)) \leq d'(x,y)$ and thus also 
	\begin{align*}
	d\big(g(x),g(y)\big) &= \sup_{i \in I} d_i\Big(f_i\big(g(x)\big), f_i\big(g(y)\big)\Big) = \sup_{i \in I} d_i\big(f_i'(x), f_i'(y) \big)\leq \sup_{i \in I}d'(x,y) = d'(x,y)\,.
	\end{align*}
	We conclude that $\big(f_i\colon (X,d) \nonexpansiveTo (X_i,d_i)\big)_{i \in I}$ is a limit cone to $D$.
	\medskip\\\noindent\emph{Cocompleteness:} Let $(f_i\colon X_i \to X)_{i \in I}$ be the colimit co-cone from $UD$ and $M_X$ be the set of all pseudometrics $d_X\colon X^2 \to \reals$ on $X$ such that the $f_i$ are nonexpansive functions $(X_i,d_i) \nonexpansiveTo (X,d_X)$. We define the pseudometric $d:=\sup_{d_X \in M_X} d_X$, as given by Lemma~\ref{le:lattice-pseudo}. With this pseudometric all $f_j$ are nonexpansive functions $(X_j, d_j) \nonexpansiveTo (X,d)$ because for all $j \in I$ and all $x,y \in X_j$ we have $d\big(f_j(x), f_j(y)\big) = \sup_{d_X \in M_X} d_X\big(f_j(x), f_j(y)\big) \leq \sup_{d_X \in M_X} d_j(x,y) = d_j(x,y)$.
	
	Moreover, if $\big(f_i'\colon (X_i,d_i) \nonexpansiveTo (X',d')\big)_{i \in I}$ is a co-cone from $D$, $(f_i'\colon X_i \to X')_{i \in I}$ is a co-cone from $UD$ and hence there is a unique function $g\colon X \to X'$ in $\Set$ satisfying $g \circ f_i = f_i'$ for all $i \in I$. We finish our proof by showing that this is a nonexpansive function $(X,d) \nonexpansiveTo (X',d')$. Let $d_g:=d'\circ (g\times g)\colon X^2 \to \reals$, then it is easy to see that $d_g$ is a pseudometric on $X$. Moreover, for all $i \in I$ and all $x,y \in X_i$ we have 
	\begin{align*}
	d_g\big(f_i(x), f_i(y)\big) = d'\Big(g\big(f_i(x)\big), g\big(f_i(y)\big)\Big) = d'\big(f_i'(x), f_i'(y)\big) \leq d_i(x,y)
	\end{align*}
	due to nonexpansiveness of $f_i'\colon (X_i,d_i) \nonexpansiveTo (X',d')$. Thus all $f_i$ are nonexpansive if seen as functions $(X_i, d_i) \nonexpansiveTo (X,d_g)$ and we have $d_g \in M_X$. Using this we observe that for all $x,y \in X$ we have $d'\big(g(x),g(y)\big) = d_g(x,y) \leq \sup_{d_X \in M_X} d_X(x,y) = d(x,y)$ which shows that $g$ is a nonexpansive function $ (X', d') \nonexpansiveTo (X,d)$. We conclude that indeed $\big(f_i\colon (X_i,d_i) \nonexpansiveTo (X,d)\big)_{i \in I}$ is a colimit co-cone from $D$.\qedhere
\end{proof}

For our purposes it will turn out to be useful to consider the following alternative characterization of the triangle inequality using the concept of nonexpansive functions.

\begin{lemma}
	\label{lem:alt-char-triangle}
	A reflexive and symmetric function (see \Cref{def:pseudometric}) $d\colon X^2\to \reals$ satisfies the triangle inequality if and only if for all $x \in X$ the function $d(x,\_)\colon (X,d) \to (\reals,d_e)$ is nonexpansive.
\end{lemma}
\begin{proof}
	Let $d\colon X^2 \to [0,1]$ be a reflexive and symmetric function. We show both implications separately for all $x,y,z \in X$.
	\begin{itemize}
		\item Using the triangle inequality and symmetry we know that $d(x,y) \leq d(x,z) + d(y,z)$ and $d(x,z) \leq d(x,y) + d(y,z)$. With \Cref{lem:sum-vs-dist} we conclude that $d_e(d(x,y),d(x,z)) \leq d(y,z)$.

		\item Using reflexivity of $d$, the triangle inequality for $d_e$ and nonexpansiveness of $d(x,\_)$ we get $d(x,z) = d_e(d(x,x),d(x,z)) \leq d_e(d(x,x),d(x,y)) + d_e(d(x,y),d(x,z)) \leq d(x,y) + d(y,z)$.\qedhere
	\end{itemize}
\end{proof}

\section{Examples of Behavioral Distances}
\label{sec:motivating-examples}
Equipped with this basic knowledge about our pseudometrics, let us now take a look at a few examples which we will use throughout the paper to demonstrate our theory. All the claims we make in these examples will be justified by our results. Our first example are the purely probabilistic systems like the system in \Cref{fig:probsys2} in the beginning. 

\begin{example}[Probabilistic Systems and Behavioral Distance \cite{vBW06}]
	\label{ex:probabilistic-1}
	The probability distribution functor $\Distributions\colon \Set \to \Set$ maps each set $X$ to the set of all probability distributions on it, i.e., $\Distributions X = \set{p \colon X \to [0,1]\mid \sum_{x \in X}p(x) = 1}$ and each function $f \colon X \to Y$ to the function $\Distributions f \colon \Distributions X \to \Distributions Y$ with $\Distributions f(p)(y) = \sum_{x \in f^{-1}[\{y\}]}p(x)$. We consider purely probabilistic transition systems without labels as coalgebras of the form $\alpha\colon X\to \Distributions (X+\one)$. Thus $\alpha(x)(y)$, for $x,y\in X$, is the probability of a transition from a state $x$ to $y$ and $\alpha(x)(\checkmark)$ gives the probability of terminating in $x$.
	
	\name[Franck]{Van Breugel} and \name[James]{Worrell} \cite{vBW06} introduced a metric for a continuous version of these systems by considering a discount factor $c \in\,]0,1[$. Instantiating their framework to the discrete case we obtain the behavioral distance $d\colon X^2\to [0,1]$, defined as the least solution (with respect to the order $d_1 \leq d_2 \iff \forall x,y \in X.d_1(x,y) \leq d_2(x,y)$) of the equation $d(x,y) = \overline{d}\big(\alpha(x),\alpha(y)\big)$ for all $x,y\in X$. The lifted pseudometric $\overline{d}\colon (\Distributions (X+\one))^2\to [0,1]$ is defined in two steps: 
	\begin{itemize}
		\item First, $\widehat{d}\colon (X+\one)^2\to [0,1]$ is defined as $\widehat{d}(x,y) = c\cdot d(x,y)$ if $x,y\in X$, $\widehat{d}(\checkmark,\checkmark) = 0$ and $\widehat{d}(x,y) = 1$ otherwise. 
		\item Then, for all $P_1,P_2\in \Distributions
                  (X+\one)$, $\overline{d}(P_1,P_2)$ is defined as the
                  supremum of all values $\sum_{x \in X+\one} f(x)\cdot (P_1(x) - P_2(x))$, with $f \colon (X+\one,\widehat{d}) \nonexpansiveTo ([0,1],d_e)$ being an arbitrary nonexpansive function.
	\end{itemize}
	
	\noindent Our concrete example from \Cref{fig:probsys2} is an instance of such a system and if we employ the aforementioned approach the behavioral distance of $u$ and $z$ is $d(u,z) = 1$ and hence $d(x,y) = c\cdot \epsilon$. We will see in \Cref{ex:probabilistic-2} that this example can be captured by our framework. Moreover, we will also see that it is possible to set $c = 1$ resulting in $d(x,y) = \epsilon$.
\end{example}

It is easy to see that also the state space of a deterministic automaton can be equipped with a pseudometric which arises as a solution of a fixed point equation.

\begin{example}[Bisimilarity Pseudometric for Deterministic Automata]
	\label{exa:bisim-pseudometric-da}
	We consider deterministic automata as coalgebras $\langle o, s\rangle \colon X \to \two \times X^A$ in $\Set$. For each state $x \in X$ the output value $o(x)$ determines whether $x$ is final ($o(x) = 1$) or not ($o(x)=0$) and the successor function $s(x)\colon A \to X$ determines, for each label $a \in A$, the unique $a$-successor $s(x)(a) \in X$. Given a pseudometric $d \colon X^2 \to [0,\new{1}]$, we obtain a new pseudometric $\LiftedMetric{F}{d}$ on $\two \times X^A$ by defining, for every $(o_1,s_1), (o_2,s_2) \in \two \times X^A$,
	\begin{align*}
	\LiftedMetric{F}{d}\big((o_1,s_1), (o_2, s_2)\big) = \max\set{d_\two(o_1,o_2), c \cdot \max_{a \in A}d\big(s_1(a), s_2(a)\big)}\,,
	\end{align*}
	where $c \in \,]0,1]$ a discount factor and $d_\two$ is the discrete $1$-bounded metric on $\two$. Using our coalgebra $\langle o, s \rangle \colon X \to \two \times X^A$ we get a fixed point equation on the complete lattice of pseudometrics on $X$ by requiring, for all $x_1, x_2 \in X$, 
	\begin{align*}
	d(x_1,x_2) = \max\set{d_\two\big(o(x_1),o(x_2)\big), c \cdot \max_{a \in A}d\big(s(x_1)(a), s(x_2)(a)\big)}\,.
	\end{align*}
	If we take \new{the least} fixed point of this equation, the
        distance of two states $x_1$ and $x_2$ will be $\new{1}$ if
        one state is final and the other is not. Otherwise their
        distance is the $c$-discounted maximal distance of their
        successors. \new{Note that for $c=1$ we obtain a discrete
          metric, where the distance of two states is $0$ if they
          accept the same language and $1$ otherwise, hence discount
          factors $c<1$ yield more interesting metrics.}
\end{example}

Let us finally consider so-called \emph{metric transition systems} \cite{dAFS09}. Each state of such a system comes equipped with a function which maps elements of a set of so-called \emph{propositions} to a kind of non-discrete truth value in a pseudometric space.

\begin{example}[{Branching Distance for Metric Transition Systems \cite{dAFS09}}]
	\label{ex:metric-ts-1}
	Let $\Sigma =\set{r_1,\dots,r_n}$ be a finite set of \emph{propositions}\index{proposition} where each proposition $r\in\Sigma$ is associated with a pseudometric space $(M_r,d_r)$ which is bounded, i.e.,  we must have a finite $\top \in\ ]0,\infty[$ such that $d_r \colon M_r^2 \to [0,\top]$. A \emph{valuation}\index{valuation} of $\Sigma$ is a function $u \colon \Sigma \to \cup_{r \in \Sigma}M_r$ that assigns to each $r\in\Sigma$ an element of $M_r$, i.e., we require $u(r) \in M_r$. We denote the set of all these valuations by $\mathcal{U}[\Sigma]$ and remark that it is apparently isomorphic to the set $M_1 \times \dots \times M_n$ by means of the bijective function which maps a valuation $u$ to the tuple $ \big(u(r_1),\dots,u(r_n)\big)$.
	
	A \emph{metric transition system}\index{metric transition system} \cite[Def.~6]{dAFS09} is a tuple $(S,\tau,\Sigma,[\cdot])$ with a set $S$ of states, a transition relation $\tau\subseteq S\times S$, a finite set $\Sigma$ of propositions and a function $[\cdot]\colon S \to \mathcal{U}[\Sigma]$ assigning a valuation $[s]$ to each state $s\in S$. We define $\tau(s) := \set{s'\in S\mid (s,s')\in\tau}$ and require that $\tau(s)$ is finite.
	
	The (directed) propositional distance between two valuations $u,v \in \mathcal{U}[\Sigma]$ is given by  \cite[Def.~10]{dAFS09} $\LiftedFunctor{\mathit{pd}}(u,v) = \max_{r\in\Sigma} d_r\big(u(r),v(r)\big)$. The (undirected) branching distance $d\colon S\times S\to\preal$ is defined as \cite[Def.~13]{dAFS09} the smallest fixed-point
	of the following equation, where $s,t\in S$:
	\begin{equation}\label{eq:Hausd}
	d(s,t) = \max\set{\LiftedFunctor{\mathit{pd}}([s],[t]),	\max_{s'\in\tau(s)}\min_{t'\in\tau(t)}  d(s',t'),\max_{t'\in\tau(t)}\min_{s'\in\tau(s)} d(s',t') } 
	\end{equation} 
	Note that, apart from the first argument, this coincides with the Hausdorff distance between the successors of $s$ and $t$.
	
		\begin{figure}[ht]
			\centering
			\begin{tikzpicture}[->,auto,node distance=1cm,semithick,>=latex]
			\tikzstyle{every state}=[fill=black!10,text=black,inner sep=0.5pt, minimum size=15pt]
			\tikzstyle{every node}=[text=black]
			% nodes system 1
			\node[state] (x1) {$x_1$};
			\node (p1) [right=0mm of x1] {$0$};
			\node[state] (x2) [below left=of x1] {$x_2$};
			\node (p2) [right=0mm of x2] {$0.4$};
			\node[state] (x3) [below right=of x1] {$x_3$};
			\node (p3) [left=0mm of x3] {$0.7$};
			% nodes system 2
			\node[state] (y1) [right=3.5cm of x1] {$y_1$};
			\node (q1) [right=0mm of y1] {$0$};
			\node[state] (y2) [below left=of y1] {$y_2$};
			\node (q2) [right=0mm of y2] {$0.5$};
			\node[state] (y3) [below right=of y1] {$y_3$};
			\node (q3) [left=0mm of y3] {$1$};
			% paths system 1
			\path (x1) edge (x2);
			\path (x1) edge (x3);
			\path (x2) edge[loop left] (x2);
			\path (x3) edge[loop right] (x3);
			% paths system 2
			\path (y1) edge (y2);
			\path (y1) edge (y3);
			\path (y2) edge[loop left] (y2);
			\path (y3) edge[loop right] (y3);
			\end{tikzpicture}
			
			\caption{A metric transition system}
			\label{fig:example-ts}
		\end{figure}
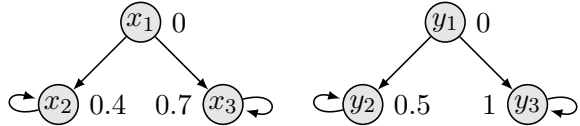
	We consider the concrete example system in \Cref{fig:example-ts}  (\cite[Fig.~1]{dAFS09}) with a single proposition $r\in \Sigma$, where $M_r = [0,1]$ is equipped with the Euclidean distance $d_e$. Since the states $x_2,x_3, y_2, y_3$ only have themselves as successors, computing their distance according to~\eqref{eq:Hausd} is easy: we just have to take the propositional distances of the valuations. This results in $d(x_2,y_2) = |0.4-0.5| = 0.1$, $d(x_2,y_3) = 0.6$, $d(x_3,y_2) = 0.2$, $d(x_3,y_3) = 0.3$. 
	
	Moreover, $\overline{\mathit{pd}}([x_1],[y_1]) = 0$ and thus $d(x_1,y_1)$ equals the Hausdorff distance of the reals associated with the sets of successors which is~$0.3$ (the maximal distance of any successor to the closest successor in the other set of successors, here: the distance from~$y_3$ to~$x_3$).

	In order to model such transition systems as coalgebras we consider the product multifunctor $P\colon \Set^n \to \Set$ where $P(X_1,\dots,X_n) = X_1 \times \dots \times X_n$. Then coalgebras are of the form $c\colon S\to P(M_{r_1},\dots,M_{r_n})\times \PowersetFinite(S)$, where $\PowersetFinite$ is the finite powerset functor and $c(s) = \big([s][r_1], \dots, [s][r_n],\tau(s)\big)$. As we will see later in \Cref{ex:metric-ts-2}, the right-hand side of \eqref{eq:Hausd} can be seen as lifting of a pseudometric $d$ on $X$ to a pseudometric on $P(M_{r_1},\dots,M_{r_n})\times \PowersetFinite(X)$.
\end{example}

We will later see that in all the examples above, we obtain a coalgebraic \emph{bisimilarity pseudometric}: For any coalgebra $c\colon X \to FX$ let us denote the respective least fixed point by $\bd_c \colon X^2 \to [0,\top]$. If a final $F$-coalgebra $z \colon Z \to FZ$ exists and some additional conditions hold (which is the case for all our examples) we have
\begin{align*}
	\bd_c(x,y) = 0 \quad \iff \quad \final{x}_c = \final{y}_c 
\end{align*}
for all $x,y \in X$ where $\final{\cdot}_c \colon X \to Z$ is the unique map into the final coalgebra.

\section[Lifting Functors to Pseudometric Spaces]{Lifting Functors to Pseudometric Spaces\sectionmark{Lifting Functors}}
\sectionmark{Lifting Functors}
\label{sec:lifting}
Generalizing from our examples, we now establish a general framework for deriving behavioral distances. The crucial step is to find, for an endofunctor $F$ on $\Set$, a way to transform a pseudometric on $X$ to a pseudometric on $FX$. This induces a lifting of the functor $F$ in the following sense.\medskip\newline
\begin{minipage}{11.7cm}
\begin{definition}[Lifting to Pseudometric Spaces]
	\label{def:liftingPMet}
	\index{lifting}\index{lifted functor}
	Let $U\colon \PMet \to \Set$ be the forgetful functor which maps every pseudometric space to its underlying set. A functor $\LiftedFunctor{F}\colon\PMet\to\PMet$ is called a \emph{lifting} of a functor $F\colon \Set \to \Set$ if the diagram on the right commutes.
		In this case, for any pseudometric space $(X,d)$, we denote by $\LiftedMetric{F}{d}$ the pseudometric on $FX$ which we obtain by applying $\LiftedFunctor{F}$ to $(X,d)$.
\end{definition}
\end{minipage}\hfill\begin{minipage}{3.4cm}\begin{diagram}
		\matrix[matrix of math nodes, column sep=.5cm, row sep=1.2cm] (m){
			\PMet & \PMet\\
			\Set & \Set\\
		};

		\draw (m-1-1) edge node[above] {$\LiftedFunctor{F}$} (m-1-2);
		\draw (m-1-1) edge node[left] {$U$} (m-2-1);
		\draw (m-2-1) edge node[below] {$F$} (m-2-2);
		\draw (m-1-2) edge node[right] {$U$} (m-2-2);
	\end{diagram}\end{minipage}\medskip
	
Such a lifting is always monotone on pseudometrics in the following sense.

\begin{theorem}[Monotonicity of Lifting]
	\label{prop:monotone}
	Let $\LiftedFunctor{F}\colon \PMet \to \PMet$ be a lifting of $F\colon \Set \to \Set$ and $d_1,d_2 \colon X \times X \to \reals$ be pseudometrics on $X$. Then $d_1 \leq d_2$ implies $\LiftedMetric{F}{d_1} \leq \LiftedMetric{F}{d_2}$. 
\end{theorem}
\begin{proof}
  \new{
  Observe first that for two pseudometrics $d_1$ and $d_2$ on the same
  set $X$, we have $d_1 \leq d_2$ iff there is some
  $f\colon (X,d_2)\to (X, d_1)$ in $\PMet$ such that $Uf = \id_X$ in
  $Set$.

  Using the above fact, we can prove our statement.  Let
  $f\colon (X,d_2)\to (X, d_1)$ in $\PMet$ be such that $Uf = \id_X$.
  Then
  $\LiftedFunctor{F}f\colon \LiftedFunctor{F} (X,d_2)\to
  \LiftedFunctor{F}(X, d_1)$ is in $\PMet$ as well by functoriality.}
%   Note that for $i = 1, 2$,
%   $U\LiftedFunctor{F}(X,d_i) = U(FX, d^F_i) = FX$.  So
%   $\LiftedFunctor{F}(X,d_i) = (FX, d^F_i)$.  And
%   $U\LiftedFunctor{F}f = FU f = F \id_X = \id_{FX}$.  Thus
%   $\LiftedFunctor{F}f$ shows that $d^F_1 \leq d^F_2$.  
\new{%
By definition of a lifting we have $\LiftedFunctor{F}(X, d_i) = (FX, d^F_i)$, hence $U\LiftedFunctor{F}(X,d_i) = U(FX, d^F_i) = FX$.  Moreover, $U\LiftedFunctor{F}f = FU f = F \id_X = \id_{FX}$.  Therefore $\LiftedFunctor{F}f$ shows that $d^F_1 \leq d^F_2$.
}

% \hrule
  
% 	Since $d_1 \leq d_2$, the identity function on the set $X$ can be regarded as a nonexpansive function $f\colon (X,d_2) \nonexpansiveTo (X, d_1)$ because we have for all $x,y \in X$ that $d_1\big(f(x),f(y)\big) = d_1(x,y) \leq d_2(x,y)$. By functoriality of $\LiftedFunctor{F}$ we know that also $\LiftedFunctor{F}f\colon (FX, \LiftedMetric{F}{d_2}) \nonexpansiveTo (FX, \LiftedMetric{F}{d_1})$ is nonexpansive, i.e., for all $t_1, t_2 \in FX$ we have $\LiftedMetric{F}{d_1}\big(FUf(t_1),FUf(t_2)\big) \leq \LiftedMetric{F}{d_2}(t_1,t_2)$ and moreover $\LiftedMetric{F}{d_1}\big(FUf(t_1),FUf(t_2)\big) = \LiftedMetric{F}{d_1}(F\id_X(t_1),F\id_X(t_2))= \LiftedMetric{F}{d_1}(\id_{FX}(t_1),\id_{FX}(t_2) )= \LiftedMetric{F}{d_1}(t_1, t_2)$ and thus $\LiftedMetric{F}{d_1} \leq \LiftedMetric{F}{d_2}$. 
\end{proof}

In order to define a lifting to $\PMet$ we will just use one simple tool, an evaluation function which describes how to transform an element of $F[0,\top]$ to a real number.

\begin{definition}[Evaluation Function, Evaluation Functor]
	\label{def:evfct} 
	\index{evaluation function}
	\index{evaluation functor}
	Let $F$ be an endofunctor on $\Set$. An \emph{evaluation function} for $F$ is a function $\ev_F\colon F\reals \to \reals$. Given such a function, we define the \emph{evaluation functor} to be the endofunctor $\EvaluationFunctor{F}$ on $\Set/\reals$, the slice category\footnote{The slice category $\Set/\reals$ has as objects all functions $g\colon X\to\reals$ where $X$ is an arbitrary set. Given $g$ as before and $h\colon Y \to \reals$, an arrow from $g$ to $h$ is a function $f\colon X \to Y$ satisfying $h \circ f = g$.} over $\reals$, via $\EvaluationFunctor{F}(g) = \ev_F\circ Fg$ for all $g \in \Set/\reals$. On arrows $\EvaluationFunctor{F}$ is defined as $F$. 
\end{definition}

We quickly remark that by definition of $\EvaluationFunctor{F}$ on arrows, it is immediately clear that one indeed obtains a functor so the name is justified.

\subsection{The Kantorovich Lifting}
\label{subsec:kantorovich-lifting}
Let us now consider an endofunctor $F$ on $\Set$ with an evaluation function $\ev_F$. Given a pseudometric space $(X,d)$, our first approach to lift $d$ to $FX$ will be to take the smallest possible pseudometric $d^F$ on $FX$ such that, for all nonexpansive functions $f\colon (X,d) \nonexpansiveTo (\reals,d_e)$, also $\EvaluationFunctor{F}f\colon (FX,d^F)\nonexpansiveTo (\reals,d_e)$ is nonexpansive again, i.e., we want to ensure that for all $t_1,t_2 \in FX$ we have $d_e\big(\EvaluationFunctor{F}f(t_1), \EvaluationFunctor{F}f(t_2)\big)\leq d^F(t_1,t_2)$. This idea immediately leads us to the following definition which corresponds to the maximization of the logistic firm's prices in the introductory example.

\begin{definition}[Kantorovich Distance]
	\label{def:kantorovich}
	\index{Kantorovich pseudometric}
	Let $F\colon \Set \to \Set$ be a functor with an evaluation function $\ev_F$. For every pseudometric space $(X,d)$ the \emph{Kantorovich distance} on $FX$ is the function $\Kantorovich{F}{d}\colon FX\times FX\to \reals$, where 
	\begin{align*}
		\Kantorovich{F}{d}(t_1,t_2) := \sup \set{d_e\Big(\EvaluationFunctor{F}f(t_1),\EvaluationFunctor{F}f(t_2)\Big) \mid f\colon (X,d) \nonexpansiveTo (\reals,d_e)}
	\end{align*}
	for all $t_1,t_2 \in FX$. 
\end{definition}

Note that a nonexpansive function $f\colon (X,d) \nonexpansiveTo (\reals,d_e)$ always exists. If $X = \emptyset$ it is the unique empty function and for $X \not=\emptyset$ every constant function is nonexpansive. Moreover, it is easy to show that $\Kantorovich{F}{d}$ is a pseudometric.

\begin{theorem}[Kantorovich Pseudometric]
	\label{prop:kantorovich-is-pseudometric}
	For every pseudometric space $(X,d)$ the Kantorovich distance $\Kantorovich{F}{d}$ is a pseudometric on $FX$.
\end{theorem}
\begin{proof}
	Reflexivity and symmetry are an immediate consequence of the fact that $d_e$ is a metric. We now show the triangle inequality. Let $t_1, t_2, t_3 \in FX$, then 
	\begin{align*}
		& \quad \Kantorovich{F}{d}(t_1, t_2) + \Kantorovich{F}{d}(t_2, t_3) \\
		&=\sup_{f \colon (X,d) \nonexpansiveTo (\reals,d_e)} {d_e\left(\EvaluationFunctor{F}f(t_1),\EvaluationFunctor{F}f(t_2)\right)} + \sup_{f \colon (X,d) \nonexpansiveTo (\reals,d_e)} {d_e\left(\EvaluationFunctor{F}f(t_2),\EvaluationFunctor{F}f(t_3)\right)}\\
		&\geq \sup_{f \colon (X,d) \nonexpansiveTo (\reals,d_e)} {\left(d_e\left(\EvaluationFunctor{F}f(t_1),\EvaluationFunctor{F}f(t_2)\right) + d_e\left(\EvaluationFunctor{F}f(t_2),\EvaluationFunctor{F}f(t_3)\right)\right)}\\
		&\geq \sup_{f \colon (X,d) \nonexpansiveTo (\reals,d_e)} {d_e\left(\EvaluationFunctor{F}f(t_1),\EvaluationFunctor{F}f(t_3)\right)} = \Kantorovich{F}{d}(t_1,t_3)
	\end{align*}
	where the first inequality is a simple property of the supremum and the second inequality follows again from the fact that $d_e$ is a metric. 
\end{proof}

Using this pseudometric we can now immediately define our first lifting.

\begin{definition}[Kantorovich Lifting]
	\index{Kantorovich lifting}
	Let $F\colon \Set \to \Set$ be a functor with an evaluation function $\ev_F$. We define the \emph{Kantorovich lifting} of $F$ to be the functor $\LiftedFunctor{F}\colon \PMet \to \PMet$, $\LiftedFunctor{F}(X,d) = (FX,\Kantorovich{F}{d})$, $\LiftedFunctor{F}f = Ff$.
\end{definition}

Since $\LiftedFunctor{F}$ inherits the preservation of identities and composition of morphisms from $F$ we just need to prove that nonexpansive functions are mapped to nonexpansive functions to obtain functoriality of $\LiftedFunctor{F}$.

\begin{theorem}
	\label{prop:kantorovich-functorial}
	The Kantorovich lifting $\LiftedFunctor{F}$ is a functor on pseudometric spaces.
\end{theorem}
\begin{proof}
	$\LiftedFunctor{F}$ preserves identities and composition of arrows because $F$ does. Moreover, it preserves nonexpansive functions: Let $f\colon (X,d_X) \nonexpansiveTo (Y,d_Y)$ be nonexpansive and $t_1, t_2 \in FX$, then
	\begin{align*}
		\Kantorovich{F}{d_Y}\big(Ff(t_1), Ff(t_2)\big)&=\sup_{g \colon (Y,d_Y) \nonexpansiveTo (\reals, d_e)} d_e\Big(\EvaluationFunctor{F}(g \circ f)(t_1),\EvaluationFunctor{F}(g \circ f)(t_2)\Big) \\
		&\leq \sup_{h \colon (X,d_X) \nonexpansiveTo (\reals,d_e)} d_e\Big(\EvaluationFunctor{F}(h)(t_1),\EvaluationFunctor{F}(h)(t_2)\Big) = \Kantorovich{F}{d_X}(t_1, t_2)
	\end{align*}
	due to the fact that since both $f$ and $g$ are nonexpansive also their composition $(g \circ f) \colon (X,d_X) \nonexpansiveTo (\reals,d_e)$ is nonexpansive.
\end{proof}

With this result at hand we are almost done: The only remaining task is to show that $\LiftedFunctor{F}$ is a lifting of $F$ in the sense of \Cref{def:liftingPMet} but this is indeed obvious by definition of $\LiftedFunctor{F}$.

An important property of this lifting is that it preserves isometries, which is a bit tricky to show. While one might be tempted to think that this is immediately true because functors preserve isomorphisms, it is easy to see that the isomorphisms of $\PMet$ are the \emph{bijective} isometries. However, there are of course also isometries which are not bijective and thus not isomorphisms. Simple examples for this arise by taking the unique discrete metric space $(\one, d_\one)$ and mapping it into any other pseudometric space $(X,d)$ with $|X| > 1$. Any function $\one \to X$ is necessarily isometric but certainly not bijective.

\begin{theorem}
	\label{prop:kantorovich-lifting-preserves-isometries}
	The Kantorovich lifting $\LiftedFunctor{F}$ of a functor $F$ preserves isometries. 
\end{theorem}
\begin{proof}
	Let $f\colon (X,d_X) \nonexpansiveTo (Y,d_Y)$ be an isometry, i.e., $f$ satisfies $d_Y \circ (f \times f) = d_X$. Since the Kantorovich lifting $\LiftedFunctor{F}$ is a functor on pseudometric spaces, we already know that $\LiftedFunctor{F}f$ is nonexpansive, i.e., we know that $\Kantorovich{F}{d_Y} \circ ({F}f \times {F}f) \leq \LiftedMetric{F}{d_X}$ thus we only have to show the opposite inequality. We will do that by constructing for every nonexpansive function $g\colon (X, d_X) \nonexpansiveTo (\reals, d_e)$ a nonexpansive function $h\colon (Y,d_Y) \nonexpansiveTo (\reals, d_e)$ such that $h \circ f = g$ which implies that for every $t_1,t_2 \in FX$ we have equality $d_e(\EvaluationFunctor{F}h(Ff(t_1)), \EvaluationFunctor{F}h(Ff(t_2)))= d_e(\EvaluationFunctor{F}g(t_1), \EvaluationFunctor{F}g(t_2))$. Then we have
	\begin{align*}
		\Kantorovich{F}{d_Y} &\circ ({F}f \times {F}f)(t_1, t_2) = \sup \set{d_e\Big(\EvaluationFunctor{F}h (Ff(t_1)),\EvaluationFunctor{F}h (Ff(t_2))\Big) \mid h\colon (Y,d_Y) \nonexpansiveTo (\reals,d_e)}\\
		& \geq \sup \set{d_e\Big(\EvaluationFunctor{F}g (t_1),\EvaluationFunctor{F}g (t_2)\Big) \mid g\colon (X,d_X) \nonexpansiveTo (\reals,d_e)} = \Kantorovich{F}{d_X}(t_1,t_2)\,.
	\end{align*}
We define $h$ via $h(y) := \inf_{x \in X} \{ g(x) + d_Y(f(x), y)\}$. This function is nonexpansive, in fact, for all $y_1, y_2 \in \new{Y}$ we have
        \begin{align*}
          h(y_1) & =     \inf_{x \in X} \{ g(x) + d_Y(f(x), y_1)\} \leq  \inf_{x \in X} \{ g(x) + d_Y(f(x), y_2) + d_Y(y_1,y_2)\}\\
                 & = \inf_{x \in X} \{ g(x) + d_Y(f(x), y_2) \} + d_Y(y_1,y_2)  = h(y_2) + d_Y(y_1,y_2)
        \end{align*}
        where the inequality follows by the fact that $d_Y$ satisfies the triangle inequality. The same calculation also yields $h(y_2) \leq h(y_1) + d_Y(y_1, y_2)$. We use \Cref{lem:sum-vs-dist} with $a = h(y_1)$, $b=h(y_2)$ and $c=d_Y(y_1, y_2)$ to conclude that $d_e(h(y_1), h(y_2)) = d_e(a, b) \leq c = d_Y(y_1,y_2)$.
        
        We can apply \Cref{lem:sum-vs-dist} with $a = g(x)$, $b=g(x')$ and $c=d(x, x')$ in the other direction using nonexpansiveness of $g$ to obtain $g(x') + d_X(x',x) = b + c \geq a = g(x)$. Thus $h(f(x))  = \inf_{x' \in X} \{ g(x') + d_Y(f(x'), f(x)) \} = \inf_{x' \in X} \{ g(x') + d_X(x', x) \} \geq g(x)$ and with $x'=x$ we obtain the equality $h(f(x))=g(x)$ for all $x \in X$.
\end{proof}

With this result in place, let us quickly discuss the name of our lifting. We chose the name \emph{Kantorovich} because our definition is reminiscent of the Kantorovich pseudometric in probability theory. If we take the proper combination of functor and evaluation function, we can recover that pseudometric (in the discrete case) as the first instance of our framework. 

\begin{example}[Kantorovich Lifting of the Distribution Functor]
	\label{exa:probability-distribution-functor}
	We take $\top = 1$ and the probability distribution functor
        $\Distributions$ (or the sub-probability distribution functor). As evaluation $\ev_\Distributions\colon \Distributions [0,1] \to [0,1]$ we define, for each $P \in \Distributions[0,1]$, $\ev(P)$ to be the expected value of the identity function on $[0,1]$, i.e., $\ev_\Distributions (P) :=  \sum_{x\in [0,1]} x\cdot P(x)$. Then for any function $g\colon X \to [0,1]$ and any (sub)distribution $P \in \Distributions X$ we can easily see that $\EvaluationFunctor{\Distributions }g(P) = \ev_\Distributions \circ \Distributions g (P)= \sum_{x \in X}g(x) \cdot P(x)$. Using this, we can see that for every pseudometric space $(X,d)$ we obtain the usual Kantorovich pseudometric $\Kantorovich{\Distributions }{d}\colon (\Distributions X)^2 \to [0,1]$, where
	\begin{align*}
		\Kantorovich{\Distributions }{d}(P_1,P_2) =
                \sup\set{ \big|\sum_{x \in X} f(x)\cdot (P_1(x) - P_2(x))\big|\ \bigg|\ f \colon (X,d) \nonexpansiveTo ([0,1],d_e)}
	\end{align*}
	for all (sub)probability distributions $P_1,P_2 \colon X \to [0,1]$.
\end{example}

Let us now consider the question whether the Kantorovich lifting preserves metrics, i.e., we want to check whether the Kantorovich pseudometric $\Kantorovich{F}{d}$ is a metric for a metric space $(X,d)$. The next example shows that this is not necessarily the case.

\begin{example}[Kantorovich Lifting of the Squaring Functor]
	\label{exa:kant-no-metric}
	\label{exa:squaring-functor}
	\index{squaring functor}
The \emph{squaring functor} on $\Set$ is the functor $S\colon \Set \to \Set$ where $SX = X \times X$ for each set $X$ and $Sf = f \times f$ for each function $f \colon X \to Y$ (see e.g. \cite[Ex.~3.20 (10)]{AHS90}).

We take $\top = \infty$ and the evaluation function $\ev_{S}\colon [0,\infty] \times [0,\infty] \to [0,\infty]$, $\ev_{S}(r_1, r_2) = r_1+r_2$. For a metric space $(X,d)$ with $|X|\geq 2$ take $t_1 = (x_1, x_2) \in SX$ with $x_1 \not = x_2$ and define $t_2 := (x_2, x_1)$. Clearly $t_1 \not = t_2$ but for every nonexpansive function $f\colon (X,d) \nonexpansiveTo (\reals, d_e)$ we have $\EvaluationFunctor{S}f(t_1) = f(x_1) + f(x_2) = f(x_2) + f(x_1) = \EvaluationFunctor{S}f(t_2)$ and thus $\Kantorovich{S}{d}(t_1,t_2) = 0$.
\end{example}

\subsection{The Wasserstein Lifting}
\label{sec:wasserstein}
We have seen that our first lifting approach bears close resemblance to the original Kantorovich pseudometric on probability measures. We will now also define a generalized version of the Wasserstein pseudometric and compare it with our generalized Kantorovich pseudometric. To do that we first need to define generalized couplings, which can be understood as a generalization of joint probability measures.

\begin{definition}[Coupling]
	\label{def:coupling}
	\index{coupling}
	Let $F \colon \Set \to \Set$ be a functor and $n \in \N$. Given a set $X$ and $t_i \in FX$ for $1 \leq i \leq n$ we call an element $t \in F(X^n)$ such that $F\pi_i(t) = t_i$ a \emph{coupling} of the $t_i$ (with respect to $F$). We write $\Couplings{F}(t_1,\dots,t_n)$ for the set of all these couplings. 
\end{definition}

Using these couplings we now proceed to define an alternative pseudometric on $FX$. As was the case for the Kantorovich distance, we only need to choose an evaluation function $\ev_F$ for our functor and then use the corresponding evaluation functor $\EvaluationFunctor{F}$ (see \Cref{def:evfct})

\begin{definition}[Wasserstein Distance]
  \label{def:wasserstein}
  Let $F\colon \Set \to \Set$ be a functor with evaluation function
  $\ev_F$. For every pseudometric space $(X,d)$ the \emph{Wasserstein
    distance} on $FX$ is the function
  $\Wasserstein{F}{d} \colon FX \times FX \to \reals$ given by
  \begin{align}
    \Wasserstein{F}{d}(t_1, t_2) := \inf
    \set{\EvaluationFunctor{F}d(t) \,\big|\, t \in
      \Couplings{F}(t_1,t_2)}\,.\label{eq:wasserstein}
  \end{align}
  for all $t_1,t_2 \in FX$.
\end{definition}

\new{Given $t_1,t_2\in FX$, we say that a coupling
  $t\in \Couplings{F}(t_1,t_2)$ is optimal whenever
  $\Wasserstein{F}{d}(t_1,t_2)=\EvaluationFunctor{F}d(t)$. Note that
  such an optimal coupling need not always exist.}

In contrast to the Kantorovich distance, where the respective set
cannot be empty because nonexpansive functions always exist, here it
might be the case that no coupling exists and thus
$\Wasserstein{F}{d}(t_1,t_2) = \inf \emptyset = \top$. Without any
additional conditions we cannot even prove that
$\Couplings{F}(t_1,t_1) \not = \emptyset$ which we would certainly
need for reflexivity and thus we do not automatically obtain a
pseudometric. The only property we get for free is symmetry.

\begin{lemma}[Symmetry of the Wasserstein Distance]
	\label{lem:wasserstein-symmetric}
	For all pseudometric spaces $(X,d)$ the Wasserstein distance $\Wasserstein{F}{d}$ is symmetric.
\end{lemma}
\begin{proof}
	Let $t_1, t_2 \in FX$  and let $\sigma := \langle \pi_2, \pi_1\rangle$ be the swap map on $X \times X$, i.e., $\sigma\colon X \times X \rightarrow X \times X$, $\sigma(x_1,x_2) = (x_2,x_1)$. If there is a coupling $t_{12} \in \Couplings{F}(t_1,t_2)$ we define $t_{21} := F\sigma(t_{12}) \in F(X \times X)$ and observe that it satisfies $F\pi_1(t_{21}) = F\pi_1(F\sigma(t_{12})) = F(\pi_1 \circ \sigma)(t_{12})= F\pi_2(t_{12}) = t_2$ and analogously $F\pi_2(t_{21}) = t_1$, thus $t_{21} \in \Couplings{F}(t_2, t_1)$. Moreover, due to symmetry of $d$ (i.e., $d \circ \sigma = d$), we obtain $\EvaluationFunctor{F}d(t_{21}) = \ev_{F} \big(Fd(t_{21})\big) = \ev_F \big(Fd(F\sigma(t_{12})\big) = \ev_F\big(F(d \circ \sigma)(t_{12}) \big)= \ev_F \big(Fd(t_{12})\big) = \EvaluationFunctor{F}d(t_{12})$ which yields the desired symmetry. If no coupling of $t_1$ and $t_2$ exists, there is also no coupling of $t_2$ and $t_1$ because otherwise we would get a coupling of $t_1$ and $t_2$ using the above method. Thus $\Wasserstein{F}{d}(t_1,t_2) = \top = \Wasserstein{F}{d}(t_2,t_1)$.
\end{proof}

In order to be able to guarantee the other two properties of a pseudometric we will restrict our attention to well-behaved evaluation functions.

\begin{definition}[Well-Behaved Evaluation Function]
	\label{def:well-behaved}
	Let $\ev_F$ be an evaluation function for a functor $F\colon \Set \to \Set$. We call $\ev_F$ \emph{well-behaved} if it satisfies the following conditions:
	\begin{wbconditions}
		\item $\EvaluationFunctor{F}$ is monotone, i.e., for $f,g\colon X\to\reals$ with  $f\le g$, we have $\EvaluationFunctor{F}f\le\EvaluationFunctor{F}g$.\label{W1}
		\item For any $t\in F(\reals^2)$ we have $d_e\big(\ev_F(t_1), \ev_F(t_2)\big) \leq \EvaluationFunctor{F}d_e(t)$ for $t_i : = F\pi_i(t)$.\label{W2}
		\item $\ev_F^{-1}\big[\{0\}\big] = Fi\big[F\{0\}\big]$ where $i \colon \set{0} \hookrightarrow\reals$ is the inclusion map. \label{W3}
	\end{wbconditions}
\end{definition}

\noindent While \cref*{W1} is quite natural, for \cref*{W2,W3} some explanations are in order. \Cref*{W2} ensures that $\EvaluationFunctor{F}\id_\reals = \ev_F\colon F\reals\to\reals$ is nonexpansive once $d_e$ is lifted to $F\reals$ (recall that for the Kantorovich lifting we require $\EvaluationFunctor{F}f$ to be nonexpansive for \emph{any} nonexpansive $f$). By definition of the evaluation functor $\EvaluationFunctor{F}$ and the $t_i$, we have $d_e\big(\EvaluationFunctor{F}\pi_1(t), \EvaluationFunctor{F}\pi_2(t)\big) =  d_e\big(\ev_F \circ F\pi_(t), \ev_F \circ F\pi_2(t)\big) = d_e\big(\ev_F(t_1), \ev_F(t_2)\big)$ so \Cref{W2} can equivalently be stated as $d_e\big(\EvaluationFunctor{F}\pi_1(t), \EvaluationFunctor{F}\pi_2(t)\big) \leq \EvaluationFunctor{F}d_e(t)$.

\Cref*{W3} requires that exactly the elements of $F\{0\}$ are mapped to $0$ via $\ev_F$. This ensures the reflexivity of the Wasserstein pseudometric. 

\begin{lemma}[Reflexivity of the Wasserstein Function]
	\label{lem:wasserstein-reflexive}
	Let $F$ be an endofunctor  on $\Set$ with evaluation function $\ev_F$. If  $\ev_F$ satisfies \Cref{W3} of \Cref{def:well-behaved} then for any pseudometric space $(X,d)$ the Wasserstein distance $\Wasserstein{F}{d}$ is reflexive.
\end{lemma}
\begin{proof}
	Let $t_1 \in FX$.  To show reflexivity we will construct a coupling $t \in \Couplings{F}(t_1,t_1)$ such that $\EvaluationFunctor{F}d(t) = 0$. In order to do that, let $\delta \colon X \to X^2, \delta(x) = (x,x)$ and define $t := F\delta(t_1)$. Then $F\pi_i(t) = F(\pi_i \circ \delta)(t_1) = F(\id_X)(t_1) = t_1$ and thus $t \in \Couplings{F}(t_1,t_1)$. Since $d$ is reflexive, $d \circ \delta\colon X \to [0,\top]$ is the constant zero function. Let $i\colon \set{0} \hookrightarrow [0,\top], i(0) = 0$ and for any set $X$ let $!_X\colon X \to \set{0}, !_X(x) = 0$. Then also $i \circ !_X \colon X \to [0,\top]$ is the constant zero function and thus $d \circ \delta = i \circ !_X$. Using this we can conclude that 
	\begin{align*}
	\EvaluationFunctor{F}{d}(t) = \EvaluationFunctor{F}{d}\big(F\delta(t_1)\big) = \EvaluationFunctor{F}(d \circ \delta)(t_1) =  \EvaluationFunctor{F}(i \circ !_X)(t_1) = \ev_F\Big(Fi\big((F!_X)(t_1)\big)\Big) = 0
	\end{align*}
	where the last equality follows from the fact that $F!_X(t_1) \in F\{0\}$ and \Cref{W3} of \Cref{def:well-behaved}.
\end{proof}

Before we continue our efforts to obtain a Wasserstein pseudometric, we check that well-behaved evaluation functions exist but not every evaluation function is well-behaved.

\begin{example}[Evaluation Function for the Powerset Functor]
	\label{ex:evaluation-max}
	We take $\top = \infty$ and consider the powerset functor $\Powerset$. First, we show that the evaluation function $\sup\colon \Powerset[0,\infty] \to [0,\infty]$ where $\sup \emptyset :=0$ is well-behaved.
	\begin{wbconditions}
		\item Let $f,g\colon X\to [0,\infty]$ with $f\le g$. Let $S\in \Powerset X$, i.e., $S\subseteq X$. Then we have $\EvaluationFunctor{\Powerset}f(S) = \sup f[S] = \sup \set{f(x) \mid x \in S} \le \sup \set{g(x) \mid x \in S}= \sup g[S] = \EvaluationFunctor{\Powerset}g(S)$.
		\item For any subset $S\subseteq [0,\infty]^2$ we have to show the inequality \begin{align}
		d_e\big(\sup \pi_1[S],\sup \pi_2[S]\big) \le \sup d_e[S]\,.\label{eq:pfin-wb:1}
		\end{align}
		For $S = \emptyset$ this is true because $\sup \emptyset = 0$ and thus both sides of the inequality are $0$. Otherwise we define $s_i := \sup\pi_i[S]$. Clearly, if $s_1 = s_2$ then the left hand side of \eqref{eq:pfin-wb:1} is $0$ and thus the inequality holds. 
		Without loss of generality we now assume $s_1 < s_2$ and distinguish two cases.

If $s_2 < \infty$ then for any $\epsilon > 0$ we can find a pair $(t_1,t_2) \in S$ such that $s_2 - \epsilon < t_2$ because $s_2$ is the supremum. Moreover, $t_1 \leq s_1$ and if $\epsilon < s_2-s_1$ also $t_1 \leq t_2$. By combining these inequalities, we conclude that for every $\epsilon \in \,]0,s_2-s_1[$ we have a pair $(t_1,t_2) \in S$ such that $d_e(s_1,s_2) - \epsilon = s_2 - s_1 - \epsilon \leq s_2 - t_1 - \epsilon = s_2-\epsilon - t_1 < t_2-t_1=d_e(t_1,t_2)$. Since $\epsilon$ can be arbitrarily small, we thus must have \eqref{eq:pfin-wb:1}.

If $s_2 = \infty$ then $d_e(s_1,s_2) = \infty$. However, $s_2 = \infty$ also implies that for every non-negative real number $r \in \R_+$ we can find an element $(t_1,t_2) \in S$ such that $t_2 > r$. Especially, for $r > s_1$ we have $t_2 > r > s_1 \geq t_1$ and thus $d_e(t_1,t_2) = t_2 - t_1 \geq t_2-s_1 > r - s_1$. Since $r$ can be arbitrarily large, we thus must have $\sup d_e[S] = \infty$ and thus \eqref{eq:pfin-wb:1} is an equality.

		\item We have $\Powerset i [\Powerset \{0\}] = \Powerset i [\{\emptyset,\{0\}\}] = \{i[\emptyset],i[\{0\}]\} = \{\emptyset, \{0\}\} = \sup^{-1}[\{0\}]$.
	\end{wbconditions}
	Whenever we work with the finite powerset functor $\PowersetFinite$ we can of course use $\max$ instead of $\sup$ with the convention $\max \emptyset = 0$.
	
	In contrast to the above, $\inf\colon \Powerset(\prealinf) \to \prealinf$ is not well-behaved. It neither satisfies \Cref{W2}, nor \Cref{W3}: $\inf d_e[S] \ge d_e\big(\inf \pi_1[S],\inf \pi_2[S]\big)$ fails for $S = \{(0,1),(1,1)\}$ and $\{0,1\} \in \inf^{-1}[\set{0}]$.
\end{example}

Staying with \Cref{W3} for a while we remark that it can be expressed as a weak pullback diagram, thus fitting nicely into a coalgebraic framework.\medskip\newline
\begin{minipage}{11.7cm}
\begin{lemma}[Weak Pullback Characterization of \Cref*{W3}]
\label{lem:W3-weak-pb}
Let $F$ be an endofunctor on $\Set$ with evaluation function $\ev_F$ and $i \colon \set{0} \hookrightarrow [0,\top]$ be the inclusion function. For any set $X$ we denote the unique arrow into $\set{0}$ by $!_X\colon X \to \set{0}$. Then $\ev_F$ satisfies \Cref{W3} of \Cref{def:well-behaved}, i.e., $\ev_F^{-1}[\{0\}] = F[F\{0\}]$ if and only if the diagram on the right is a weak pullback.
\end{lemma}
\end{minipage}\hfill
\begin{minipage}{3.5cm}\begin{diagram}
		\matrix(m)[matrix of math nodes, column sep=20pt, row sep=30pt,ampersand replacement=\&]{
			F\set{0} \& \set{0}\\
			F[0,\top] \& {[0,\top]}\\
		};
		\draw[->] (m-1-1) edge node[above]{$!_{F\set{0}}$} (m-1-2);
		\draw[->] (m-1-1) edge node[left]{$Fi$} (m-2-1);
		\draw[right hook->] (m-1-2) edge node[right]{$i$} (m-2-2);
		\draw[->] (m-2-1) edge node[below]{$\ev_F$} (m-2-2);
	\end{diagram}\end{minipage}
	
\begin{proof}
	We consider the extended diagram on the right hand side below. Commutativity of the diagram is equivalent to $\ev_F^{-1}[\set{0}] \supseteq  Fi[F\set{0}]$.\\
	\begin{minipage}{10cm}
	\parindent=1.5em\relax
	\indent Given a set $X$ and a function $f \colon X \to F[0,\top]$ as depicted below, we conclude again by commutativity ($i \circ !_X = \ev_F \circ f$) that $f(x) \in \ev_F^{-1}[\set{0}]$ for all $x \in X$.
	
	Now we show that the weak universality is equivalent to the other inclusion. First suppose that $\ev_F^{-1}[\set{0}] \subseteq  Fi[F\set{0}]$ then for $f(x) \in \ev_F^{-1}[\set{0}]$ we can choose a (not necessarily unique) $x_0 \in F\set{0}$ such that $f(x) = Fi(x_0)$. If we define $\phi\colon X \to F\set{0}$ by $\phi(x)  = x_0$ then clearly $\phi$ makes the above diagram commute and thus we have a weak pullback.
\end{minipage}\hfill
\begin{minipage}{5cm}\begin{diagram}
		\matrix(m)[matrix of math nodes, column sep=20pt, row sep=20pt]{
			X\\
			&F\set{0} & \set{0}\\
			&F[0,\top] & {[0,\top]}\\
		};
		\draw[->] (m-1-1) edge[bend left] node[above] {$!_X$} (m-2-3);
		\draw[->] (m-1-1) edge[bend right] node[left] {$f$} (m-3-2);
		\draw[->] (m-1-1) edge[dashed] node[above] {$\phi$} (m-2-2);
		\draw[->] (m-2-2) edge node[above]{$!_{F\set{0}}$} (m-2-3);
		\draw[->] (m-2-2) edge node[left]{$Fi$} (m-3-2);
		\draw[right hook->] (m-2-3) edge node[right]{$i$} (m-3-3);
		\draw[->] (m-3-2) edge node[below]{$\ev_F$} (m-3-3);
	\end{diagram}\end{minipage}
	\indent Conversely if the diagram is a weak pullback we consider the set $X = \ev_F^{-1}[\set{0}]$ and the function $f \colon \ev_F^{-1}[\set{0}] \hookrightarrow F[0,\top], f(x) = x$. Now for any $x \in \ev_F^{-1}[\set{0}]$ we have $Fi (\phi(x)) = (Fi \circ \phi)(x) = f(x) = x$, hence -- since $\phi(x) \in F\set{0}$ -- we have $x \in Fi[F\set{0}]$. This shows that indeed $\ev_F^{-1}[\set{0}] \subseteq Fi[F\set{0}]$ holds.
\end{proof}

The only missing step towards the Wasserstein pseudometric is the observation that if $F$ preserves weak pullbacks we can define new couplings based on given ones.

\begin{lemma}[Gluing Lemma]
	\label{lem:coupling}
	\index{Gluing lemma}
	Let $F$ be an endofunctor on $\Set$, $X$ a set, $t_1, t_2, t_3 \in FX$, $t_{12} \in \Couplings{F}(t_1,t_2)$, and $t_{23} \in \Couplings{F}(t_2,t_3)$ be couplings. If $F$ preserves weak pullbacks then there is a coupling $t_{123} \in \Couplings{F}(t_1, t_2, t_3)$ such that
	\begin{align*}
		F\langle \pi_1, \pi_2\rangle(t_{123}) = t_{12} \quad  \text{and} \quad F\langle\pi_2, \pi_3\rangle(t_{123}) = t_{23}
	\end{align*} 
	where $\pi_i\colon X^3 \to X$ are the projections of the ternary product. Moreover, $t_{13}:= F\langle\pi_1, \pi_3\rangle(t_{123})$ is a coupling of $t_1$ and $t_3$, i.e., we have $t_{13} \in \Couplings{F}(t_1,t_3)$.
\end{lemma}
\begin{proof}
	Let $\tau_i\colon X \times X$ be the projections of the binary product. We first observe that the following diagram is a pullback square.
	\begin{center}\begin{diagram}
		\matrix(m)[matrix of math nodes, column sep=15pt, row sep=10pt]{
			& X\times X \times X &\\
			X\times X && X\times X\\
			& X \\
		};
		\draw[->] (m-1-2) edge node[above left]{$\langle\pi_1, \pi_2\rangle$} (m-2-1);
		\draw[->] (m-1-2) edge node[above right]{$\langle\pi_2, \pi_3\rangle$} (m-2-3);
		\draw[->] (m-2-1) edge node[below left]{$\tau_2$} (m-3-2);
		\draw[->] (m-2-3) edge node[below right]{$\tau_1$} (m-3-2);
	\end{diagram}\end{center}
	Given any set $P$ along with functions $p_1,p_2\colon P \to X \times X$ satisfying the condition $\tau_2 \circ p_1 = \tau_1 \circ p_2$ the unique mediating arrow $u\colon P \to X\times X \times X$ is given by $u = \langle\tau_1 \circ p_1, \tau_2\circ p_1, \tau_2\circ p_2\rangle = \langle\tau_1 \circ p_1, \tau_1\circ p_2, \tau_2\circ p_2\rangle$. Now let us look at the following diagram. 
	 
	 \begin{center}\begin{diagram}
 		\matrix(m)[matrix of math nodes, column sep=40pt, row sep=10pt]{
 			&& F(X\times X \times X)\\
 			& F(X\times X) && F(X\times X)\\
 			FX && FX && FX\\
 		};
 		\draw[->] (m-1-3) edge node[above left]{$F\langle\pi_1, \pi_2\rangle$} (m-2-2);
 		\draw[->] (m-1-3) edge node[above right]{$F\langle\pi_2, \pi_3\rangle$} (m-2-4);
 		\draw[->] (m-2-2) edge node[below left]{$F\tau_2$} (m-3-3);
 		\draw[->] (m-2-4) edge node[below right]{$F\tau_1$} (m-3-3);
 		\draw[->] (m-2-2) edge node[above left]{$F\tau_1$} (m-3-1);
 		\draw[->] (m-2-4) edge node[above right]{$F\tau_2$} (m-3-5);
	\end{diagram}\end{center}
	 Since $F$ preserves weak pullbacks, the inner part of this diagram is a weak pullback. We recall that $t_{12} \in \Couplings{F}(t_1,t_2)$ and $t_{23} \in \Couplings{F}(t_2,t_3)$ so we have $F\tau_2(t_{12}) = t_2 = F\tau_1(t_{23})$. Using this, we can use the (weak) universality of the pullback to obtain\footnote{Explicitly: Consider $\set{t_2}$ with functions $p_1, p_2\colon \set{t_2}\to F(X\times X)$ where $p_1(t_2) = t_{12}$ and $p_2(t_2) = t_{23}$, then by the weak pullback property there is a (not necessarily unique) function $u\colon \set{t_2} \to F(X \times X \times X)$ satisfying $F(\langle\pi_1, \pi_2\rangle) \circ u = p_1$ and $F(\langle\pi_2, \pi_3\rangle) \circ u = p_2$. We simply define $t_{123} := u(t_2)$.} an element $t_{123} \in F(X\times X \times X)$ which satisfies the two equations of the lemma and moreover $F\pi_1(t_{123}) = F\big(\tau_1 \circ \langle\pi_1, \pi_2\rangle\big)(t_{123}) = F\tau_1 \circ F\langle\pi_1, \pi_2\rangle(t_{123}) = F\tau_1(t_{12}) = t_1$ and analogously $F\pi_2(t_{123}) = t_2$, $F\pi_3(t_{123}) = t_3$ yielding $t_{123} \in \Couplings{F}(t_1, t_2, t_3)$.
	 
	For $t_{13} := F\langle\pi_1, \pi_3\rangle (t_{123})$ we observe that $F\tau_1(t_{13}) = F\tau_1\big(F\langle\pi_1, \pi_3\rangle (t_{123})\big) = F\big(\tau_1\circ \langle\pi_1, \pi_3\rangle\big) (t_{123}) = F\pi_1 (t_{123}) = t_1$ and similarly $F\tau_2(t_{13}) = t_3$ so $t_{13} \in \Couplings{F}(t_1,t_3)$ as claimed.
\end{proof}

With the help of this lemma we can now finally give sufficient conditions to guarantee that the Wasserstein distance satisfies the triangle inequality. Apparently, since we use the above lemma, this will work only for weak pullback preserving functors and we will also need \Cref{W1,W2} of \Cref{def:well-behaved} for the proof. 

\begin{lemma}[Triangle Inequality for the Wasserstein Function]
	\label{lem:wasserstein-triangle-inequality}
	Let $F$ be an endofunctor  on $\Set$ with evaluation function $\ev_F$. If $F$ preserves weak pullbacks and $\ev_F$ satisfies \Cref{W1,W2} of \Cref{def:well-behaved} then for any pseudometric space $(X,d)$ the Wasserstein distance $\Wasserstein{F}{d}$ satisfies the triangle inequality.
\end{lemma}
\begin{proof}
	We will use the characterization of the triangle inequality given by \Cref{lem:alt-char-triangle} (\cpageref{lem:alt-char-triangle}). Hence, given any pseudometric space $(X,d)$ we just have shown that for every $t_1 \in FX$ the function $\Wasserstein{F}{d}(t_1,\_)\colon (FX, \Wasserstein{F}{d}) \nonexpansiveTo (\reals,d_e)$ is nonexpansive, i.e., that the inequality
	\begin{align}
		d_e\big(\Wasserstein{F}{d}(t_1,t_2),\Wasserstein{F}{d}(t_1,t_3)\big) \leq \Wasserstein{F}{d}(t_2,t_3)\label{eq:wasserstein-triangle1}
	\end{align}
	holds for all $t_2, t_3 \in FX$. We will show this in several steps.
	
	First of all we consider the case where no coupling of $t_2$ and $t_3$ exists. In this case the right hand side of \eqref{eq:wasserstein-triangle1} is $\top$ and it is easy to see that the the left hand side can never exceed that value because $\Wasserstein{F}{d}$ is non-negative. Thus in the remainder of the proof we only consider the case where $\Couplings{F}(t_2,t_3) \not = \emptyset$. 
	
	As next step we observe that if $\Couplings{F}(t_1,t_2)  = \Couplings{F}(t_1,t_3) = \emptyset$ the left hand side of \eqref{eq:wasserstein-triangle1} is $0$ and the right hand side is non-negative. Thus we are left with the cases where $\Couplings{F}(t_1,t_2) \not = \emptyset$ or $\Couplings{F}(t_1,t_3) \not = \emptyset$.
	
	Let us first assume that $\Couplings{F}(t_1,t_2) \not = \emptyset$ and recall that we required $\Couplings{F}(t_2,t_3) \not = \emptyset$. With the \nameref{lem:coupling} (\Cref{lem:coupling}) we can then conclude that also $\Couplings{F}(t_1,t_3) \not = \emptyset$. Similarly, if $\Couplings{F}(t_1,t_3) \not = \emptyset$ we can use the swap map as in the proof of \Cref{lem:wasserstein-symmetric} to see that $\Couplings{F}(t_3, t_1) \not = \emptyset$. As above the \nameref{lem:coupling} yields $\Couplings{F}(t_2,t_1) \not = \emptyset$ and again using the swap map we conclude that $\Couplings{F}(t_1,t_2) \not = \emptyset$. Thus the sole remaining case is the case where all couplings exists, i.e., we have $\Couplings{F}(t_1,t_2) \not = \emptyset$, $\Couplings{F}(t_2,t_3) \not = \emptyset$ and $\Couplings{F}(t_1,t_3) \not = \emptyset$.
	
	As intermediate step we recall that for all $x \in X$ the function $d(x,\_)$ is nonexpansive (see \Cref{lem:alt-char-triangle}, \cpageref{lem:alt-char-triangle}). Using the projections $\pi_i\colon X^3 \to X$ of the product this can be formulated as the inequality $d_e \circ (d \times d) \circ \big\langle \langle\pi_1, \pi_2\rangle, \langle\pi_1, \pi_3\rangle \big\rangle \leq d \circ \langle \pi_2, \pi_3\rangle$ and the monotonicity of $\EvaluationFunctor{F}$ (\Cref{W1}) implies that also the inequality \begin{align}
		\EvaluationFunctor{F}\big(d_e \circ (d \times d) \circ \big\langle \langle\pi_1, \pi_2\rangle, \langle\pi_1, \pi_3\rangle \big\rangle \big) \leq \EvaluationFunctor{F}(d \circ \langle \pi_2, \pi_3\rangle)\label{eq:wasserstein-triangle2}
	\end{align} holds. We will now use this inequality to prove \eqref{eq:wasserstein-triangle1} for the remaining case in which all couplings exists.
	
	As already pointed out before, for any $t_{12} \in \Couplings{F}(t_1, t_2)$ and $t_{23} \in \Couplings{F}(t_2,t_3)$ the \nameref{lem:coupling} (\Cref{lem:coupling}) yields a $t_{123} \in \Couplings{F}(t_1,t_2,t_3)$ and a coupling $t_{13} := F(\langle \pi_1, \pi_3\rangle)(t_{123}) \in \Couplings{F}(t_1,t_3)$. Plugging in $t_{123}$ in the inequality \eqref{eq:wasserstein-triangle2} above yields $\EvaluationFunctor{F}d_e\left(F(d\times d)(t_{12}, t_{13})\right) \leq \EvaluationFunctor{F}d(t_{23})$. Using well-behavedness (\Cref{W2}) of $\ev_F$ on the left hand side we obtain the following, intermediary result:
	\begin{align}
		d_e\Big(\EvaluationFunctor{F}d(t_{12}), \EvaluationFunctor{F}d(t_{13})\Big) \leq \EvaluationFunctor{F}d(t_{23}) \,.\label{eq:proof-wasserstein-triangle-intermediary-result}
	\end{align}
	If we define $d_{ij} := \Wasserstein{F}{d}(t_i,t_j)$ we can express \eqref{eq:wasserstein-triangle1} as $d_e(d_{12}, d_{13}) \leq d_{23}$. This is obviously true for $d_{12} = d_{13}$ so without loss of generality we assume $d_{12} < d_{13}$ and claim that for all $\epsilon > 0$ there is a coupling $t_{12} \in \Couplings{F}(t_1,t_2)$ such that for all couplings $t_{13} \in \Couplings{F}(t_1,t_3)$ we have
	\begin{align}
	d_e(d_{12}, d_{13}) \leq \epsilon + d_e\Big(\EvaluationFunctor{F}d(t_{12}), \EvaluationFunctor{F}d(t_{13})\Big)\,.\label{eq:proof-wasserstein-triangle}
	\end{align}
	To prove this claim, we recall that the Wasserstein distance is defined as an infimum, so we have $d_{13} \leq \EvaluationFunctor{F}d(t_{13})$ for all couplings $t_{13}$. Moreover, for the same reason we can pick, for every $\epsilon > 0$, a coupling $t_{12} \in \Couplings{F}(t_1,t_2)$, such that $\EvaluationFunctor{F}d(t_{12}) - d_{12} \leq \epsilon$ which can equivalently be stated as $\EvaluationFunctor{F}d(t_{12}) \leq d_{12} + \epsilon$. With this fixed coupling we now proceed to establish \eqref{eq:wasserstein-triangle2} for all $t_{13} \in \Couplings{F}(t_1,t_3)$.
	
	If $d_{13} = \infty$ we have $d_e(d_{12},d_{13}) = \infty$ but also $\EvaluationFunctor{F}d(t_{13}) = \infty$ and $\EvaluationFunctor{F}d(t_{12}) \leq d_{12}+ \epsilon < \infty + \epsilon = \infty$ and therefore $d_e\big(\,\EvaluationFunctor{F}d(t_{12}), \EvaluationFunctor{F}d(t_{13})\big) = \infty$ and thus \eqref{eq:proof-wasserstein-triangle} is valid. For $d_{13} < \infty$ we have 
	\begin{align*}
		d_e(d_{12}, d_{13}) &= d_{13} - d_{12} \leq \EvaluationFunctor{F}d(t_{13}) - \Big(\EvaluationFunctor{F}d(t_{12})-\epsilon\Big) = \epsilon + \Big(\EvaluationFunctor{F}d(t_{13}) - \EvaluationFunctor{F}d(t_{12})\Big)\\ 
		&\leq \epsilon +\big|\,\EvaluationFunctor{F}d(t_{13})- \EvaluationFunctor{F}d(t_{12})\big| \leq \epsilon + d_e\Big(\EvaluationFunctor{F}d(t_{12}), \EvaluationFunctor{F}d(t_{13})\Big)
	\end{align*}
	where the last inequality is due to the fact that $\EvaluationFunctor{F}d(t_{12}) < \infty$. Hence we have established our claimed validity of \eqref{eq:proof-wasserstein-triangle}. Using this, \eqref{eq:proof-wasserstein-triangle-intermediary-result} and the fact that -- as above -- given $\epsilon > 0$ we have a coupling $t_{23}$ such that $\EvaluationFunctor{F}d(t_{23}) \leq d_{23} + \epsilon$ we obtain the inequality
	\begin{align*}
	d_e(d_{12}, d_{13}) \leq \epsilon + d_e\Big(\EvaluationFunctor{F}d(t_{12}), \EvaluationFunctor{F}d(t_{13})\Big)\leq \epsilon + \EvaluationFunctor{F}d(t_{23}) \leq 2\epsilon + d_{23}
	\end{align*}
	which also proves $d_e(d_{12}, d_{13}) \leq d_{23}$. Indeed if $d_e(d_{12}, d_{13})> d_{23}$ then we would have $d_e(d_{12}, d_{13}) = d_{23}+\epsilon'$ and we just take $\epsilon < \epsilon'/2$ which yields the contradiction $	d_e(d_{12}, d_{13}) \leq 2\epsilon + d_{23} < \epsilon' + d_{23} = d_e(d_{12}, d_{13})$.
\end{proof}

Combining this result  with our previous considerations we finally obtain the desired result which guarantees that the Wasserstein distance is indeed a pseudometric. 

\begin{theorem}[Wasserstein Pseudometric]
	\label{prop:wasserstein-is-pseudometric}
Let $F$ be an endofunctor  on $\Set$ with evaluation function $\ev_F$. If $F$ preserves weak pullbacks and $\ev_F$ is well-behaved then for any pseudometric space $(X,d)$ the Wasserstein distance $\Wasserstein{F}{d}$ is a pseudometric.
\end{theorem}
\begin{proof}
	Reflexivity is given by \Cref{lem:wasserstein-reflexive}, symmetry by \Cref{lem:wasserstein-symmetric} and the triangle inequality by \Cref{lem:wasserstein-triangle-inequality}.
\end{proof}

With this result in place we can now finally study the Wasserstein lifting of a functor. Of course, our requirements on $F$ and $\ev_F$ are just sufficient conditions to prove that the Wasserstein distance is a pseudometric so it might be possible to give a more general definition. However, we will always work with weak pullback preserving functors and well-behaved evaluation functions so the following definition suffices.

\begin{definition}[Wasserstein Lifting] 
	Let $F$ be a weak pullback preserving endofunctor on $\Set$ with well-behaved evaluation function $\ev_F$. We define the \emph{Wasserstein lifting} of $F$ to be the functor $\LiftedFunctor{F}\colon \PMet\to\PMet$, $\LiftedFunctor{F}(X,d) = (FX,\Wasserstein{F}{d})$, $\LiftedFunctor{F}f = Ff$.
\end{definition}

Of course, we will have to check the functoriality. Its proof relies on \Cref{W1} of \Cref{def:well-behaved}, the monotonicity of $\EvaluationFunctor{F}$.

\begin{theorem}
	\label{prop:wasserstein-lifting-is-functorial}
	The Wasserstein lifting $\LiftedFunctor{F}$ is a functor on pseudometric spaces.
\end{theorem}
\begin{proof}
	$\LiftedFunctor{F}$ preserves identities and composition of arrows because $F$ does. Moreover, it preserves nonexpansive functions: Let $f\colon (X,d_X) \nonexpansiveTo (Y,d_Y)$ be nonexpansive and $t_1,t_2 \in FX$. Every $t \in \Couplings{F}(t_1,t_2)$ satisfies $Ff(t_i) = Ff\big(F\pi_i(t)\big) = F(f\circ \pi_i)(t) = F\big(\pi_i\circ (f\times f)\big)(t) = F\pi_i\big(F(f\times f)(t)\big)$. Hence we can calculate
	\begin{align}
	\Wasserstein{F}{d_X}(t_1,t_2) &= \inf \set{\EvaluationFunctor{F}d_X(t) \mid t \in \Couplings{F}(t_1,t_2)}\nonumber \\
	&\geq \inf \set{\EvaluationFunctor{F}d_X(t) \mid t \in F(X \times X),\ F\pi_i\big(F(f\times f)(t)\big) = Ff(t_i)}\label{eq:wassersteinLiftingIsFunctor:proof:1}\\
	&\geq \inf \set{\EvaluationFunctor{F}d_Y\big(F(f\times f)(t)\big) \mid t \in F(X \times X),\ F\pi_i\big(F(f\times f)(t)\big) = Ff(t_i)}\label{eq:wassersteinLiftingIsFunctor:proof:2}\\
	&\geq \inf \set{\EvaluationFunctor{F}d_Y(t') \mid t' \in \Gamma_F\big(Ff(t_1), Ff(t_2)\big)} =  \Wasserstein{F}{d_Y}\big(Ff(t_1),Ff(t_2)\big)\,.\label{eq:wassersteinLiftingIsFunctor:proof:3}
	\end{align}
	In this calculation the inequality \eqref{eq:wassersteinLiftingIsFunctor:proof:1} is due to our initial observation. Furthermore, \eqref{eq:wassersteinLiftingIsFunctor:proof:2} holds because $f$ is nonexpansive, i.e., $d_X \geq d_Y\circ (f \times f)$ and applying the monotonicity (\Cref{W1} of \Cref{def:well-behaved}) of $\EvaluationFunctor{F}$  yields $\EvaluationFunctor{F}d_X \geq \EvaluationFunctor{F}(d_Y \circ (f\times f)) = \EvaluationFunctor{F}d_Y \circ F(f\times f)$. The last inequality, \eqref{eq:wassersteinLiftingIsFunctor:proof:3}, is due to the fact that there might be more couplings $t'$ than those obtained via $F(f\times f)$.
\end{proof}

Let us now study the properties of the Wasserstein lifting. As was the case for the Kantorovich lifting, also the Wasserstein lifting preserves isometries.

\begin{theorem}
	\label{prop:wasserstein-preserves-isometries}
The Wasserstein lifting $\LiftedFunctor{F}$ of a functor $F$ preserves isometries. 
\end{theorem}
\begin{proof}
	Let $f\colon (X,d_X)\nonexpansiveTo (Y,d_Y)$ be an isometry. Since $\LiftedFunctor{F}$ is a functor, $\LiftedFunctor{F}f$ is nonexpansive, i.e., for all $t_1,t_2\in FX$ we have $\Wasserstein{F}{d_X}(t_1,t_2) \geq \Wasserstein{F}{d_Y}\big(Ff(t_1),Ff(t_2)\big)$. Now we show the opposite direction, i.e., that for all $t_1,t_2\in FX$ we have $\Wasserstein{F}{d_X}(t_1,t_2) \le \Wasserstein{F}{d_Y}\big(Ff(t_1),Ff(t_2)\big)$. 
	
	If $\Couplings{F}\big(Ff(t_1), Ff(t_2)\big) = \emptyset$ we have $\Wasserstein{F}{d_Y}\big((Ff(t_1),Ff(t_2)\big) = \top \geq \Wasserstein{F}{d_X}(t_1,t_2)$. Otherwise we will construct for each coupling $t \in \Couplings{F}\big(Ff(t_1), Ff(t_2)\big)$ a coupling $\gamma(t) \in \Gamma_F(t_1,t_2)$ such that $\EvaluationFunctor{F}d_X\big(\gamma(t)\big)=\EvaluationFunctor{F}d_Y(t)$ because then we have
	\begin{align*}
		\Wasserstein{F}{d_X}(t_1,t_2) = \inf_{t'\in \Couplings{F}(t_1,t_2)} \EvaluationFunctor{F}d_X(t') &\leq \inf_{t\in \Couplings{F}\big(Ff(t_1),Ff(t_2)\big)} \EvaluationFunctor{F}d_X\big(\gamma(t)\big)\\
		&=\inf_{t\in \Couplings{F}\big(Ff(t_1),Ff(t_2)\big)} \EvaluationFunctor{F}d_Y(t) = \Wasserstein{F}{d_Y}\big(Ff(t_1), Ff(t_2)\big)
	\end{align*}
	as desired. In this calculation the inequality is due to the fact that $\gamma(t) \in \Gamma_F(t_1,t_2)$ is a coupling and there might be other couplings which are not in the image of $\gamma$. 
	
	In order to construct $\gamma\colon \Gamma_F\big(Ff(t_1), Ff(t_2)\big) \to \Gamma_F(t_1,t_2)$, we consider the diagram below where $\pi_1 \colon X \times Y \to X$, $\pi_2 \colon Y \times X \to X$, and $\tau_i \colon Y \times Y \to Y$ are the respective projections of the products.
	
	\begin{center}\begin{diagram}
		\matrix(m)[matrix of math nodes, column sep=60pt, row sep=10pt]{
				&				& X \times X\\
				& X \times Y & 					& Y \times X\\
			X 	&				& Y \times Y		&				& X\\
				& Y			&					& Y\\
		};
		\draw[->] (m-1-3) edge node[above left]{$\id_X \times f$} (m-2-2);
		\draw[->] (m-1-3) edge node[above right]{$f \times \id_X$} (m-2-4);
		\draw[->] (m-2-2) edge node[below left]{$f \times \id_Y$} (m-3-3);
		\draw[->] (m-2-4) edge node[below right]{$\id_Y \times f$} (m-3-3);
		\draw[->] (m-2-2) edge node[above left]{$\pi_1$} (m-3-1);
		\draw[->] (m-2-4) edge node[above right]{$\pi_2$} (m-3-5);
		\draw[->] (m-3-1) edge node[below left]{$f$} (m-4-2);
		\draw[->] (m-3-3) edge node[below right]{$\tau_1$} (m-4-2);
		
		\draw[->] (m-3-3) edge node[below left]{$\tau_2$} (m-4-4);
		\draw[->] (m-3-5) edge node[below right]{$f$} (m-4-4);
	\end{diagram}\end{center}
	
	\noindent This diagram consists of pullbacks: it is easy to
        check that the diagram commutes. The unique mediating arrows
        are constructed as follows\new{: in every item we assume that
          $P$ is another object -- with a commuting square -- and we
          construct the mediating morphism to the pullback object.}
	\begin{itemize}
		\item For the lower left part let $P$ be a set with $p_1\colon P \to X$, $p_2 \colon P \to Y \times Y$ such that $f \circ p_1 = \tau_1 \circ p_2$, then define $u\colon P \to X \times Y$ as $u :=\langle p_1, \tau_2 \circ p_2\rangle$. 
		\item Analogously, for the lower right part let $P$ be a set with $p_1\colon P \to Y \times Y$, $p_2 \colon P \to X$ such that $\tau_2 \circ p_1 = f\circ p_2$, then define $u\colon P \to Y \times X$ as $u := \langle\tau_1 \circ p_1, p_2\rangle$.
		\item Finally, for the upper part let $P$ be a set with $p_1\colon P \to X \times Y$, $p_2 \colon P \to Y \times X$ such that $(f \times \id_Y) \circ p_1 = (\id_Y \times f) \circ p_2$, then define $u\colon P \to X \times X$ as $u := \langle\pi_1 \circ p_1, \pi_2 \circ p_2\rangle$.
	\end{itemize}
	
	\noindent We apply the weak pullback preserving functor $F$ to the diagram and obtain the following diagram which hence consists of three weak pullbacks.
	
	\begin{center}\begin{diagram}
		\matrix(m)[matrix of math nodes, column sep=40pt, row sep=10pt]{
				&				& F(X \times X)\\
				& F(X \times Y) & 					& F(Y \times X)\\
			FX 	&				& F(Y \times Y)		&				& FX\\
				& FY			&					& FY\\
		};
		\draw[->] (m-1-3) edge node[above left]{$F(\id_X \times f)$} (m-2-2);
		\draw[->] (m-1-3) edge node[above right]{$F(f \times \id_X)$} (m-2-4);
		\draw[->] (m-2-2) edge node[below left]{$F(f \times \id_Y)$} (m-3-3);
		\draw[->] (m-2-4) edge node[below right]{$F(\id_Y \times f)$} (m-3-3);
		\draw[->] (m-2-2) edge node[above left]{$F\pi_1$} (m-3-1);
		\draw[->] (m-2-4) edge node[above right]{$F\pi_2$} (m-3-5);
		\draw[->] (m-3-1) edge node[below left]{$Ff$} (m-4-2);
		\draw[->] (m-3-3) edge node[below right]{$F\tau_1$} (m-4-2);
		
		\draw[->] (m-3-3) edge node[below left]{$F\tau_2$} (m-4-4);
		\draw[->] (m-3-5) edge node[below right]{$Ff$} (m-4-4);
	\end{diagram}\end{center}
	
	\noindent Given a coupling $t \in \Gamma_F\big(Ff(t_1), Ff(t_2)\big) \subseteq F(Y \times Y)$ we know $F\tau_i(t) = Ff(t_i) \in FY$. 
	
	Since the lower left square in the diagram is a weak pullback, we obtain an element $s_1\in F(X\times Y)$ with $F\pi_1(s_1) = t_1$ and $F(f\times \id_Y)(s_1) = t$. Similarly, from the lower right square, we obtain $s_2\in F(Y\times X)$ with $F\pi_2(s_2) = t_2$ and $F(\id_Y\times f)(s_2) = t$. Again by the weak pullback property we obtain our $\gamma(t) \in F(X\times X)$ with $F(\id_X\times f)\big(\gamma(t)\big) = s_1$, $F(f\times \id_X)\big(\gamma(t)\big) = s_2$.
	
	We convince ourselves that $\gamma(t)$ is indeed a coupling of $t_1$ and $t_2$: Let $\pi_i'\colon X \times X \to X$ be the missing projections. Using $\pi_1 \circ (\id_X \times f) = \pi_1'$ we have $F\pi_1'\big(\gamma(t)\big) = F\big(\pi_1\circ (\id_X\times f)\big)\big(\gamma(t)\big) = F \pi_1\circ F(\id_X\times f)\big(\gamma(t)\big) = F\pi_1(s_1) = t_1$
	and analogously $F\pi_2'\big(\gamma(t)\big) = F\big(\pi_2\circ (f\times\id_X)\big)\big(\gamma(t)\big) = F\pi_2\circ F(f\times\id_X)\big(\gamma(t)\big) = F\pi_2(s_2) = t_2$. Moreover, using $f \times f = (f \times \id_Y) \circ (\id_X \times f)$ and functoriality we have 
	\begin{align*}
		F(f\times f)\big(\gamma(t)\big) &= F\big((f\times \id_Y)\circ (\id_X\times f)\big)\big(\gamma(t)\big) \\
		&= F(f\times \id_Y)\circ F(\id_X\times f)\big(\gamma(t)\big) = F(f\times \id_Y)(s_1) = t
	\end{align*} and thus $\EvaluationFunctor{F}d_Y(t) = \EvaluationFunctor{F}d_Y(F(f\times f)(\gamma(t))) = \EvaluationFunctor{F}(d_Y\circ (f\times f))({\gamma(t)}) = \EvaluationFunctor{F}(d_X({\gamma(t)}))$ as desired. Note that the last equality is due to the fact that $f$ is an isometry.
\end{proof}

In contrast to the Kantorovich lifting, we can prove that metrics are preserved by the Wasserstein lifting in certain situations.

\begin{theorem}[Preservation of Metrics]
	\label{prop:preservation}
	Let $F$ be a weak pullback preserving endofunctor on $\Set$ with well-behaved evaluation function $\ev_F$ and $(X,d)$ be a metric space. If for all $t_1,t_2 \in FX$ where $\Wasserstein{F}{d}(t_1,t_2) = 0$ there is an optimal $F$-coupling $\gamma(t_1,t_2) \in \Couplings{F}(t_1,t_2)$ such that $0=\Wasserstein{F}{d}(t_1,t_2)=\EvaluationFunctor{F}d\big(\gamma(t_1,t_2)\big)$ then  $\Wasserstein{F}{d}$ is a metric and thus $\LiftedFunctor{F}(X,d) = (FX,\Wasserstein{F}{d})$ is a metric space.
\end{theorem}
\begin{proof}
	Let $(X,d)$ be a metric space. By \Cref{prop:wasserstein-is-pseudometric} $\Wasserstein{F}{d}$ is a pseudometric. Thus we just have to show that for any $t_1,t_2\in FX$ the fact that $\Wasserstein{F}{d}(t_1,t_2) = 0$ implies $t_1 = t_2$. 
	
	Since $d$ is a metric the preimage $d^{-1}[\set{0}]$ is the set $\Delta_X = \{(x,x)\mid x\in X\}$. Hence the square on the left below is a pullback and adding the projections yields $\pi_1\circ e =\pi_2\circ e$ where $e\colon \Delta_X \hookrightarrow X\times X$ is the inclusion. Furthermore, by \Cref{lem:W3-weak-pb} we know that due to \Cref{W3} of \Cref{def:well-behaved} the square on the right is a weak pullback.
	
	\begin{center}\begin{diagram}
		\matrix(m)[matrix of math nodes, column sep=30pt, row sep=10pt]{
				& \Delta_X 		& \set{0} 	& & F\set{0} & \set{0}\\
			X 	& X \times X 	& \reals	& & F\reals & \reals\\
		};
		\draw[->] (m-1-2) edge node[above]{$!_{\Delta_X}$} (m-1-3);
		\draw[right hook->] (m-1-2) edge node[left]{$e$} (m-2-2);
		\draw[right hook->] (m-1-3) edge node[right]{$i$} (m-2-3);
		\draw[->] (m-2-2) edge node[above]{$d$} (m-2-3);
		\draw[->, transform canvas={yshift=2pt}] (m-2-2) edge node[above]{$\pi_1$} (m-2-1);
		\draw[->, transform canvas={yshift=-2pt}] (m-2-2) edge node[below]{$\pi_2$} (m-2-1);
		
		% second diagram
		\draw[->] (m-1-5) edge node[above]{$!_{F\set{0}}$} (m-1-6);
		\draw[->] (m-1-5) edge node[left]{$Fi$} (m-2-5);
		\draw[right hook->] (m-1-6) edge node[right]{$i$} (m-2-6);
		\draw[->] (m-2-5) edge node[above]{$\ev_F$} (m-2-6);
	\end{diagram}\end{center}
	
	\noindent Since $F$  preserves weak pullbacks, applying it to the first diagram yields a weak pullback. By combining this diagram with the right diagram we obtain the diagram below where the outer rectangle is again a weak pullback. 
	\begin{center}\begin{diagram}
		\matrix(m)[matrix of math nodes, column sep=40pt, row sep=10pt]{
				& F\Delta_X 		& F\set{0} 	& \set{0}\\
			FX 	& F(X \times X) 	& F\reals	& \reals\\
		};
		\draw[->] (m-1-2) edge node[above]{$F!_{\Delta_X}$} (m-1-3);
		\draw[->] (m-1-2) edge node[left]{$Fe$} (m-2-2);
		\draw[->] (m-1-3) edge node[right]{$Fi$} (m-2-3);
		\draw[->] (m-2-2) edge node[above]{$Fd$} (m-2-3);
		\draw[->, transform canvas={yshift=2pt}] (m-2-2) edge node[above]{$F\pi_1$} (m-2-1);
		\draw[->, transform canvas={yshift=-2pt}] (m-2-2) edge node[below]{$F\pi_2$} (m-2-1);
	
		\draw[->] (m-1-3) edge node[above]{$!_{F\set{0}}$} (m-1-4);
		\draw[right hook->] (m-1-4) edge node[right]{$i$} (m-2-4);
		\draw[->] (m-2-3) edge node[above]{$\ev_F$} (m-2-4);
		
		\draw[->] (m-2-2) edge[bend right=10] node[below]{$\EvaluationFunctor{F}d$} (m-2-4);
	\end{diagram}\end{center}
	
	\noindent Let $t := \gamma(t_1,t_2) \in F(X\times X)$, i.e., $\Wasserstein{F}{d}(t_1,t_2) = \EvaluationFunctor{F}d(t) = 0$. Since we have a weak pullback, there exists $t'\in F\Delta_X$ with $Fe(t') = t$. (Since $Fe$ is an embedding, $t'$ and $t$ actually coincide.) This implies that $t_1 = F\pi_1(t) = F\pi_1\big(Fe(t')\big) = F\pi_2\big(Fe(t')\big) = F\pi_2(t) = t_2$.
\end{proof}

Apparently, \Cref{prop:preservation} admits the following simple corollary. 

\begin{corollary}
	Let $F$ be a weak pullback preserving endofunctor on $\Set$ with well-behaved evaluation function $\ev_F$ and $(X,d)$ be a metric space. If the infimum in \eqref{eq:wasserstein} is always a minimum then $\Wasserstein{F}{d}$ is a metric and thus $\LiftedFunctor{F}(X,d) = (FX,\Wasserstein{F}{d})$ is a metric space.
\end{corollary}

Please note that a similar restriction for the Kantorovich lifting (i.e., requiring that the supremum in \Cref{def:kantorovich} is a maximum) does \emph{not} yield preservation of metrics: In \Cref{exa:kant-no-metric} the supremum is always a maximum but we do not get a metric. Let us now compare both lifting approaches.

\begin{lemma}
	\label{prop:wasserstein-upper-bound}
	Let $F$ be an endofunctor on $\Set$ with evaluation function $\ev_F$ and $(X,d)$ be a pseudometric space. If $\ev_F$ satisfies \Cref{W1,W2} of \Cref{def:well-behaved} then for all $t_1, t_2 \in FX$, all $t\in\Couplings{F}(t_1,t_2)$ and all nonexpansive functions $f\colon (X,d)\nonexpansiveTo (\reals,d_e)$ we have $d_e\big(\EvaluationFunctor{F}f(t_1),\EvaluationFunctor{F}f(t_2)\big) \leq \EvaluationFunctor{F}d(t)$.
\end{lemma}
\begin{proof}
We have $d_e\circ (f\times f) \le d$ since $f$ is nonexpansive. Due to monotonicity of the evaluation functor (\Cref{W1})  we obtain $  \EvaluationFunctor{F}d_e\circ F(f\times f) = \EvaluationFunctor{F}\big(d_e\circ (f\times f)\big) \le \EvaluationFunctor{F}d$. Furthermore:
\begin{align*}
	&d_e\big(\EvaluationFunctor{F}f(t_1),\EvaluationFunctor{F}f(t_2)\big) 	=	d_e\Big(\EvaluationFunctor{F}f\big(F\pi_1(t)\big),\EvaluationFunctor{F}f\big(F\pi_2(t)\big)\Big) =  d_e\Big(\EvaluationFunctor{F}(f\circ\pi_1)(t),\EvaluationFunctor{F}(f\circ\pi_2)(t)\Big) \\
							&\ =	d_e\Big(\EvaluationFunctor{F}\big(\pi_1\circ (f\times f)\big)(t),\EvaluationFunctor{F}\big(\pi_2\circ (f\times f)\big)(t)\Big)  =  d_e\Big(\EvaluationFunctor{F}\pi_1\big(F(f\times f)(t)\big),\EvaluationFunctor{F}\pi_2\big(F(f\times f)(t)\big)\Big) \\
							&\ \le \EvaluationFunctor{F}d_e\big(F(f\times f)(t)\big) = \EvaluationFunctor{F}\big(d_e\circ (f\times f)\big)(t) \le \EvaluationFunctor{F}d(t)
\end{align*}
where the first inequality is due to \Cref{W2}, the second due to the above observation which was based on \Cref{W1}.
\end{proof}

Using this result we can see that under certain conditions the Wasserstein distance is an upper bound for the Kantorovich pseudometric.

\begin{theorem}[Comparison of the two Liftings]
	\label{prop:wasserstein-vs-kantorovich}
	Let  $F$  be an endofunctor on $\Set$. If $\ev_F$ satisfies \Cref{W1,W2} of \Cref{def:well-behaved} then for all pseudometric spaces $(X,d)$ we have $\Kantorovich{F}{d} \leq \Wasserstein{F}{d}$.
\end{theorem}
\begin{proof}
	Let $t_1,t_2 \in FX$. We know (see the discussion after \Cref{def:kantorovich}) that a nonexpansive function $f \colon (X,d) \nonexpansiveTo ([0,\top],d_e)$ always exists but (see the discussion after \Cref{def:wasserstein}) couplings do not have to exist, so we have to distinguish two cases. If $\Couplings{F}(t_1,t_2) = \emptyset$ we have $\Wasserstein{F}{d}(t_1,t_2) = \top$ and clearly $\Kantorovich{F}{d} \leq \top$. Otherwise we can apply \Cref{prop:wasserstein-upper-bound} to obtain the desired inequality.
\end{proof}

This inequality can be strict as the following example shows.

\begin{example}[Wasserstein Lifting of the Squaring Functor]
	\label{exa:kant-no-metric2}
	It is easy to see that the squaring functor $S\colon \Set \to \Set$, $SX = X \times X$, $Sf = f \times f$ preserves weak pullbacks. Some simple calculations show that the evaluation function $\ev_S\colon S[0,\infty] \to [0,\infty], \ev_S(r_1, r_2) = r_1+r_2$ given in \Cref{exa:squaring-functor} is well-behaved \cite[Ex.~5.4.28]{Ker16}. We now continue \Cref{exa:squaring-functor} where we considered a metric space $(X,d)$ with at least two elements, chose an element  $t_1=(x_1,x_2) \in SX=X \times X$ with $x_1 \not = x_2$ and defined $t_2=(x_2,x_1)$. The unique coupling $t \in \Gamma_S(t_1,t_2)$ is $t=\big((x_1, x_2), (x_2, x_1)\big)$. Using that $d$ is a metric we conclude that $\Wasserstein{S}{d}(t_1, t_2) = \EvaluationFunctor{S}d(t) = d(x_1,x_2) + d(x_2,x_1) = 2d(x_1,x_2) > 0$. However, in \Cref{exa:squaring-functor} we calculated $ \Kantorovich{S}{d}(t_1,t_2) = 0$.
\end{example}

Whenever the inequality in \Cref{prop:wasserstein-vs-kantorovich} can be replaced by an equality we will in the following say that the \emph{Kantorovich-Rubinstein duality}\index{Kantorovich-Rubinstein duality} or simply \emph{duality}\index{duality (Kantorovich Rubinstein)} holds. In this case we obtain a canonical notion of distance on $FX$ for any given pseudometric space $(X,d)$. 

To prove that the duality holds and simultaneously to calculate the distance of $t_1,t_2\in FX$ it is sufficient to find a nonexpansive function $f\colon (X,d)\nonexpansiveTo (\reals,d_e)$ and a coupling $t \in \Gamma_F(t_1,t_2)$ such that $d_e\big(\EvaluationFunctor{F}f(t_1),\EvaluationFunctor{F}f(t_2)\big) = \EvaluationFunctor{F}d_e(t)$. Due to \Cref{prop:wasserstein-vs-kantorovich} this value equals $\Kantorovich{F}{d}(t_1,t_2) = \Wasserstein{F}{d}(t_1,t_2)$. We will often employ this technique for the upcoming examples.

\begin{example}[Duality for the Identity Functor] 
	\label{ex:kant-rubinst}
	We consider the identity functor $\Id$ with the identity function as evaluation function, i.e., $\ev_\Id = \id_\reals$. For any $t_1,t_2\in X$, $t := (t_1, t_2)$ is the unique coupling of $t_1, t_2$. Hence, $\Wasserstein{F}{d}(t_1, t_2) = d(t_1,t_2)$. With the function $d(t_1,\_) \colon (X,d) \nonexpansiveTo (\reals,d_e)$ which is nonexpansive due to \Cref{lem:alt-char-triangle} we obtain duality because we have $d(t_1,t_2) = d_e\big(d(t_1, t_1), d(t_1,t_2)\big) \leq \Kantorovich{F}{d}(t_1,t_2) \leq \Wasserstein{F}{d}(t_1,t_2) = d(t_1,t_2)$ and thus equality. Similarly, if we define $\ev_\Id(r) = c\cdot r$ for $r\in \reals$, $0<c\leq1$, the Kantorovich and Wasserstein liftings coincide and we obtain the discounted distance $\Kantorovich{F}{d}(t_1,t_2) = \Wasserstein{F}{d}(t_1,t_2) = c\cdot d(t_1,t_2)$.
\end{example}

\begin{example}[Duality for the Distribution Functors]
	\label{exa:probability-distribution-functor2}
	It is known that the probability distribution functor $\Distributions$ and its variants preserve weak pullbacks \cite[Proposition~3.3]{Sok11}. It is easy show that the evaluation function $\ev_\Distributions\colon \Distributions [0,1] \to [0,1]$, $\ev_\Distributions (P) = \sum_{x\in [0,1]} x\cdot P(x)$ which we have defined in \Cref{exa:probability-distribution-functor} is well-behaved for all variants of the distribution functor (i.e., distributions or subdistributions with countable or finite supports). In order to do so, one just has to use the definition of the evaluation functor and the triangle inequality for the absolute value. We omit this proof \cite[Ex.~5.4.30]{Ker16} here and just 
	observe that we recover the usual Wasserstein pseudometric, i.e., for any (sub)probability distributions $P_1,P_2\colon X \to [0,1]$ we have
	\begin{align*}
		\Wasserstein{\Distributions }{d}(P_1,P_2)= \inf \set{\sum_{x_1,x_2 \in X} d(x_1,x_2) \cdot P(x_1, x_2) \mid P \in \Couplings{\Distributions }(P_1, P_2)}
	\end{align*}
	and -- for proper distributions only -- the Kantorovich-Rubinstein duality \cite{Vil09} from transportation theory for the discrete case. Moreover, in this case it is known (and easy to see using a linear program formulation) that for finite supports the above infimum is always a minimum. In the case of subdistributions we do not have duality: Let $P, Q \colon \one \to [0,1]$ be subdistributions on the singleton set $\one$, i.e., $P(\checkmark) = p$ and $Q(\checkmark) = q$ with $p,q \in [0,1]$. The only pseudometric on $\one$  is the discrete metric $d$ so any function $f \colon \one \to [0,1]$ is nonexpansive and we have $\EvaluationFunctor{\Distributions}f(P) = f(\checkmark)\cdot P(\checkmark)$.  Hence the Kantorovich distance of $P$ and $Q$ is achieved for the function $f$ where $f(\checkmark) = 1$ and equals $\Kantorovich{F}{d}(P,Q) = |p-q|$. However, if $p\not = q$ it is easy to see that there are no couplings of $P$ and $Q$ so $\Wasserstein{F}{d}(P,Q) = 1$. Thus for any $p,q$ where $|p - q| < 1$ we do not have equality. 
\end{example}

\begin{example}[The Hausdorff Pseudometric for Finite Sets]
	\label{exa:hausdorff}
	\index{Hausdorff pseudometric}
	Similar to \Cref{ex:evaluation-max} we assume $\top = \infty$ but here we just consider the finite powerset functor $\PowersetFinite$ with evaluation function $\max\colon \PowersetFinite(\prealinf) \to \prealinf$ with $\max \emptyset = 0$ and $\min \emptyset = \infty$. We claim that in this setting we obtain duality and both pseudometrics are equal to the \emph{Hausdorff pseudometric} $d_H$ on $\PowersetFinite(X)$ which is defined as, for all $X_1,X_2\in \PowersetFinite X$,
	\begin{align*}
		 d_H(X_1,X_2) = \max\left\{\max_{x_1\in X_1} \min_{x_2\in X_2} d(x_1,x_2), \max_{x_2\in X_2} \min_{x_1\in X_1} d(x_1,x_2) \right\}\,.
	\end{align*}
	Note that this distance is $\infty$, if either $X_1$ or $X_2$ is empty.
	
	We show our claim by proving that if $X_1,X_2$ are both non-empty there exists a coupling and a nonexpansive function that both witness the Hausdorff distance. Assume that the first value $\max_{x_1\in X_1} \min_{x_2\in X_2} d(x_1,x_2)$ is maximal and assume that $y_1\in X_1$ is the element of $X_1$ for which the maximum is reached. Furthermore let $y_2\in X_2$ be the closest element in $X_2$, i.e.,  the element for which $d(y_1,y_2)$ is minimal. We know that for all $x_1\in X_1$ there exists $x_2^{x_1}$ such that $d(x_1,x_2^{x_1}) \le d(y_1,y_2)$ and for all $x_2\in X_2$ there exists $x_1^{x_2}$ such that $d(x_1^{x_2},x_2) \le d(y_1,y_2)$. Specifically, $x_2^{y_1} = y_2$. We use the coupling $T\subseteq X\times X$ with 
	\begin{align*}
		T = \set{(x_1,x_2^{x_1})\mid x_1\in X_1} \cup \set{(x_1^{x_2},x_2)\mid x_2\in X_2}\,.
	\end{align*}
	Indeed, we obviously have $\PowersetFinite \pi_i (T) = X_i$ and $\PowersetFinite d(T)$ contains all distances between the elements above, of which the distance $d(y_1,y_2) = d^H(X_1,X_2)$ is maximal.
	We define a nonexpansive function $f\colon (X,d)\to (\reals,d_e)$ as $f(x) = \min_{x_2\in X_2} d(x,x_2)$. Then we have
	\begin{align*}
		\max \PowersetFinite f(X_1) = \max f[X_1] = \max_{x_1\in X_1} \min_{x_2\in X_2} d(x_1,x_2) = d^H(X_1,X_2)
	\end{align*} 
	and $\max \PowersetFinite f(X_2) = \max f[X_2] = 0$. Hence, the difference of both values is $d^H(X_1,X_2)$. It remains to show that $f$ is nonexpansive. Let $x,y\in X$ and let $x_2,y_2\in X_2$ be elements for which the distances $d(x,x_2),d(y,y_2)$ are minimal. Hence $d(x,x_2) \le d(x,y_2) \le d(x,y) + d(y,y_2)$ and $d(y,y_2) \le d(y,x_2) \le d(y,x) + d(x,x_2)$. \Cref{lem:sum-vs-dist} implies that $d(x,y) \ge d_e\big(d(x,x_2),d(y,y_2)\big) = d_e\big(f(x),f(y)\big)$.
	
	If $X_1=X_2=\emptyset$, we can use the coupling $T = \emptyset=\emptyset\times\emptyset$ and any function $f$. If, instead $X_1=\emptyset$, $X_2\neq\emptyset$, no coupling exists thus $\Wasserstein{F}{d} = \infty$ and we can take the constant $\infty$-function to show that also $\Kantorovich{F}{d} = \infty$ is attained.
\end{example}

It would also be interesting to consider the general or countable
powerset functor, use the supremum as (well-behaved) evaluation
function and consider the resulting liftings. By generalizing the
argument and finding suitable approximations for the infimum and
supremum, one can again show that the Wasserstein as well as the
Kantorovich lifting coincide with the Hausdorff metric (with
supremum/infimum replacing maximum/minimum).

On the other hand, we can argue that in this case optimal couplings do
not necessarily exist, because the Hausdorff pseudometric for
countable sets does not preserve metrics. If we take the Euclidean
metric and consider the sets
$X _1= \set{0} \cup \set{1/n \mid n \in \N}$ and
$X_2 = \set{1/n\mid n \in \N}$ then their Hausdorff distance is $0$
although $X_1$ and $X_2$ are different. Thus, due to
\Cref{prop:preservation}, there cannot be an optimal coupling for
$X_1$ and $X_2$.

As another example of our lifting approaches, we consider the input
functor $\_^A$ on $\Set$ where $A$ is an arbitrary set. It maps each
set $X$ to the set $X^A$ of functions with domain $A$ and codomain $X$
and each function $f \colon X \to Y$ to the function
$f^A \colon X^A \to Y^A$ where $f^A(g\colon A \to X) = f \circ
g$. \new{\name[Jan]{Rutten} showed in \cite{Rut00} that functors on $\Set$ which
  preserve pullbacks also preserve weak pullbacks, which applies to
  the input functor.} We will now present different well-behaved
evaluation functions for this functor and the resulting Wasserstein
pseudometrics.

\begin{example}[Wasserstein Lifting for the Input Functor]
	\label{exa:wasserstein-input-functor}
	\index{input functor}
	We consider the input functor $\_^A $ with finite input set
        $A$ and claim that the evaluation functions $\ev_F\colon
        [0,\top]^A \to [0,\top]$ which are listed in the table below
        are well-behaved and yield the given Wasserstein pseudometric
        on $X^A$ for any pseudometric space $(X,d)$. \new{Note that in
          the third case it is necessary to normalize with $|A|^{-1}$
          so that the distance does not exceed $\top$.}
	\begin{center}\begin{tabular}{c|c|c}
		maximal distance $\top$ & $\ev_F(s)$ & $\Wasserstein{F}{d}(s_1,s_2)$ \\
		\hline
		$\top \in\,]0,\infty]$ & $\max\limits_{a \in A} s(a)$ & $\max\limits_{a \in A}d\big(s_1(a), s_2(a)\big)$\\
		$\top=\infty$ & $\sum\limits_{a \in A} s(a)$ &$\sum\limits_{a \in A}d\big(s_1(a), s_2(a)\big)$\\
		$\top \in\,]0,\infty[$ & $|A|^{-1}\sum\limits_{a \in A} s(a)$ &$|A|^{-1}\sum\limits_{a \in A}d\big(s_1(a), s_2(a)\big)$
	\end{tabular}\end{center}
	In order to show this we first observe that for any $f\colon X \to [0,\top]$ we have $\EvaluationFunctor{F}f = \ev_F \circ f^A$ so applying it to $s\in X^A$ yields either $\max_{a \in A} f\big(s(a)\big)$ or $\sum_{a \in A} f\big(s(a)\big)$ or $|A|^{-1}\sum_{a \in A} f\big(s(a)\big)$. With this we proceed to show well-behavedness.
	\begin{wbconditions}
		\item For $f_1,f_2  \colon X \to [0,\top]$ with $f_1 \leq f_2$ we obviously also have $\EvaluationFunctor{F}f_1 \leq \EvaluationFunctor{F}f_2$.
		\item Let $s \in ([0,\top]^2)^A$ and $s_i := \pi_i^A(t)$, i.e., necessarily $s = \langle s_1, s_2\rangle$. We have to show the inequality
		$d_e\big(\ev_F(s_1), \ev_F(s_2)\big) \leq \EvaluationFunctor{F}d_e(s)$ where the right hand side evaluates to $\ev_F \big(d_e^A(s)\big) = \ev_F (d_e \circ s) = \ev_F(d_e \circ \langle s_1,s_2\rangle)$. Using this we can see that the inequality is an immediate consequence of \Cref{lem:max-sum} (\cpageref{lem:max-sum}) by taking $f=s_1$, $g=s_2$ or, in the last case, $f = |A|^{-1}s_1$ and $g=|A|^{-1}s_2$.
		\item We have $\ev_F^{-1}[\set{0}] = \set{s \colon A \to [0,\top] \mid \ev_F(s) =0}$. Clearly for all functions this is the case only if $s$ is the constant $0$-function. Since $\set{0}$ is a final object in $\Set$, there is a unique function $z \colon A \to \set{0}$. Thus $Fi[F\set{0}] = i^A[\set{0}^A] = \set{i^A(z)} = \set{i \circ z}$ and clearly $i \circ z \colon A \to [0,\top]$ is also the constant $0$-function. 
	\end{wbconditions}
	Now if we have $s_1,s_2 \in X^A$ their unique coupling is $s:=\langle s_1,s_2\rangle \colon A \to X \times X$. Moreover $\EvaluationFunctor{F}d(s) = \ev_F \big( d^A(s)\big) = \ev_F\big(d\circ \langle s_1,s_2\rangle\big)$ and using the different evaluation functions we obtain the pseudometrics given in the table above.
\end{example}

We conclude our list of examples with the Wasserstein lifting of the machine functor which we will use several times in the remainder of this paper.

\begin{example}[Wasserstein Lifting of the Machine Functor]
  \label{exa:wasserstein-machine}
  We equip the machine functor $M_B = B \times \_^A$ with the
  evaluation function
  $\ev_{M_B} \colon B \times [0,\top]^A\to [0, \top]$,
  $(o,s) \mapsto c \cdot \ev_I(s)$ where $c \in \,]0,1]$ is a discount
  factor and and $\ev_I$ is one of the evaluation functions for the
  input functor from \Cref{exa:wasserstein-input-functor}. For any
  pseudometric space $(X,d)$ we can easily see that for two elements
  $(o_1,s_1), (o_2,s_2) \in B \times X^A$ we have a unique coupling if
  and only if $o_1 = o_2$, namely $(o_1,\langle s_1,s_2\rangle)$ (for
  $o_1 \not = o_2$ no coupling exists at all). Thus the Wasserstein
  distance of any two elements as above is $1$ if $o_1 \not = o_2$ and
  $c \cdot \ev_I (d \circ \langle s_1,s_2 \rangle)$ otherwise.

        % The latter is either
        % $c \cdot \max_{a \in A}d\big(s_1(a), s_2(a)\big)$ or
        % $\sum_{a \in A}d\big(s_1(a), s_2(a)\big)$ depending on the
        % choice of $\ev_I$.
\end{example}

\subsection{Lifting Multifunctors}
\label{sec:multifunctors}
While the functors we considered so far can be nicely lifted using our theory, there are other functors that require a more general treatment. For instance, consider the output functor $F = B\times \_$ for some fixed set $B$. As in \Cref{exa:wasserstein-machine} we have a coupling for $t_1,t_2 \in FX = B \times X$ with $t_i=(b_i,x_i)$ if and only if $b_1 = b_2$. Consequently, if $b_1 \neq b_2$ then irrespective of the evaluation function we choose and of the distance between $x_1$ and $x_2$ in $(X, d)$, the lifted Wasserstein pseudometric will always result in $ \Wasserstein{F}{d}(t_1,t_2) =\top$.  This can be counterintuitive, e.g., taking $B=[0,1]$, $X \not=\emptyset$ and $t_1 = (0,x)$ and $t_2=(\epsilon,x)$ for a small $\epsilon > 0$ and an $x \in X$.
The reason is that we think of $B = [0,1]$ as if it were endowed with a non-discrete pseudometric, like e.g. the Euclidean metric $d_e$, plugged into the product \emph{after} the lifting. This intuition can be formalized by considering the lifting of the product seen as a functor from $\Set \times \Set$ into $\Set$. More generally, it can be seen that the definitions and results introduced so far for endofunctors in $\Set$ extend to multifunctors on $\Set$, i.e., to functors $F\colon \Set^n \to \Set$ on the product category $\Set^n$ for any natural number $n \in \N$. The only difference is that we start with $n$ pseudometric spaces instead of one. Due to this, the definitions and results are technically more complicated than in the endofunctor setting but they capture exactly the same ideas as before. 

For clarity we provide some of the multifunctor results here but it is safe to skip the results at a first read, continue with studying the product and coproduct bifunctors in \Cref{subsec:product-coproduct-bifunctors} (\cpageref{subsec:product-coproduct-bifunctors}) -- they will play an important role for the later development of our theory -- and only look at the exact multifunctor definitions when necessary.

The formal definition of multifunctor lifting is a straightforward extension of \Cref{def:liftingPMet} with only a little bit of added technical complexity.

\begin{definition}[Lifting of a Multifunctor]
	\label{def:liftingPMet-multi}
	\index{lifting (of a multifunctor)}
	Let $U\colon \PMet \to \Set$ be the forgetful functor which maps every pseudometric space to its underlying set and denote by $U^n\colon \PMet^n \to \Set^n$ the n-fold product of $U$ with itself. A functor $\LiftedFunctor{F}\colon\PMet^n\to\PMet$ is called a \emph{lifting} of a functor $F\colon \Set^n \to \Set$ if it satisfies $U\circ \LiftedFunctor{F} = F\circ U^n$.
\end{definition}

These multifunctor liftings can be used to obtain endofunctor liftings by fixing all but one parameter \cite[Lem.~5.4.37]{Ker16}. As in the endofunctor case any multifunctor lifting is monotone in the following sense: If we have pseudometrics $d_i \leq e_i$ on common sets $X_i$ we also have $\LiftedMetric{F}{(d_1,\dots,d_n)} \leq \LiftedMetric{F}{(e_1,\dots,e_n)}$ where $\LiftedMetric{F}{(d_1,\dots,d_n)}$ and $\LiftedMetric{F}{(e_1,\dots,e_n)}$ denote the pseudometrics on $F(X_1,\dots,X_n)$ which we obtain by applying $\LiftedFunctor{F}$ to $(X_1,d_1),\dots,(X_n,d_n)$ or $(X_1,e_1),\dots,(X_n,e_n)$ respectively. We omit the proof of this property since the line of argument is exactly the same as used in the proof of \Cref{prop:monotone} (\cpageref{prop:monotone}). One just has to take proper care of putting the universal quantification (for all $1 \leq i \leq n$) in the right place. Also for the following results we will just provide a reference to the corresponding endofunctor result and omit the simple (but admittedly tedious) calculations.

The really useful feature of considering multifunctor liftings is based on the fact that we have a slightly different domain of definition for evaluation functions which will also help us to solve the problems we described initially.

\begin{definition}[Multifunctor Evaluation Function and Evaluation Multifunctor]
	\label{def:evfct-multi} 
	\index{evaluation function (for multifunctors)}
	
	Let $F\colon \Set^n \to \Set$ be a multifunctor. We call any function $\ev_F\colon F(\reals,\dots,\reals) \to \reals$ an \emph{evaluation function} for $F$.
	Given such an evaluation function, the \emph{evaluation multifunctor} is the functor 
$\EvaluationFunctor{F} \colon(\Set/\reals)^n \to \Set/\reals$ where $\EvaluationFunctor{F}(g_1,\dots,g_n) := \ev_F\circ F(g_1,\dots,g_n)$ for all $g_i \in \Set/\reals$ and on arrows $\EvaluationFunctor{F}$ coincides with $F$. 
\end{definition}

Using this function, we immediately get the Kantorovich pseudometric and the corresponding lifting.

\begin{definition}[Kantorovich Distance for Multifunctors]
	\label{def:kantorovich-multi}
	\index{Kantorovich distance (for multifunctors)}
	Let $F\colon \Set^n\to \Set$ be a functor with evaluation function $\ev_F\colon F([0,\top],\dots,[0,\top]) \to [0,\top]$ and let $(X_1,d_1)$, \dots, $(X_n, d_n)$ be arbitrary pseudometric spaces. The \emph{Kantorovich distance} is the function $\Kantorovich{F}{(d_1,\dots,d_n)}\colon \big(F(X_1,\dots,X_n)\big)^2\to \reals$, where

	\begin{align*}
		\Kantorovich{F}{(d_1,\dots,d_n)}(t_1,t_2) := \sup_{\new{\substack{(f_1,\dots,f_n)\\f_i\colon (X_i,d_i) \nonexpansiveTo (\reals,d_e)}}} \!\!\!\!d_e\Big(\EvaluationFunctor{F}(f_1,\dots,f_n)(t_1),\EvaluationFunctor{F}(f_1,\dots,f_n)(t_2)\Big)
	\end{align*}
	for all $t_1, t_2 \in F(X_1,\dots,X_n)$.
\end{definition}

By adapting the proof of \Cref{prop:kantorovich-is-pseudometric} (\cpageref{prop:kantorovich-is-pseudometric}) we can prove that this function is indeed a pseudometric on $F(X_1,\dots, X_n)$ and can thus define the Kantorovich lifting of $F$ as $\LiftedFunctor{F}\colon \PMet^n\to\PMet$, $\LiftedFunctor{F}((X_1,d_1),\dots, (X_n,d_n)) = (F(X_1,\dots,X_n),\Kantorovich{F}{(d_1,\dots,d_n)})$ and $\LiftedFunctor{F}f = Ff$. The soundness of this definition can be shown using the same line of argument as in \Cref{prop:kantorovich-functorial} (\cpageref{prop:kantorovich-functorial}).

Not only the Kantorovich but also the Wasserstein lifting can be transferred to the multifunctor setting. For this we first need to define couplings. 

\begin{definition}[Coupling]
	Let $F \colon \Set^n \to \Set$ be a functor and $m \in \N$. Given sets $X_1,\dots,X_n$ and elements $t_j \in F(X_1,\dots,X_n)$ for $1 \leq j \leq m$ we call an element $t \in F(X_1^m,\dots, X_n^m)$ such that $F(\pi_{1,j},\dots,\pi_{n,j})(t) = t_j$ a \emph{coupling} of the $t_j$ (with respect to $F$) where $\pi_{i,j}$ are the projections $\pi_{i,j} \colon X_i^m \to X_i$. We write $\Couplings{F}(t_1, t_2, \dots, t_m)$ for the set of all these couplings.
\end{definition}

Using these couplings we can then again define a Wasserstein distance.

\begin{definition}[Wasserstein Distance for Multifunctors]
	\label{def:wasserstein-multi}
	Let $F\colon \Set^n\to \Set$ be a functor with evaluation function $\ev_F\colon F([0,\top],\dots,[0,\top]) \to [0,\top]$ and let $(X_1,d_1)$, \dots, $(X_n, d_n)$ be arbitrary pseudometric spaces.
	The \emph{Wasserstein distance} is the function $\Wasserstein{F}{(d_1,\dots,d_n)}\colon \big(F(X_1,\dots,X_n)\big)^2\to \reals$, where 
	\begin{align*}
	\Wasserstein{F}{(d_1,\dots,d_n)}(t_1, t_2) := \inf_{t \in \Couplings{F}(t_1,t_2)} \EvaluationFunctor{F}(d_1,\dots,d_n)(t).\,
	\end{align*}
	for all $t_1,t_2 \in F(X_1,\dots,X_n)$.
\end{definition}

As before, we will use well-behaved evaluation functions along with pullback preserving functors to obtain a Wasserstein pseudometric.

\begin{definition}[Well-Behaved Multifunctor Evaluation Function]
	\label{def:well-behaved-multi}
We call a multifunctor evaluation function $\ev_F \colon F(\reals,\dots,\reals) \to \reals$ \emph{well-behaved} if it satisfies the following three properties.
\begin{wbconditions}
	\item $\EvaluationFunctor{F}$ is monotone, i.e., given $f_i,g_i\colon X_i\to\reals$ with $f_i\le g_i$ for all $1\leq i \leq n$, we also have $\EvaluationFunctor{F}(f_1,\dots,f_n)\le\EvaluationFunctor{F}(g_1,\dots,g_n)$.\label{W1multi}
	\item \label{cond:evfct-ge-multi} Let $\pi_i \colon [0,\top]^2 \to [0,\top]$ be the projections. For all couplings $t\in F(\reals^2,\dots,\reals^2)$ we require $d_e\big(\EvaluationFunctor{F}(\pi_1,\dots,\pi_1)(t),\EvaluationFunctor{F}(\pi_2,\dots,\pi_2)(t)\big) \le \EvaluationFunctor{F}(d_e,\dots,d_e)(t)$. \label{W2multi}
	\item \label{cond:evfct-zero-multi} $\ev_F^{-1}[\set{0}] = F(i,\dots,i)[F(\{0\},\dots,\{0\})]$ where $i \colon \set{0} \hookrightarrow\reals$ is the inclusion map. \label{W3multi}
\end{wbconditions}
\end{definition}

\noindent By using a generalization \cite[Lem.~5.4.47]{Ker16} of the \nameref{lem:coupling} (\Cref{lem:coupling}) and well-behavedness we can prove sufficient conditions for the Wasserstein distance to be a pseudometric just as we did in \Cref{prop:wasserstein-is-pseudometric}.

\begin{theorem}
	Let $F\colon \Set^n \to \Set$ be a functor with evaluation function $\ev_F$. If $F$ preserves weak pullbacks and $\ev_F$ is well-behaved then the Wasserstein distance is a pseudometric.
\end{theorem}

This result gives rise to the Wasserstein lifting: For a weak pullback preserving functor $F\colon \Set^n\to \Set$ with well-behaved evaluation function $\ev_F\colon F([0,\top],\dots,[0,\top]) \to [0,\top]$ the \emph{Wasserstein lifting} is the functor $\LiftedFunctor{F}\colon \PMet^n\to\PMet$, where $\LiftedFunctor{F}\big((X_1,d_1),\dots, (X_n,d_n)\big) = \big(F(X_1,\dots,X_n),d\big)$ with $d = \Wasserstein{F}{(d_1,\dots,d_n)}$ and $\LiftedFunctor{F}f = Ff$. This definition is justified by adapting the proof of \Cref{prop:wasserstein-lifting-is-functorial} to obtain functoriality.

With these results at hand we quickly summarize a few of the properties for the multifunctor liftings which arise as natural generalizations of \Cref{prop:wasserstein-upper-bound} and \Cref{prop:wasserstein-vs-kantorovich,prop:kantorovich-lifting-preserves-isometries,prop:wasserstein-preserves-isometries,prop:preservation}.

\begin{theorem}
	\label{prop:multifunctor-lifting-properties}
	Let $F\colon \Set^n \to \Set$ be a functor with evaluation function $\ev_F$.
	\begin{enumerate}
		\item Both liftings preserve isometries.
		\item If $\ev_F$ satisfies \Cref{W1multi,W2multi} then $\Kantorovich{F}{(d_1,\dots,d_n)} \leq \Wasserstein{F}{(d_1,\dots,d_n)}$ holds for all pseudometric spaces $(X_i,d_i)$.
		\item If $F$ preserves weak pullbacks, $\ev_F$ is well-behaved and for all $t_1,t_2 \in F(X_1, \dots, X_n)$ with $\Wasserstein{F}{(d_1,\dots,d_n)}(t_1,t_2) = 0$ there is a coupling $\gamma(t_1,t_2) \in \Couplings{F}(t_1,t_2)$ such that $0=\Wasserstein{F}{(d_1,\dots,d_n)}(t_1,t_2) = \EvaluationFunctor{F}(d_1,\dots,d_n)\big(\gamma(t_1,t_2)\big)$ then $\Wasserstein{F}{(d_1,\dots,d_n)}$ is a metric for all metric spaces $(X_i, d_i)$.
	\end{enumerate}
\end{theorem}

\noindent Of course, whenever the two pseudometrics coincide for a functor and an evaluation function, we say that the \emph{Kantorovich-Rubinstein duality} or short \emph{duality} holds. 

\subsection{The Product and Coproduct Bifunctors}
\label{subsec:product-coproduct-bifunctors}
We conclude our section on multifunctors by considering two important examples at length, the product and the coproduct bifunctor. 

\begin{definition}[Product Bifunctor]
	\label{def:product-bifunctor}
	\index{product bifunctor}
	The \emph{product bifunctor} is the bifunctor $F\colon \Set^2 \to \Set$ where $F(X_1,X_2) = X_1\times X_2$ for all sets $X_1, X_2$ and $F(f_1,f_2) = f_1 \times f_2$ for all functions $f_i \colon X_i \to Y_i$.
\end{definition}

This functor fits nicely into our theory since it preserves pullbacks
\cite[Lem.~5.4.53]{Ker16}. The proof is simple and hence omitted. Let
us now discuss possible evaluation functions for this functor. They
are similar to those for the input functor in
\Cref{exa:wasserstein-input-functor} but we add some additional
parameters \new{as weighting factors to have additional flexibility
  and to demonstrate that one could choose different evaluation
  functions} (this could be done analogously for the input functor).
	
\begin{lemma}[Evaluation Functions for the Product Bifunctor]
	\label{lem:evfct-product-bifunctor}
	Let $F$ be the product bifunctor. The evaluation functions $\ev_F\colon [0,\top]^2 \to [0,\top]$ presented below are well-behaved.
	\begin{center}\begin{tabular}{c|c|c}
			Maximal Distance $\top$ & Other Parameters & $\ev_F(r_1,r_2)$ \\
			\hline
			$\top \in\ ]0, \infty]$ & $c_1,c_2 \in \,]0,1]$ & $\max \set{c_1r_1,c_2r_2}$ \\
			$\top = \infty $ & $c_1,c_2 \in \,]0,\infty[$\,, $p \in \N$ & $(c_1 x_1^p + c_2 x_2^p)^{1/p}$ \\
			$\top \in\,]0, \infty[$ & $c_1,c_2 \in \,]0,1], c_1+c_2 \leq 1$, $p \in \N$ & $(c_1 x_1^p + c_2 x_2^p)^{1/p}$ 
		\end{tabular}\end{center}
\end{lemma}
\begin{proof}
	Apparently the only difference between the second and the third row is the range of the parameters. It ensures that $\ev_F(r_1,r_2) \in [0,\top]$. We proceed by checking all three conditions for well-behavedness:
	\begin{wbconditions}
		\item Let $f_i, g_i\colon X_i\to\reals$ with $f_i\le g_i$ be given. For the maximum we have $\EvaluationFunctor{F}(f_1,f_2) = \max\set{c_1f_1,c_2f_2} \leq \max\set{c_1g_1,c_2g_2}= \EvaluationFunctor{F}(g_1,g_2)$ 
		and for the second evaluation function, we also obtain $\EvaluationFunctor{F}(f_1, f_2) = \big(c_1\cdot f_1^p + c_2\cdot f_2^p\big)^{1/p} \le \big(c_1\cdot g_1^p + c_2\cdot g_2^p\big)^{1/p} = \EvaluationFunctor{F}(g_1, g_2)$
		due to monotonicity of all involved functions since $c_1,c_2>0$.
		
		\item Let $\pi_i \colon [0,\top]^2 \to [0,\top]$ be the product projections and define $t:=(x_{11},x_{21},x_{12},x_{22})\in F([0,\top]^2,[0,\top]^2) = [0,\top]^2\times [0,\top]^2$. We have to show the inequality 
		\begin{align}
			d_e\Big(\EvaluationFunctor{F}(\pi_1,\pi_1)(t), \EvaluationFunctor{F}(\pi_2, \pi_2)(t)\Big) \leq \EvaluationFunctor{F}(d_e,d_e)(t)\,.\label{eq:evfct-product1}
		\end{align}
		 To do this, we first observe that the right hand side of this inequality evaluates to $\EvaluationFunctor{F}(d_e,d_e)(t) = \ev_F\big(d_e(x_{11},x_{21}), d_e(x_{12},x_{22})\big)$. Moreover, we have $\EvaluationFunctor{F}(\pi_i,\pi_i)(t) = \ev_F(x_{i1},x_{i2})$ so if we define $z_i = \ev_F(x_{i1},x_{i2})$ the left hand side of \eqref{eq:evfct-product1} can be rewritten as $d_e(z_1,z_2)$. Thus \eqref{eq:evfct-product1} is equivalent to
		\begin{equation}
			d_e\left(z_1,z_2\right) \leq \ev_F\big(d_e(x_{11},x_{21}), d_e(x_{12},x_{22})\big)\,.\label{eq:evfct-product2}
		\end{equation}
		If $z_1=z_2$ this is obviously true because $d_e(z_1, z_2) = 0$ and the right hand side of \eqref{eq:evfct-product2} is non-negative. We now assume $z_1 > z_2$ (the other case is symmetrical). For $\infty = z_1 > z_2$ inequality \eqref{eq:evfct-product2} holds because then $x_{11} = \infty$ or $x_{12} = \infty$ and $x_{21},x_{22} < \infty$ (otherwise we would have $z_2 = \infty$) so both the left hand side and the right hand side are $\infty$. Thus we can now restrict our attention to $\infty > z_1 > z_2$ where necessarily also $x_{11}, x_{12}, x_{21}, x_{22} < \infty$ (otherwise we would have $z_1 = \infty$ or $z_2=\infty$). According to \Cref{lem:sum-vs-dist} (\cpageref{lem:sum-vs-dist}), the inequality \eqref{eq:evfct-product2} is equivalent to showing the two inequalities $z_1 \leq z_2 + \ev_F\big(d_e(x_{11},x_{21}), d_e(x_{12},x_{22})\big)$ and $z_2 \leq z_1 + \ev_F\big(d_e(x_{11},x_{21}), d_e(x_{12},x_{22})\big)$. By our assumption ($\infty > z_1>z_2$) the second of these inequalities is satisfied, so we just have to show the first. 
			\begin{enumerate}
				\item For the discounted maximum as evaluation function we have $z_i = \max\set{c_1x_{i1}, c_2x_{i2}}$. If for $z_1$ the maximum is attained for the first element, i.e., if $z_1 = c_1x_{11}$, we can conclude that 
$z_2 + \max\set{c_1d_e(x_{11}, x_{21}), c_2 d_e(x_{12},x_{22})} \geq z_2 + c_1d_e(x_{11},x_{21}) = z_2 + c_1|x_{11}-x_{21}| \geq z_2 + c_1(x_{11}-x_{21})= z_2 + c_1x_{11} - c_1x_{21} = z_2 + z_1 - c_1x_{21} = z_1 + (z_2-c_1x_{21}) \geq z_1$ because $z_2 = \max\set{c_1x_{21},c_2x_{22}} > c_1x_{21}$ and therefore $(z_2-c_1x_{21}) \geq 0$. The same line of argument can be applied if $z_1 = c_2x_{12}$.
				
				\item For the second evaluation function we can simply use the Minkowski inequality to obtain the result \cite[Lem.~5.4.54]{Ker16}.
			\end{enumerate}
			
		\item Both functions satisfy \Cref{cond:evfct-zero-multi} of \Cref{def:evfct-multi}, because $F(i,i)[F(\set{0},\set{0})] = (i\times i)[\set{0} \times \set{0}] = \set{(0,0)}$ and for both evaluation functions apparently $\ev_F^{-1}[\set{0}] = \set{(0,0)}$.\qedhere
	\end{wbconditions}
\end{proof}

\noindent Using these well-behaved evaluation functions we can now lift the product bifunctor using our multifunctor lifting framework.

\begin{lemma}[Product Pseudometrics] 
	\label{lem:product-pseudometrics}
	\index{product pseudometrics}
	Let $F$ be the product bifunctor of \Cref{def:product-bifunctor}. For the evaluation functions presented in \Cref{lem:evfct-product-bifunctor} the Kantorovich-Rubinstein duality holds and the supremum [infimum] of the Kantorovich [Wasserstein] pseudometric is always a maximum [minimum]. Moreover, for all pseudometric spaces $(X_1,d_1)$, $(X_2,d_2)$ we obtain the lifted pseudometrics $\LiftedMetric{F}{(d_1,d_2)}\colon (X_1 \times X_2)^2 \to [0,\top]$ as given in the table below.
	\begin{center}\begin{tabular}{c|c}
			$\ev_F(r_1,r_2)$ & $\LiftedMetric{F}{(d_1,d_2)}\big((x_1,x_2),(y_1,y_2)\big)$ \\
			\hline
			 $\max \set{c_1r_1,c_2r_2}$ & $\max\set{c_1d_1(x_1,y_1), c_2d_2(x_2,y_2)}$\\
			 $\Big(c_1 x_1^p + c_2 x_2^p\Big)^{1/p}$ &$\Big(c_1d_1(x_1,y_1)^p+c_2d_2(x_2,y_2)^p\Big)^{1/p}$
		\end{tabular}\end{center}
	\end{lemma} 
\begin{proof}
Let $(X_1,d_1)$, $(X_1,d_2)$ be pseudometric spaces, $\pi_i \colon X_1^2 \to X_1$ and $\tau_i\colon X_2^2 \to X_2$ be the projections and let $t_1 = (x_1,x_2), t_2 = (y_1,y_2) \in F(X_1,X_2) = X_1 \times X_2$ be given. We define $t := (x_1,y_1,x_2,y_2) \in F(X_1^2,X_2^2)$ and observe that $F(\pi_1,\tau_1)(t) = t_1$, $F(\pi_2,\tau_2)(t) = t_2$ and thus $t \in \Gamma_F(t_1,t_2)$ is a coupling of $t_1$ and $t_2$. In the following we will construct nonexpansive functions $f_i\colon(X_i,d_i) \nonexpansiveTo (\reals,d_e)$ such that $d_e\big(\EvaluationFunctor{F}(f_1,f_2)(t_1), \EvaluationFunctor{F}(f_1,f_2)(t_2)\big) = \EvaluationFunctor{F}(d_1,d_2)(t)$ holds. Due to \Cref{prop:multifunctor-lifting-properties} (\cpageref{prop:multifunctor-lifting-properties}) we can then conclude that duality holds and both supremum and infimum are attained.
	\begin{enumerate}
	\item For the first evaluation function we have $\EvaluationFunctor{F}(d_1,d_2)(t) = \max\set{c_1d_1(x_1,y_1),c_2d_2(x_2,y_2)}$ and assume without loss of generality that $c_1d_1(x_1,y_1)$ is the maximal element. We define $f_1:=d_1(x_1,\_)$, which is nonexpansive due to \Cref{lem:alt-char-triangle} (\cpageref{lem:alt-char-triangle}), and $f_2$ to be the constant zero-function which is obviously nonexpansive as a constant function. Then we have:
	\begin{align*}
		& d_e\Big(\EvaluationFunctor{F}(f_1,f_2)(t_1),\EvaluationFunctor{F}(f_1,f_2)(t_2)\Big) \\
		& \quad = d_e\big(\max\set{c_1f_1(x_1),c_2f_2(x_2)},\max\set{c_1f_1(y_1),c_2f_2(y_2)}\big)\\
		& \quad = d_e\big(\max\set{c_1f_1(x_1),0},\max\set{c_1f_1(y_1),0}\big) = d_e\big(c_1f_1(x_1),c_1f_1(y_1)\big) \\
		& \quad = c_1d_1(x_1,y_1) = \max\set{c_1d_1(x_1,y_1),c_2d_2(x_2,y_2)} = \EvaluationFunctor{F}(d_1,d_2)(t)
	\end{align*}
	The case where $c_2d_2(x_2,y_2)$ is the maximal element is treated analogously.
	\item For the second evaluation function we define $f_1 := d_1(x_1,\_)$ and $f_2:=d_2(x_2,\_)$ which are nonexpansive by \Cref{lem:alt-char-triangle} and obtain
	\begin{align*}
		& d_e\Big(\EvaluationFunctor{F}(f_1,f_2)(t_1),\EvaluationFunctor{F}(f_1,f_2)(t_2)\Big) \\
		&=d_e\left(\left(c_1f_1^p(x_1)+c_2f_2^p(x_2)\right)^{1/p},\left(c_1f_1^p(y_1)+c_2f_2^p(y_2)\right)^{1/p}\right)\\
		& = d_e\left(0, \left(c_1d_1^p(x_1, y_1)+c_2d_2^p(x_2,y_2)\right)^{1/p}\right) \\
		&= \left(c_1d_1^p(x_1, y_1)+c_2d_2^p(x_2,y_2)\right)^{1/p} = \EvaluationFunctor{F}(d_1,d_2)(t)
	\end{align*}
	which completes the proof.\qedhere
	\end{enumerate}
\end{proof}

\noindent While all the product pseudometrics of \Cref{lem:product-pseudometrics} are well-known, we point out a specifically interesting one, the undiscounted maximum pseudometric.

\begin{lemma}[Binary Products in $\PMet$]
	\label{lem:binary-prod-PMet}
	If $c_1=c_2=1$ for the first evaluation function in \Cref{lem:product-pseudometrics} we obtain for two given pseudometric spaces $(X_1, d_1)$ and $(X_2,d_2)$ as lifted pseudometric the function $d_\infty\colon (X_1 \times X_2)^2 \to [0,\top]$, with $d\big((x_1,x_2),(y_1,y_2)\big) = \max\set{d_1(x_1,y_1), d_2(x_2,y_2)}$. The resulting pseudometric space $(X_1 \times X_2, d_\infty)$ is exactly the categorical product of $(X_1,d_1)$ and $(X_2,d_2)$ in $\PMet$.
\end{lemma}
\begin{proof}
	This follows from \Cref{prop:PMetcomplete} (\cpageref{prop:PMetcomplete}) by taking $I=\two$ and $f_i = \pi_i \colon X_1 \times X_2 \to X$.
\end{proof}

In a completely analogous way as for the product bifunctor, we will now introduce and study the coproduct bifunctor. 

\begin{definition}[Coproduct Bifunctor]
	\label{def:coproduct-bifunctor}
	The \emph{coproduct bifunctor} is the functor $F\colon \Set^2 \to \Set$, where $F(X_1,X_2) = X_1 + X_2$ for all sets $X_1$, $X_2$ and $F(f_1,f_2) = f_1+f_2$ for all functions $f_1\colon X_1 \to Y_1$, $f_2\colon X_2 \to Y_2$. Explicitly\footnote{We use the representation $X_1 + X_2  \cong X_1\times\{1\} \cup X_2\times \{2\}$.}, the function $f_1+f_2\colon X_1+X_2 \to Y_1 + Y_2$ is given via the assignment $f_1+f_2(x,i) = \big(f_i(x),i\big)$.
\end{definition}

As for the product bifunctor, one can easily show that this bifunctor preserves pullbacks \cite[5.4.58]{Ker16} and we omit the simple proof. For this functor we will just consider one type of evaluation function, parametrized by the maximal element $\top$ of the pseudometrics.

\begin{lemma}[Evaluation Function for the Coproduct Bifunctor]
	\label{lem:coproduct-bifunctor-evfct}
	Let $F$ be the coproduct bifunctor. The function $\ev_F\colon\reals+\reals\to\reals$, where $\ev_F(x,i) = x$, is well-behaved. 
\end{lemma}
\begin{proof}
	We show the three properties of a well-behaved evaluation function. 
	\begin{wbconditions}
		\item Let $f_1,f_2,g_1,g_2\colon X\to\reals$ with $f_1\le g_1$, $f_2\le g_2$ and $(z,i)\in F(X_1,X_2) = X_1+ X_2$. We have $\EvaluationFunctor{F}(f_1,f_2)(z,i) = \ev_F\big(F(f_1,f_2)(z,i)\big) = f_i(z) \le g_i(z) = \ev_F\big(F(g_1,g_2)(z,i)\big) = \EvaluationFunctor{F}(g_1,g_2)(z,i)$.

		\item Let $t = \big((x,y),i\big)\in F(\reals^2,\reals^2) = \reals^2\times\{1,2\}$. We obtain $\EvaluationFunctor{F}(d_e,d_e)(t) = \ev_F\big(d_e(x,y),i\big) = d_e(x,y)= d_e\big(\ev_F(x,i),\ev_F(y,i)\big) = d_e\Big(\EvaluationFunctor{F}(\pi_1,\pi_1)(t),\EvaluationFunctor{F}(\pi_2,\pi_2)(t)\Big)$.
		\item Let $i\colon {0} \hookrightarrow \reals$ be the inclusion function. We have $Fi[F(\set{0},\set{0})] = (i + i) [\set{0} + \set{0}] =\set{0} \times \set{1,2} = \ev_F^{-1}[\set{0}]$.\qedhere
	\end{wbconditions}
\end{proof}

\noindent With this evaluation function we can now employ our multifunctor lifting framework to obtain the following coproduct pseudometric.

\begin{lemma}[Coproduct Pseudometric]
	\label{lem:coproduct-pseudometric}
	\index{coproduct pseudometric}
	For the coproduct bifunctor of  \Cref{def:coproduct-bifunctor} and the evaluation function of \Cref{lem:coproduct-bifunctor-evfct} the Kantorovich-Rubinstein duality holds, the supremum of the Kantorovich pseudometric is always a maximum, the infimum of the Wasserstein pseudometric is a minimum whenever a coupling exists and we obtain the coproduct pseudometric
	\begin{align*}
		d_+ \colon (X_1 + X_2)^2 \to [0,\top], \quad d_+\big((x_1,i_1), (x_2, i_2)\big) =%
		\begin{cases}
			d_{i}(x_1,x_2), & \text{if } i_1=i_2=i\\
			\top, & \text{else}
		\end{cases}\,.
	\end{align*}
\end{lemma}

\begin{proof}
Let $(X_1,d_1)$, $(X_1,d_2)$ be pseudometric spaces, $\pi_i \colon X_1^2 \to X_1$ and $\tau_i\colon X_2^2 \to X_2$ be the projections and and $t_1,t_2 \in F(X_1,X_2) = X_1 + X_2$, say $t_1 = (z,i)$, $t_2 = (z',i')$. We distinguish two cases.
	\begin{enumerate}
		\item For $i=i'$ we define $t = \big((z,z'),i\big)$ and observe that $F(\pi_1,\tau_1)\big((z,z'),i\big) = t_1$, and $F(\pi_2,\tau_2)\big((z,z'),i\big) = t_2$, thus $t \in \Gamma_F(t_1,t_2)$. Furthermore $\EvaluationFunctor{F}(d_1,d_2)(t) = d_i(z,z')$. If  $i=i'=1$ we define $f_1:=d_1(z,\_)\colon (X_1,d_1) \nonexpansiveTo (\reals,d_e)$ which is nonexpansive according to \Cref{lem:alt-char-triangle} and consider an arbitrary nonexpansive function $f_2\colon (X_2,d_2) \nonexpansiveTo (\reals,d_e)$ (e.g. the constant zero-function). Then we have $d_e\big(\EvaluationFunctor{F}(f_1,f_2)(t_1),\EvaluationFunctor{F}(f_1,f_2)(t_2)\big) = d_e\big(\EvaluationFunctor{F}(f_1,f_2)(z,1),\EvaluationFunctor{F}(f_1,f_2)(z',1)\big) =  d_e\big(f_1(z),f_1(z')\big)= d_e\big(0, d_1(z,z')\big) = d_1(z,z') = d_i(z,z')$. The case $i=i'=2$ can be treated analogously.
	\item  In the case where $i\neq i'$, there is no coupling that projects to $(z,i)$ and $(z',i')$, thus $\Wasserstein{F}{(d_1,d_2)}(t_1,t_2) = \top$. We show that also $\Kantorovich{F}{(d_1,d_2)}(t_1,t_2) = \top$. We define $f_1$ to be the constant $0$-function and $f_2$ the constant $\top$-function. We have $d_e\big(\EvaluationFunctor{F}(f_1,f_2)(t_1),\EvaluationFunctor{F}(f_1,f_2)(t_2)\big) = d_e\big(\EvaluationFunctor{F}(f_1,f_2)(z,i),\EvaluationFunctor{F}(f_1,f_2)(z',j)\big) =  d_e\big(f_i(z),f_j(z')\big) = d_e(0,\top) = \top$.\qedhere
	\end{enumerate}
\end{proof}

\noindent We conclude this section by the observation that this yields the categorical coproduct.

\begin{lemma}[Binary Coproducts in $\PMet$]
	The pseudometric space $(X_1+X_2,d_+)$ where $d_+$ is the pseudometric given in \Cref{lem:coproduct-pseudometric} is exactly the categorical coproduct of $(X_1,d_1)$ and $(X_2,d_2)$ in $\PMet$.
\end{lemma}
\begin{proof}
	This follows from \Cref{prop:PMetcomplete} by taking $I=\two$ and $f_i = \iota_i \colon X_i \to X_1 + X_2$.
\end{proof}

\section{Bisimilarity Pseudometrics}
\label{sec:final-coalgebra}
We now want to use our lifting framework to derive bisimilarity pseudometrics. We assume an arbitrary lifting $\LiftedFunctor{F}\colon \PMet \to \PMet$ of an endofunctor $F$ on $\Set$ and, for any pseudometric space $(X,d)$, we
write $\LiftedMetric{F}{d}$ for the pseudometric obtained by applying
$\LiftedFunctor{F}$ to $(X,d)$. Such a lifting can be obtained as an
endofunctor lifting, by taking a lifted multifunctor and fixing all parameters apart from one or by the composition of such functors.

We first observe that any coalgebra \new{$c\colon X\to FX$} for a set
functor $F\colon \Set \to \Set$ can be ``lifted'' to a coalgebra for
the lifting $\LiftedFunctor{F}\colon\PMet\to\PMet$ by endowing $X$
with a canonical pseudometric, i.e., the least pseudometric making $c$
a nonexpansive function (that actually turns $c$ into an isometry).

\begin{lemma}[Lifting Coalgebras to $\PMet$]
  \label{lem:lifting-coalg}
  Let $\LiftedFunctor{F}\colon\PMet\to\PMet$ be a lifting of a functor
  $F\colon \Set\to \Set$ and let
  $c\colon X\to FX$ be a coalgebra. The mapping associating each pseudometric
  $d \colon X\times X\to\reals$ with
  $\LiftedMetric{F}{d}\circ(c\times c)$ is monotonic over the lattice
  of pseudometrics on $X$, hence it has a least fixed point
  \begin{center}
    $d_c = \inf \set{ d \mid d \colon X\times X\to\reals\ \text{pseudometric} \land\
    \LiftedMetric{F}{d}\circ(c\times c) \leq d }$
  \end{center}
  and
  $c \colon (X,d_c) \nonexpansiveTo
  (FX,\LiftedMetric{F}{d_c})$ is an isometry.
  
  Moreover, if $c' \colon Y \to FY$ is another coalgebra
  and $f \colon X \to Y$ is a coalgebra homomorphism then
  $f \colon (X, d_c) \nonexpansiveTo (Y, d_{c'})$ is nonexpansive. It is an
  isometry if additionally $\LiftedFunctor{F}$ preserves isometries.
\end{lemma}

\begin{proof}
  We first observe that the mapping associating each pseudometric
  $d \colon X\times X\to\reals$ with
  $\LiftedMetric{F}{d}\circ(c\times c)$ is monotonic over the lattice
  of pseudometrics on $X$. Let $d, d' \colon X\times X\to\reals$ be
  pseudometrics such that $d \leq d'$. Then by
  \Cref{prop:monotone},
  $\LiftedMetric{F}{d} \leq \LiftedMetric{F}{d'}$ and thus
  $\LiftedMetric{F}{d}\circ(c\times c) \leq
  \LiftedMetric{F}{d'}\circ(c\times c)$, as desired.
  The fact that $d_c$ is a fixed point exactly says that
  $c \colon (X,d_c) \nonexpansiveTo (FX,\LiftedMetric{F}{d_c})$ is an
  isometry.
  
  Now, let $c' \colon Y \to FY$ be another coalgebra and
  $f\colon X \to Y$ be a coalgebra homomorphism.
We first show that
  $f\colon (X, d_c) \nonexpansiveTo (Y, d_{c'})$ is nonexpansive.
  Consider the pseudometric $d_Y \colon Y\times Y\to\reals$ defined as
  \[
    d_Y = \sup \set{ d \mid d \colon Y\times Y\to\reals\ \text{pseudometric}\ \land\ f \colon (X, d_c) \nonexpansiveTo (Y, d)\ \text{nonexpansive} }\,.
  \]
  Since the supremum is taken pointwise,
  $f \colon (X,d_c) \nonexpansiveTo (Y, d_Y)$ is nonexpansive and thus
  $Ff \colon (FX,\LiftedMetric{F}{d_c}) \nonexpansiveTo (FY,
  \LiftedMetric{F}{d_Y})$ is also nonexpansive. This allows us to deduce that
  \begin{equation}
    \label{eq:inf}
    \LiftedMetric{F}{d_Y} \circ (c'\times c') \leq d_Y\,.
  \end{equation} 
  In fact we have
  \new{
  \begin{align*}
    & \LiftedMetric{F}{d_Y} \circ (c'\circ f \times c'\circ f)
    & \\
    & \quad = \LiftedMetric{F}{d_Y} \circ (Ff\circ c \times Ff\circ c) 
    & \mbox{[since $f$ is coalgebra hom.: $Ff \circ c = c' \circ f$]}\\
    & \quad \leq \LiftedMetric{F}{d_c} \circ (c \times c)
    & \mbox{[since  $Ff$ is nonexpansive]}\\
    & \quad \leq d_c \,.
    & \mbox{[since $c$ is nonexpansive]}
  \end{align*}
  }
  The above means that $f\colon (X, d_c) \nonexpansiveTo (Y, \LiftedMetric{F}{d_Y} \circ (c' \times c'))$ is nonexpansive. Since $d_Y$ is defined as the supremum of the pseudometrics making $f$ nonexpansive, we conclude \Cref{eq:inf}.
  Given the characterization of $d_{c'}$ as least of pre-fixed points, it follows that $d_{c'} \leq d_Y$. Recalling that $f\colon (X,d_c) \to (Y, d_Y)$ was nonexpansive, we conclude that $d_{c'} \circ (f\times f) \leq d_Y \circ (f \times f) \leq d_c$ thus $f\colon (X,d_c) \to (Y, d_{c'})$ is nonexpansive, as desired.

  We finally prove that, whenever $\LiftedFunctor{F}$ preserves
  isometries, $f\colon (X, d_c) \to (Y, d_{c'})$ is an isometry.
  Define a pseudometric $d_X \colon X\times X\to\reals$ as
  $d_X = d_{c'} \circ (f \times f)$.
  Observe that $f \colon (X, d_X) \to (Y, d_{c'})$ is isometric by
  definition. Hence
  $Ff\colon (FX,\LiftedMetric{F}{d_X}) \to (FY, \LiftedMetric{F}{d_{c'}})$ is
  an isometry.
  Moreover we have
  \new{
  \begin{align*}
    d_X(x,x')
	& \quad = d_{c'} \circ (f\times f)
	& \mbox{{[by definition of $d_X$]}}\\
	& \quad = \LiftedMetric{F}{d_{c'}} \circ (c'\circ f \times
        c'\circ f)
	& \mbox{{[since $c' \colon (Y, d_{c'}) \nonexpansiveTo (FY, \LiftedMetric{F}{d_{c'}})$ is an isometry]}}\\
    & \quad = \LiftedMetric{F}{d_{c'}} \circ (Ff\circ c \times
    Ff\circ c) & \mbox{{[since $f$ is coalgebra hom.: $Ff \circ c = c' \circ f$]}}\\
	& \quad = \LiftedMetric{F}{d_X} \circ (c\times c)\,.
	& \mbox{{[since $Ff\colon (FX,\LiftedMetric{F}{d_X}) \to (FY, \LiftedMetric{F}{d_{c'}})$ is an isometry]}}
      \end{align*}
      }
  The above shows that $\LiftedMetric{F}{d_X} \circ (c \times c) =
  d_X$. Therefore, by the characterization of $d_c$ as least of
  pre-fixed points, it follows that $d_c \leq d_X$. Hence
  \new{
  \[
    d_c \leq d_X = d_{c'}\circ (f\times f) \leq d_c
  \]
  }
  the last inequality being nonexpansiveness of
  $f\colon (X,d_c) \to (Y, d_{c'})$. Hence we conclude that $f$ is indeed an isometry.
\end{proof}

For a coalgebra $c\colon X\to FX$, the least fixpoint pseudometric $d_c\colon X\times X\to\reals$ above, given by the Knaster-Tarski theorem \cite{t:lattice-fixed-point}, can be characterized also via transfinite induction. For every ordinal $i$ we construct a pseudometric
$d_i\colon X\times X\to\reals$ as follows:
\begin{itemize}
\item $d_0 := 0$ is the zero pseudometric, 
\item $d_{i+1} := \LiftedMetric{F}{d_i}\circ(c\times c)$ for all ordinals $i$ and
\item $d_j := \sup_{i<j} d_i$ for all limit ordinals $j$.
\end{itemize}
This sequence converges to $d_\theta$ for some ordinal $\theta$.

Moreover, if $c'\colon X\times X\to\reals$ is another coalgebra and
$f \colon X \to Y$ is a coalgebra homomorphism and the fixpoint
pseudometric $d_{c'} \colon X \times X \to \reals$ is reached by the
construction above at ordinal $\zeta$, then $\theta \leq \zeta$.

We can now easily show that if $z \colon Z \to F Z$ is a final $F$-coalgebra,
then its lifting is a final $\LiftedFunctor{F}$-coalgebra.

\begin{theorem}[Final Coalgebra Construction for Liftings]
  \label{thm:final-coalgebra}
  Let $\LiftedFunctor{F}\colon\PMet\to\PMet$ be a lifting of a functor
  $F\colon \Set\to \Set$ such that $F$ has a final coalgebra
  $z\colon Z\to FZ$. Then 
  $z\colon (Z,d_z) \nonexpansiveTo
  (FZ,\LiftedMetric{F}{d_z})$, where $d_z$ is defined as in \Cref{lem:lifting-coalg} is the final
  $\LiftedFunctor{F}$-coalgebra.
\end{theorem}

\begin{proof}
  By \Cref{lem:lifting-coalg}, if we let
  $d_z = \inf \{ d \mid d \colon Z\times Z\to\reals\ \text{pseudometric}\land\
  \LiftedMetric{F}{d}\circ(z\times z) \leq d \}$ then
  $z \colon (X,d_z) \nonexpansiveTo (FZ,\LiftedMetric{F}{d_z})$ is an
  isometry.
  
  Let $c\colon (X,d_X)\nonexpansiveTo \LiftedFunctor{F}(X,d_X)$ be any
  $\LiftedFunctor{F}$-coalgebra. Again, \Cref{lem:lifting-coalg},
  provides a pseudometric
  $d_c = \inf \{ d \mid d \colon X\times X\to\reals\ \text{pseudometric}\land\
  \LiftedMetric{F}{d}\circ(c\times c) \leq d \}$ such that
  $c \colon (X,d_c) \nonexpansiveTo (FZ,\LiftedMetric{F}{d_c})$ is nonexpansive.
  Since $c$ is also nonexpansive with respect to $d_X$, i.e.,  $\LiftedMetric{F}{d_X} \circ (c \times c) \leq d_X$, we must have $d_c \leq d_X$.
  
  Now, consider the underlying $F$-coalgebra $c\colon X\to FX$ in
  $\Set$. Since $z$ is the final $F$-coalgebra, there is a unique
  function $f\colon X\to Z$ such that $z\circ f = Ff\circ c$. By
  \Cref{lem:lifting-coalg}, $f\colon (X,d_c)\nonexpansiveTo (Z,d_z)$ is
  nonexpansive, thus $d_z \circ (f\times f) \leq d_c$. Recalling that $d_c \leq d_X$ we conclude that also $f\colon (X,d_X)\nonexpansiveTo (Z,d_z)$ is nonexpansive, as desired.
\end{proof}

As a first simple example of this construction we consider the machine functor. We will look at more examples at the end of this section.

\begin{example}[Final Coalgebra for the Lifted Machine Functor]
	\label{exa:final-coalg-lifted-machine}
	We consider the machine endofunctor $M_\two = \two \times \_^A$, take $\top=1$ and as evaluation function we use $\ev_{M_\two} \colon [0,1] \times [0,1]^A$ with $\ev_{M_\two}(o,s) = c \cdot \max_{a \in A} s(a)$ for $0<c <1$ as in \Cref{exa:wasserstein-machine} (\cpageref{exa:wasserstein-machine}). 
	
	It is well-known that the carrier of the final $M_\two$-coalgebra is $\two^{A^*}$. Moreover, from \Cref{exa:wasserstein-machine} we know that for any pseudometric $d$ on $\two^{A^*}$ we obtain as Wasserstein pseudometric the function $\Wasserstein{F}{d}\colon \big(\two \times (\two^{A^*})^A\big)^2 \to [0,1]$ where, for all $(o_1,s_1), (o_2, s_2) \in \two \times (\two^{A^*})^A$, 
	\begin{align*}
	\Wasserstein{F}{d}\big((o_1,s_1), (o_2, s_2)\big) = \max\Big\{d_\two(o_1,o_2), c \cdot \max_{a \in A}d\big(s_1(a), s_2(a)\big)\Big\}
	\end{align*}
	where $d_\two$  is the discrete metric on $\two$. Thus the fixed-point equation induced by \Cref{lem:lifting-coalg} is given by, for $L_1,L_2 \in \two^{A^*}$,
	\begin{align*}
		d(L_1,L_2) = \max\Big\{d_\two\big(L_1(\epsilon), L_2(\epsilon)\big), c \cdot \max_{a \in A} d\big(\lambda w.L_1(aw), \lambda w.L_2(aw)\big)\Big\}\,.
	\end{align*}
	Now because $d_\two$ is the discrete metric with $d_\two(0,1) = 1$ we can easily see that $d_{2^{A^*}}$ as defined below is indeed the least fixed-point of this equation and thus $(\two^{A^*}, d_{\two^{A^*}})$ is the carrier of the final $\LiftedFunctor{M_\two}$-coalgebra.
	\begin{align*}
	d_{\two^{A^*}} \colon \two^{A^*} \times \two^{A^*} \to [0,1], \quad d_{\two^{A^*}}(L_1,L_2) = c^{\inf\set{n \in \N \mid \exists w \in A^n.L_1(w) \not = L_2(w)}}\,.
	\end{align*} 
	Thus the distance between two languages $L_1,L_2\colon A^*\to \two$ can be determined by looking for a word $w$ of minimal length which is contained in one and not in the other. Then, the distance is computed as $c^{|w|}$. %This is similar to a standard ultrametric between traces \cite{BN80}.
\end{example}

We already noted in the beginning of this paper that for any set $X$, the set of pseudometrics over $X$, with pointwise order, is a complete lattice. Moreover, by \Cref{prop:monotone} the lifting $\LiftedFunctor{F}$ induces a monotone function $\_^F$ which maps any pseudometric $d$ on $X$ to $d^F$ on $FX$. If, additionally, such a function is $\omega$-continuous\index{$\omega$-continuity}, i.e., if it preserves the supremum of $\omega$-chains, the least fixpoint metric
in \Cref{lem:lifting-coalg} can be obtained by a construction that converges at an ordinal less than or equal to $\omega$. It is easy to see that this is the case in \Cref{exa:final-coalg-lifted-machine}. With some more effort one can also show that also the liftings induced by the finite powerset functor and the probability distribution functor with finite support are $\omega$-continuous \cite[Prop. P.6.1]{BBKK14}. 

The lifting of the final $F$-coalgebra to a final $\LiftedFunctor{F}$-coalgebra, which is provided by \Cref{thm:final-coalgebra}, allows us to move from a qualitative to a quantitative behavior analysis: Instead of just considering equivalences, in $\PMet$ we can now measure the distance of behaviors using the final coalgebra. 

\begin{definition}[Bisimilarity Pseudometric]
  \label{def:bisim-pseudometric}
  Let $\LiftedFunctor{F}\colon \PMet \to \PMet$ be a lifting of a
  $\Set$-endofunctor $F$ for which a final coalgebra $z \colon Z \to FZ$
  exists and $d_z \colon X \times X \to \reals$ be the fixpoint pseudometric of \Cref{lem:lifting-coalg}. For any $F$-coalgebra $c\colon X \to FX$, the \emph{bisimilarity distance} on
  $X$ is the function $\bd_c \colon X \times X \to [0,\top]$ where
  \begin{align*}
    \bd_c(x,y) := d_z\big(\final{x}_c,\final{y}_c\big)
  \end{align*} 
  for all $x,y \in X$ with $\final{\cdot}_c \colon X \to Z$ being the unique map into the final coalgebra. Since $d_z$ is a pseudometric also $\bd_c$ is a pseudometric.
\end{definition}

Let us check how this definition applies to deterministic automata.

\begin{example}[Bisimilarity Pseudometric for Deterministic Automata]
	\label{exa:bisim-pseudo-det-aut}
	We instantiate the above definition for the machine functor $M_\two = \two \times \_^A$ with maximal distance $\top=1$ and evaluation function $\ev_{M_\two} \colon [0,1] \times [0,1]^A$ with $\ev_{M_\two}(o,s) = c \cdot \max_{a \in A} s(a)$ for $0<c <1$ as in \Cref{exa:final-coalg-lifted-machine}. We recall that for any coalgebra $\alpha \colon X \to \two \times X^A$ the unique map $\final{\cdot}_\alpha \colon X \to \two^{A^*}$ into the final coalgebra maps each state $x \in X$ to the language $\final{x}_\alpha \colon A^* \to \two$ it accepts. Using the final coalgebra pseudometric from \Cref{exa:final-coalg-lifted-machine} we have
	\begin{align*}
		\bd_\alpha \colon X \times X \to [0,1], \quad \bd_\alpha(x,y) = c^{\inf\set{n \in \N \mid \exists w \in A^n.\final{x}_\alpha(w) \not = \final{y}_\alpha(w)}}\,.
	\end{align*} 
	Thus the distance between two states $x,y \in X$ is determined by the shortest word $w$ which is contained in the language of one state and not in the language of the other. Then the distance is computed as $c^{|w|}$.
\end{example}

The name \emph{bisimilarity} pseudometric is motivated by the fact that states which are mapped to the same element in the final coalgebra are at distance $0$.

\begin{lemma}
	\label{thm:bisim-pseudometric}
	Let $\LiftedFunctor{F}$, $F$, $c\colon X \to FX$, $\final{\cdot}_c$ and $\bd_c$ be as in \Cref{def:bisim-pseudometric}. For all $x,x' \in X$,
	\begin{align*}
		\final{x}_c = \final{x'}_c \implies \bd_c(x,x') = 0\,.
	\end{align*}
\end{lemma}
\begin{proof}
  Since $d_z$ is a pseudometric, it is reflexive and thus $\bd_c(x,x') = d_z\big(\final{x}_c,\final{y}_c\big) = 0$.
\end{proof}

If $d_z$ is a metric, also the converse of the above implication holds. We will later (\Cref{thm:d-omega-is-metric}, \cpageref{thm:d-omega-is-metric}) provide sufficient conditions which guarantee that this is the case.

Before doing so, we show that under some mild conditions the bisimilarity pseudometric can be computed without exploring the entire final coalgebra (which might be too large) but, in a sense, restricting to its relevant part.

\begin{theorem}[Bisimilarity Pseudometric]
  \label{thm:comp-dist}
  Let $\LiftedFunctor{F}\colon\PMet\to\PMet$ be a lifting of a functor
  $F\colon \Set\to \Set$ such that $F$ has a final coalgebra and
  assume that $\LiftedFunctor{F}$ preserves isometries.
  Let furthermore $c\colon X\to FX$ be an arbitrary coalgebra. Then
  the pseudometric $d_c \colon X \times X \to \reals$ defined in \Cref{lem:lifting-coalg} is
  the bisimilarity pseudometric, i.e., we have $\bd_c = d_c$.
\end{theorem}

\begin{proof}
  Let $z \colon Z \to FZ$ be the final coalgebra of $F$ in $\Set$ and
  let $\final{\cdot}_c\colon X \to Z$ be the unique function such that
  $z \circ \final{\cdot}_c = F\final{\cdot}_c \circ c$. By \Cref{lem:lifting-coalg}, if
  $d_z$ and $d_c$ are the least fixpoint distances on $Z$ and $X$
  respectively, $\final{\cdot}_c\colon (X,d_c) \nonexpansiveTo (Z, d_z)$ is an isometry. Moreover, by \Cref{thm:final-coalgebra},
  $z\colon (Z, d_z) \nonexpansiveTo \LiftedFunctor{F} (Z, d_z)$ is the
  final coalgebra for $\LiftedFunctor{F}$. Therefore
  $\bd_c(x,y) = d_z\big(\final{x}_c,\final{y}_c\big) = d_c(x,y)$.
\end{proof}

We now want to find conditions which ensure that the behavioral distance is a proper metric. To this aim, we proceed by recalling the final coalgebra construction via the final chain which was first presented in the dual setting (free/initial algebra).

\begin{definition}[Final Chain Construction  \cite{Ada74}]
	\label{def:final-chain}
	\index{final chain}
	\index{connection morphism}
	Let $\cat{C}$ be a category with terminal object $\one$ and limits of ordinal-indexed cochains. For any functor $F\colon \cat{C} \to \cat{C}$ the \emph{final chain} consists of objects $W_i$ for all ordinals $i$ and \emph{connection morphisms} $p_{i,j} \colon W_j \to W_i$ for all ordinals $i \le j$. The objects are defined as $W_0 := \one$, $W_{i+1} := FW_i$ for all ordinals $i$, and $W_j := \lim_{i<j}W_i$ for all limit ordinals $j$. The morphisms are determined by $p_{0,i} :=\ !_{W_i}\colon W_i \to \one$, $p_{i,i} = \id_{W_i}$ for all ordinals $i$, $p_{i+1,j+1} := Fp_{i,j}$ for all ordinals $i < j$ and if $j$ is a limit ordinal the $p_{i,j}$ are the morphisms of the limit cone. They satisfy $p_{i,k} = p_{i,j} \circ p_{j,k}$ for all ordinals $i \leq j \leq k$. We say that the chain \emph{converges} in $\lambda$ steps if $p_{\lambda,\lambda+1}\colon W_{\lambda+1} \to W_\lambda$ is an isomorphism. 
\end{definition}

This construction does not necessarily converge (e.g. for the unrestricted powerset functor $\Powerset$ on $\Set$), but if it does, we always obtain a final coalgebra. 

\begin{theorem}[Final Coalgebra via the Final Chain \cite{Ada74}] 
	Let $\cat{C}$ be a category with terminal object $\one$ and limits of ordinal-indexed cochains. If the final chain of a functor $F\colon \cat{C} \to \cat{C}$ converges in $\lambda$ steps then $p_{\lambda, \lambda+1}^{-1}\colon W_\lambda \to FW_\lambda$ is the final coalgebra.
\end{theorem}

We now show under which circumstances the least fixpoint pseudometric $d_z$ is a metric and how our construction relates to the construction of the final chain.

\begin{theorem}[Final Coalgebra Metric]
  \label{thm:d-omega-is-metric}
	Let $\LiftedFunctor{F}\colon\PMet\to\PMet$ be a lifting of a functor $F\colon \Set\to \Set$ which has a final coalgebra $z\colon Z\to FZ$. Let $d_z \colon X \times X \to \reals$ be the least fixpoint pseudometric in \Cref{lem:lifting-coalg}. If $\LiftedFunctor{F}$ preserves metrics, and the final chain for $F$ converges then $d_z$ is a metric, i.e., all for $z,z'\in Z$ we have $d_z(z,z') = 0\iff z=z'$. This implies that for every coalgebra $c \colon X \to FX$ and all $x, x' \in X$ we have 
	\begin{align*}
		\final{x}_c = \final{x'}_c \iff \bd_c(x,x') = 0\,.
	\end{align*}
\end{theorem}
\begin{proof}

	Let $W_i, p_{i,j} \colon W_j \to W_i$ be as in
        \Cref{def:final-chain}. Define a pseudometric over each
        element of the chain $e_i\colon W_i \times W_i \to \reals$ as
        follows: $e_0 \colon \one \times \one \to \reals$ is the
        (unique!) zero metric on $\one$, $e_{i+1} :=
        \LiftedMetric{F}{e_i}\colon W_{i+1} \times W_{i+1} \to \reals$
        for all ordinals $i$ and $e_j := \sup_{i<j} e_i\circ
        (p_{i,j}\times p_{i,j}) \text{ }$ for all limit ordinals
        $j$. Since the functor preserves metrics $e_{i+1}$ is a metric
        if $e_i$ is. Given a limit ordinal $j$, it follows from Lemma~\ref{le:lattice-pseudo} that $e_j$ is a pseudometric provided that all the $e_i$ with $i < j$ are pseudometrics. To see that $e_j$ is also a metric when all $e_i$ with $i<j$ are metrics we proceed as follows: Suppose $e_j(x,y) = 0$ for some $x,y \in W_j$, then we know that for all $i<j$ we must have $e_i\big(p_{i,j}(x),p_{i,j}(y)\big)=0$ and thus $p_{i,j}(x) = p_{i,j}(y)$ because the $e_i$ are metrics. Since the cone $(W_j \stackrel{p_{i,j}}{\to} W_i)_{i < j}$ is by definition a limit in $\Set$ we can now conclude that $x=y$. This is due to the universal property of the limit: Let us assume $x \not =y$ then for the cone $(\set{x,y} \stackrel{f_i}{\to} W_i)_{i<j}$ with $f_i(x) := p_{i,j}(x)$, and also $f_i(y) := p_{i,j}(y)=p_{i,j}(x)$ there would have to be a unique function $u\colon \set{x,y} \to W_j$ satisfying $p_{i,j}\circ u = f_i$. However, for example $u, u'\colon \set{x,y} \to W_j$ where $u(x) = u(y) = x$ and $u'(x) = u'(y) = y$ are distinct functions satisfying this commutativity which is a contradiction to the uniqueness. Thus our assumption ($x \not= y$) is false and $e_j$ is indeed a metric.

  Consider now for each ordinal $i$, a function
  $\eta_i \colon Z \to W_i$ defined as follows. For $i=0$ we let
  $\eta_0 \colon Z \to W_0$ be the unique mapping to the terminal
  object. For each ordinal $i$ we define
  $\eta_{i+1} = F\eta_i \circ z$, and for a limit ordinal $j$, we let
  $\eta_j$ be the unique mapping into the limit $\lim_{i<j} W_i$. Observe that for each ordinal $i$, it holds that
  \begin{equation}
    \label{eq:coalg}
    \eta_i = p_{i,i+1} \circ F \eta_i \circ z
  \end{equation}
  For $i=0$, $\eta_0 = p_{0,1} \circ F \eta_{1} \circ z$ is
  the unique map into the final object. For any ordinal $i$, assuming
  \eqref{eq:coalg} we can prove that the same property holds for $i+1$. In fact
  $\eta_{i+1} = F \eta_i \circ z = F(p_{i, i+1} \circ F\eta_i \circ z) \circ z = F(p_{i,i+1} \circ \eta_{i+1}) \circ
  z = F p_{i,i+1} \circ F \eta_{i+1} \circ z = p_{i+1,i+2} \circ F
  \eta_{i+1} \circ z$, as desired. Finally, for a limit ordinal $j$,
  assume the property \eqref{eq:coalg} for $i<j$. Recall that
  $\eta_j \colon Z \to W_j$ is the mediating arrow into the limit,
  i.e., for all ordinals $i<j$ we have $\eta_i = p_{i,j} \circ \eta_j$
  and thus
  $F\eta_i = F p_{i,j} \circ F \eta_j = p_{i+1, j+1} \circ F
  \eta_j$. Therefore
  \begin{align*}
    & p_{i,j} \circ (p_{j,j+1} \circ F\eta_j \circ z)  \\
    & \quad = p_{i,j+1} \circ F\eta_j \circ z = p_{i,i+1} \circ p_{i+1,j+1} \circ F\eta_j \circ z
    & \mbox{[by definition of $p_i,j$]}\\
    & \quad = p_{i,i+1} \circ F\eta_i \circ z
    & \mbox{[by the above observation]}\\
    & \quad = \eta_i 
    & \mbox{[by inductive hypothesis]}
  \end{align*} 
  Hence $p_{j,j+1} \circ F\eta_j \circ z$ is a mediating morphisms
  into the limit and thus it must coincide with $\eta_j$, namely
  $\eta_j = p_{j,j+1} \circ F\eta_j \circ z$ as desired.
   
  Also observe that each
  $\eta_i \colon (Z, d_z) \nonexpansiveTo (W_i, e_i)$ is
  nonexpansive. This is trivially true for $i=0$, since $e_0$ is the
  zero pseudometric. For any ordinal $i$, assuming
  $\eta_i \colon (Z, d_z) \nonexpansiveTo (W_i, e_i)$ nonexpansive,
  we obtain that
  $F \eta_i \colon (FZ, \LiftedMetric{F}{d_z}) \nonexpansiveTo (FW_i,
  \LiftedMetric{F}{e_i}) = (W_{i+1}, e_{i+1})$ is nonexpansive, hence
  $\eta_{i+1} = F\eta_i \circ z$ is nonexpansive. For a limit ordinal
  $j$,
\begin{align*}
    &  e_j \circ (\eta_j \times \eta_j) \\
    & \quad = (\sup_{i<j} (e_i \circ (p_{i,j} \times p_{i,j}))) \circ (\eta_j \times \eta_j) 
    & \mbox{[definition of $e_j$]}\\
    & \quad = \sup_{i<j} (e_i \circ (p_{i,j} \circ \eta_j \times p_{i,j} \circ \eta_j))
    & \mbox{[$\sup$ componentwise]}\\
    & \quad = \sup_{i<j} (e_i \circ (\eta_i \times \eta_i)) 
    & \mbox{[since $\eta_j$ is the mediating arrow into the limit]}\\
    & \quad = \sup_{i < j} d_Z = d_Z
    & \mbox{[nonexpansiveness of $\eta_i\colon (Z, d_Z) \nonexpansiveTo (W_i, e_i)$]}
  \end{align*}  
  which means that $\eta_j \colon (Z, d_z) \nonexpansiveTo (W_j, e_j)$
  is nonexpansive.

  Now, let $\lambda$ be the ordinal for which the final chain
  converges in $\Set$. We know that
  $p_{\lambda,\lambda+1}^{-1}\colon W_{\lambda} \to F W_{\lambda}$ is the
  final coalgebra in $\Set$. Since also $z \colon Z \to FZ$ is, and by
  \eqref{eq:coalg} $\eta_\lambda \colon Z \to W_{\lambda}$ is a
  coalgebra morphism, we deduce that $\eta_\lambda$ is an isomorphism
  of sets. Moreover, we proved that
  $\eta_{\lambda} \colon (Z, d_z) \nonexpansiveTo (W_\lambda,
  e_\lambda)$ is nonexpansive. Recalling that $e_\lambda$ is a metric,
  we conclude that also $d_z$ is so. In fact, if $w, w' \in Z$ are
  such that $d_z(w,w') = 0$, then
  $e_\lambda(\eta_\lambda(w),\eta_\lambda(w')) \leq d_z(w,w')
  =0$. Therefore $\eta_\lambda(w) = \eta_\lambda(w')$, and thus
  $w=w'$ since $\eta_\lambda$ is an iso in $\Set$ and thus bijective. Since we know now that $d_z$ is a metric, we also know that $0=\bd_c(x,x') = d_z(\final{x}_c, \final{x'}_c) $ holds if and only if $\final{x}_c = \final{x'}_c$.
\end{proof}

We will now get back to the examples studied at the beginning of this paper (\Cref{ex:probabilistic-1,ex:metric-ts-1}, \cpageref{ex:probabilistic-1,ex:metric-ts-1}) and discuss in which sense they are instances of our framework.

\begin{example}[Bisimilarity Pseudometric for Probabilistic Systems]
	\label{ex:probabilistic-2}
	In order to model the discounted behavioral distance for purely probabilistic systems as given in \Cref{ex:probabilistic-1} (\cpageref{ex:probabilistic-1}) in our framework, we set $\top=1$ and proceed to lift the following three functors: we first consider the identity functor $\Id$ with evaluation map $\ev_\Id\colon [0,1]\to [0,1]$, $\ev_\Id(z) = c\cdot z$ in order to integrate the discount (\Cref{ex:kant-rubinst}, \cpageref{ex:kant-rubinst}). Then, we take the coproduct with the singleton metric space (\Cref{def:coproduct-bifunctor} and \Cref{lem:coproduct-bifunctor-evfct,lem:coproduct-pseudometric}, pp.~\pageref{def:coproduct-bifunctor} ff.). The combination of the two functors yields the discrete version of the refusal functor of \cite{vBW06}, namely $\LiftedFunctor{R}(X,d) = (X+\one, \widehat{d})$ where $\widehat{d}$ is the coproduct pseudometric taken from \Cref{ex:probabilistic-1}.  Finally, we lift the probability distribution functor $\Distributions$ to obtain $\LiftedFunctor{\Distributions }$ (\Cref{exa:probability-distribution-functor}, \cpageref{exa:probability-distribution-functor}). All functors satisfy the Kantorovich-Rubinstein duality and preserve metrics.
	
	It is easy to see that $\LiftedFunctor{\Distributions}(\LiftedFunctor{R}(X,d))=(\Distributions(X+1), \overline{d})$, where $\overline{d}$ is defined as in \Cref{ex:probabilistic-1}).  Then, the least solution of $d(x,y) = \overline{d}\big(c(x),c(y)\big)$ can be computed as in \Cref{thm:comp-dist}.
\end{example}

\begin{example}[Bisimilarity Pseudometric for Metric Transition Systems]
	\label{ex:metric-ts-2}
	As in \Cref{ex:metric-ts-1} we consider metric transition systems as  coalgebras $c \colon S \to M_1 \times \dots \times M_n \times \PowersetFinite(S)$ where $S$ is a finite set of states and $(M_i,d_i)$ are pseudometric spaces. If, for $i \in \set{1, \dots, n}$, $\pi_i \colon M_1 \times \dots \times M_n \times \PowersetFinite(S) \to M_i$ and $\pi_{n+1}\colon  M_1 \times \dots \times M_n \times \PowersetFinite(S) \to \PowersetFinite(S)$ are the projections of the product, then each state $s \in S$ is assigned a valuation function $[s] \colon \set{1,\dots, n} \to \cup_{r=1}^n M_r$ where, of course, $[s](i) = \pi_i\big(c(S)\big)$, and the set $\pi_{n+1}\big(c(S)\big)$ of successor states.
	
	To obtain behavioral distances for metric transition systems using our framework we set $\top = \infty$. Moreover, analogously to the product bifunctor of \Cref{def:product-bifunctor} we can equip the product multifunctor $P\colon \Set^{n+1} \to \Set$ with the evaluation function $\ev_P \colon [0,\infty]^{n+1} \to [0,\infty]$ where  $\ev_P(r_1,\dots,r_{n+1}) = \max\set{r_1,\dots,r_{n+1}}$ which is a natural generalization of the function presented in \Cref{lem:evfct-product-bifunctor}. As in that lemma, this function is well-behaved in the sense of \Cref{def:well-behaved-multi} and analogously to \Cref{lem:product-pseudometrics} we can easily see that duality holds and we obtain the categorical product pseudometric, i.e., for given pseudometric spaces $(X_i,d_i)$ the new pseudometric $\LiftedMetric{P}{d} \colon (X_1 \times \dots \times X_{n+1})^2 \to [0,\infty], \quad \LiftedMetric{P}{d} = \max\set{d_1,\dots,d_{n+1}}$. Let $\LiftedFunctor{P}$ be the corresponding lifted multifunctor. We instantiate the given pseudometric spaces $(M_{i},d_{i})$ as fixed parameters and obtain the endofunctor $\LiftedFunctor{F}\colon \PMet \to \PMet$ with
	\begin{align*}
		\LiftedFunctor{F}(X,d) = \LiftedFunctor{P}\big((M_{1},d_{1}),\dots,(M_{n},d_{n}) \times \LiftedFunctor{\PowersetFinite}(X,d)\big)
	\end{align*}
	where $\LiftedFunctor{\PowersetFinite}$ is  the lifting of the powerset functor using the evaluation function $\max\colon \PowersetFinite([0,\infty]) \to [0,\infty]$ with $\max \emptyset = 0$ as presented in \Cref{exa:hausdorff}. Then, via \Cref{thm:comp-dist}, we obtain exactly the least solution of
	\begin{align*}
		d(s,t) = \max\set{\LiftedFunctor{\mathit{pd}}([s],[t]),\max_{s'\in\tau(s)}\min_{t'\in\tau(t)}	d(s',t'),\max_{t'\in\tau(t)}\min_{s'\in\tau(s)} d(s',t') } 
	\end{align*} 
	as in \eqref{eq:Hausd} in \Cref{ex:metric-ts-1}. Except for the fact that we allow $+\infty$ and consider only undirected (symmetric) pseudometrics, this is exactly the branching distance $\mathit{bd}^{Ss}$ for metric transition systems \cite[Def.~13]{dAFS09}.
\end{example}
\section{Compositionality of Liftings}
\label{sec:compositionality}
In the remainder of this paper we want to turn our attention to trace
(or linear-time) pseudometrics. Since we plan to apply the generalized powerset construction we will have to lift not only functors but also monads from $\Set$ to $\PMet$. As a preparation for that we now study \emph{compositionality} of functor liftings, i.e., we set off to identify some sufficient conditions ensuring $\LiftedFunctor{F}\,\LiftedFunctor{G} = \LiftedFunctor{FG}$. Unfortunately, this seems to be a quite difficult question in this general setting so our main result only deals with the Wasserstein lifting and requires the existence of optimal couplings. However, whenever it can be applied it allows us to reason modularly and, consequently, to simplify the proofs needed for the treatment of our examples. 

As further preparation for the trace pseudometric we will also consider two examples involving the finite powerset monad where optimal couplings do not always exist and manually prove that compositionality holds for these specific cases. 

\subsection{Compositionality for Endofunctors}
Given evaluation functions $\ev_F$ and $\ev_G$, we can easily construct an evaluation function for the composition $FG$ as follows.

\begin{definition}[Composition of Evaluation Functions]
	\index{composed evaluation function}
	\index{evaluation function!composition}
	Let $F$ and $G$ be endofunctors on $\Set$ with evaluation functions $\ev_F$ and $\ev_G$. We define the composition of $\ev_F$ and $\ev_G$ to be the  evaluation function $\ev_F \ast \ev_G\colon FG[0,\top] \to [0,\top]$ for the composed functor $FG$ as $\ev_F \ast \ev_G:= \EvaluationFunctor{F}\ev_G= \ev_F \circ F\ev_G$. 
\end{definition}

Whenever $F$ and $G$ preserve weak pullbacks well-behavedness is inherited. 

\begin{theorem}[Well-Behavedness of Composed Evaluation Function]
	\label{prop:comp:well-behaved}
	Let $F$, $G$ be endofunctors on $\Set$ with evaluation functions $\ev_F$,  $\ev_G$. If both functors preserve weak pullbacks and both evaluation functions are well-behaved then also $\ev_F \ast \ev_G$ is well-behaved.
\end{theorem}

We will split the proof into two technical lemmas. The first of these just summarizes some useful rules for calculations.

\begin{lemma}
	\label{lem:comp1}
	Let $F,G$ be endofunctors on $\Set$ with evaluation functions $\ev_F, \ev_G$ and $a:= \langle G\pi_1, G\pi_2\rangle$ (i.e., the unique mediating arrow into the product $GX \times GX$, where $\pi_i \colon X^2 \to X$ are the projections) and $(X,d)$ an arbitrary pseudometric space. Then the following holds.
	\begin{enumerate}
		\item We have $\EvaluationFunctor{G}d \ge \Wasserstein{G}{d} \circ a $ and if $\ev_G$ satisfies \Cref{W1,W2} we also have $\EvaluationFunctor{G}d \ge \Kantorovich{G}{d} \circ a$
		\item $\forall t_1,t_2 \in FGX: \quad t \in \Couplings{FG}(t_1,t_2) \implies Fa(t) \in \Couplings{F}(t_1,t_2)$.
		\item If $F$ and $G$ preserve weak pullbacks then so does $FG$.
		\item For any $f \in \Set/[0,\top]$ we have $\EvaluationFunctor{FG}f = \EvaluationFunctor{F}(\EvaluationFunctor{G}f)$.
	\end{enumerate}
\end{lemma}
\begin{proof}
	We first of all observe that $a$ is the unique mediating arrow into the product $GX \times GX$ satisfying $\tau_i \circ a = G\pi_i$ for the projections $\pi_i \colon X^2 \to X$ and $\tau_i\colon (GX)^2 \to GX$.
\begin{enumerate}
		\item Let $s \in G(X\times X)$ and define $s_i:=G\pi_i(s) = \tau_i \circ a(s)$. Then by definition $s \in \Couplings{G}(s_1,s_2)$ and we conclude $\EvaluationFunctor{G}d(s) \geq \inf\{\EvaluationFunctor Gd(s') \mid s' \in \Couplings{G}(s_1,s_2)\} = \Wasserstein{G}{d}(s_1,s_2) = \Wasserstein{G}{d}\big(\tau_1 \circ a(s),\tau_2 \circ a(s)\big) = \Wasserstein{G}{d} \circ a (s)$. Since $\Wasserstein{G}{d} \geq \Kantorovich{G}{d}$ as shown in \Cref{prop:wasserstein-vs-kantorovich}, the statement follows.
		\item We compute $F\tau_i \big(Fa(t)\big) = F(\tau_i \circ a)(t) = F(G\pi_i)(t) = FG\pi_i = t_i$. 
		\item This is indeed clear by definition.
		\item Let $f\colon X \to [0,\top]$, then $\EvaluationFunctor{FG}f = \ev_F \ast \ev_G \circ FGf = \ev_F \circ F\ev_G \circ FGf = \ev_F \circ F(\ev_G \circ Gf) = \EvaluationFunctor{F}(\EvaluationFunctor{G}f)$.\qedhere
	\end{enumerate}
\end{proof}

\noindent We can now use these calculations to prove the second lemma which already finishes the proof of the inheritance of well-behavedness.

\begin{lemma}
	\label{lem:comp:evfct}
	Let $F$, $G$ be functors with evaluation functions $\ev_F$, $\ev_G$.
\begin{enumerate}
		\item If $\EvaluationFunctor{F}$ and $\EvaluationFunctor{G}$ are monotone (\Cref{W1}), then so is the evaluation functor $\EvaluationFunctor{FG}$ with respect to the composed evaluation function $\ev_F \ast \ev_G$.
		\item If $G$ preserves weak pullbacks, $\ev_G$ is well-behaved and $\EvaluationFunctor{F}$ is monotone (\Cref{W1}) then $\ev_F \ast \ev_G$ satisfies \Cref{W2} of \Cref{def:well-behaved}.
		\item If $F$ preserves weak pullbacks and $\ev_F,\ev_G$ satisfy \Cref{W3} of well-behavedness, then also $\ev_F \ast \ev_G$ satisfies \Cref{W3} of \Cref{def:well-behaved}.	
	\end{enumerate}
\end{lemma}
\begin{proof}
	\begin{enumerate}
		\item Let $f,g\colon X \to [0,\top]$ with $f \leq g$, then by monotonicity of $\ev_G$ we have $\EvaluationFunctor{G}f \leq \EvaluationFunctor{G}g$ and using monotonicity of $\ev_F$ we get $\EvaluationFunctor{FG}f = \EvaluationFunctor{F}(\EvaluationFunctor{G}f) \leq \EvaluationFunctor{F}(\EvaluationFunctor{G}g) = \EvaluationFunctor{FG}g$. 
		
		\item Let $\pi_i\colon X^2 \to X$ be the projections of the product. For $t \in FG([0,\top]^2)$ we define $t_i := FG\pi_i(t) \in FG[0,\top]$. By definition $t \in \Couplings{FG}(t_1,t_2)$ so \Cref{lem:comp1} (\cpageref{lem:comp1}) tells us $Fa(t) \in \Couplings{F}(t_1,t_2)$ for $a:=\langle G\pi_1, G\pi_2\rangle$. Moreover, since the evaluation function $\ev_G\colon (G[0,\top],\Kantorovich{G}{d_e}) \nonexpansiveTo ([0,\top],d_e)$ is nonexpansive (by definition of the Kantorovich pseudometric), we can apply \Cref{prop:wasserstein-upper-bound} (\cpageref{prop:wasserstein-upper-bound}) to obtain the inequality $d_e\big((\ev_F \ast \ev_G)(t_1), \ev_F \ast \ev_G(t_2)\big) = d_e\big(\EvaluationFunctor{F}\ev_G(t_1), \EvaluationFunctor{F}\ev_G(t_2)\big) \leq \EvaluationFunctor{F}\Kantorovich{G}{d_e}\big(Fa(t)\big) = \EvaluationFunctor{F}\big(\Kantorovich{G}{d_e} \circ a\big)(t)$. By \Cref{lem:comp1} we have $\Kantorovich{G}{d_e} \circ a \leq \EvaluationFunctor{G}{d_e}$ and using monotonicity of $\EvaluationFunctor{F}$ we can continue our inequality with $\EvaluationFunctor{F}\big(\Kantorovich{G}{d_e} \circ a\big)(t) \leq \EvaluationFunctor{F}\big(\EvaluationFunctor{G}{d_e}\big)(t) = \EvaluationFunctor{FG}d_e(t)$ which concludes the proof.
		\item By \Cref{lem:W3-weak-pb} (\cpageref{lem:W3-weak-pb}) we have to show that the following diagram is a weak pullback.
		\begin{center}
			\begin{diagram}
				\matrix(m)[matrix of math nodes, column sep=80pt, row sep=17pt]{
					FG\set{0} & F\set{0} & \set{0}\\
					FG[0,\top] & F{[0,\top]} & {[0,\top]}\\
				};
				\draw[->] (m-1-1) edge node[below]{$F!_{G\set{0}}$} (m-1-2);
				\draw[->] (m-1-1) edge node[left]{$FGi$} (m-2-1);
				\draw[->] (m-1-2) edge node[right]{$Fi$} (m-2-2);
				\draw[->] (m-2-1) edge node[above]{$F\ev_G$} (m-2-2);
				
				\draw[->] (m-1-2) edge node[below]{$!_{F\set{0}}$} (m-1-3);
				\draw[right hook->] (m-1-3) edge node[right]{$i$} (m-2-3);
				\draw[->] (m-2-2) edge node[above]{$\ev_F$} (m-2-3);
				
				% combined arrows
				\draw[->] (m-1-1) edge[bend left=10] node[above]{$!_{FG\set{0}}$} (m-1-3);
				\draw[->] (m-2-1) edge[bend right=10] node[below]{$\ev_F \ast \ev_G$} (m-2-3);
			\end{diagram}
		\end{center}
		\Cref{lem:W3-weak-pb} tells us that the right square is a weak pullback and since $F$ preserves weak pullbacks also the left square is. The outer part is thus necessarily a weak pullback again yielding by \Cref{lem:W3-weak-pb} that $\ev_F \ast \ev_G$ satisfies \Cref{W3}.\qedhere
	\end{enumerate}
\end{proof}

\noindent This lemma concludes the proof of \Cref{prop:comp:well-behaved}. In the light of this result we know that, whenever we start with weak pullback preserving functors $F$, $G$ along with well-behaved evaluation functions $\ev_F$, $\ev_G$, the Wasserstein distance for $FG$ (with respect to $\ev_F \ast \ev_G$) is always a pseudometric (see \Cref{prop:wasserstein-is-pseudometric}, \cpageref{prop:wasserstein-is-pseudometric}) so we can safely talk about \emph{the} Wasserstein lifting of $FG$ and study compositionality. 

We will now show that a sufficient criterion for compositionality of the Wasserstein lifting is the existence of optimal couplings for $G$. Again we start with a technical lemma which also contains a small statement about the Kantorovich lifting.
\begin{lemma}
	\label{lem:compositionality2}
	Let $F$ and $G$ be endofunctors on $\Set$ together with evaluation functions $\ev_F\colon F[0,\top]\to [0,\top]$, $\ev_G\colon G[0,\top] \to [0,\top]$. We define $\ev_F \ast \ev_G := \ev_F \circ F\ev_G$. Then the following properties hold for every pseudometric space $(X,d)$.
	\begin{enumerate}
		\item $\Kantorovich{FG}{d} \leq \Kantorovich{F}{(\Kantorovich{G}{d})}$.\label{item:leqKantorovich} 
		\item If $\ev_F$ and $\ev_G$ satisfy \Cref{W1,W2} then $\Wasserstein{FG}{d} \geq \Wasserstein{F}{(\Wasserstein{G}{d})}$.\label{item:geqWasserstein}
		\item If for all $t_1,t_2 \in FGX$ there is a function $\nabla (t_1,t_2)\colon \Couplings{F}(t_1, t_2) \to \Couplings{FG}(t_1,t_2)$ such that $\EvaluationFunctor{FG}d \circ \nabla(t_1,t_2) = \EvaluationFunctor{F}\Wasserstein{G}{d}$ then $\Wasserstein{FG}{d} \leq \Wasserstein{F}{(\Wasserstein{G}{d})}$.\label{prop:comp:nabla}
	\end{enumerate}
\end{lemma}
\begin{proof}
	Let $t_1,t_2 \in FGX$. 
	\begin{enumerate}
		\item Recall that $\Kantorovich{G}{d}$ is the smallest pseudometric such that for every nonexpansive function $f\colon (X,d) \nonexpansiveTo ([0,\top],d_e)$ also $\EvaluationFunctor{G}f \colon (GX,\Kantorovich{G}{d}) \nonexpansiveTo ([0,\top],d_e)$ is nonexpansive (see remark in the beginning of \Cref{subsec:kantorovich-lifting} on \cpageref{subsec:kantorovich-lifting}). Moreover, $\EvaluationFunctor{FG}f = \EvaluationFunctor{F}(\EvaluationFunctor{G}f)$ by \Cref{lem:comp1}. Thus 
		\begin{align*}
			\Kantorovich{FG}{d}&(t_1,t_2) = \sup\set{d_e\Big(\EvaluationFunctor{FG}f(t_1),\EvaluationFunctor{FG}f(t_2)\Big) \,\big|\, f \colon (X,d) \nonexpansiveTo ([0,\top],d_e)}\\
			& = \sup\set{d_e\Big(\EvaluationFunctor{F}(\EvaluationFunctor{G}f)(t_1),\EvaluationFunctor{F}(\EvaluationFunctor{G}f)(t_2)\Big) \,\big|\, f \colon (X,d) \nonexpansiveTo ([0,\top],d_e)}\\
			& \leq \sup\set{d_e\Big(\EvaluationFunctor{F}(g)(t_1),\EvaluationFunctor{F}(g)(t_2)\big) \,\big|\, g \colon (GX,\Kantorovich{G}{d}) \nonexpansiveTo ([0,\top],d_e)} = \Kantorovich{F}{(\Kantorovich{G}{d})}(t_1,t_2)
		\end{align*}
		
		\item \Cref{lem:comp1} tells us $\EvaluationFunctor{G}d \geq \Wasserstein{G}{d} \circ a$ and for any coupling $t \in \Couplings{FG}(t_1,t_2)$ we have $Fa(t) \in \Couplings{F}(t_1,t_2)$. Using these facts and the monotonicity of $\EvaluationFunctor{F}$ we obtain:
		\begin{align*}
			\Wasserstein{FG}{d}&(t_1,t_2) = \inf\set{\EvaluationFunctor{FG}d(t) \,\big|\, t\in \Couplings{FG}(t_1,t_2)} = \inf\set{\EvaluationFunctor{F}(\EvaluationFunctor{G}d)(t) \,\big|\, t\in \Couplings{FG}(t_1,t_2)}\\
			& \geq  \inf\set{\EvaluationFunctor{F}(\Wasserstein{G}{d} \circ a)(t) \,\big|\, t\in \Couplings{FG}(t_1,t_2)} =  \inf\set{\EvaluationFunctor{F}\Wasserstein{G}{d} \big(Fa(t)\big) \,\big|\, t\in \Couplings{FG}(t_1,t_2)} \\
			& \geq  \inf\set{\EvaluationFunctor{F}\Wasserstein{G}{d} (t') \,\big|\, t'\in \Couplings{F}(t_1,t_2)} = \Wasserstein{F}{(\Wasserstein{G}{d})}(t_1,t_2)
		\end{align*}
		
		\item Using $\nabla(t_1,t_2)$ we compute
		\begin{align*}
			\Wasserstein{FG}{d}(t_1,t_2) &=  \inf{\set{\EvaluationFunctor{FG}d (t') \,\big|\, t'\in \Couplings{FG}(t_1,t_2)} } \leq \inf{\set{\EvaluationFunctor{FG}d \big(\nabla(t_1,t_2)(t)\big) \,\big|\, t\in \Couplings{F}(t_1,t_2)} } \\
			&= \inf{\set{\EvaluationFunctor{F}\Wasserstein{G}{d} (t) \,\big|\, t\in \Couplings{F}(t_1,t_2)} } =\Wasserstein{F}{(\Wasserstein{G}{d})}(t_1,t_2) 
		\end{align*}
		which concludes the proof.\qedhere		
	\end{enumerate}
\end{proof}

\noindent With this result at hand we can now prove compositionality for the Wasserstein lifting.

\begin{theorem}[Compositionality of the Wasserstein Lifting]
\label{prop:comp}
Let $F,G$ be weak pullback preserving endofunctors on $\Set$ with well-behaved evaluation functions $\ev_F$, $\ev_G$ and $(X,d)$ be a pseudometric space. If for all $t_1,t_2 \in GX$ there is an optimal $G$-coupling  $\gamma(t_1,t_2) \in \Gamma_G(t_1,t_2)$ such that $\Wasserstein{G}{d}(t_1,t_2) = \EvaluationFunctor{G}d\big(\gamma(t_1,t_2)\big)$ then we have the equality $\Wasserstein{FG}{d} = \Wasserstein{F}{(\Wasserstein{G}{d})}$.
\end{theorem}
\begin{proof} 
	From \Cref{lem:compositionality2}.\ref{item:geqWasserstein} we know $\Wasserstein{FG}{d} \geq \Wasserstein{F}{(\Wasserstein{G}{d})}$. We just have to show the other inequality. By our requirement we have a function $\gamma\colon GX \times GX \to G(X\times X)$, such that $\Wasserstein{G}{d} = \EvaluationFunctor{G}d \circ \gamma$. Moreover, let $\pi_i \colon X^2 \to X$ and $\tau_i\colon (GX)^2 \to GX$ be the projections of the product then $\gamma$ satisfies $G\pi_i \circ \gamma = \tau_i$. Given $t_1,t_2 \in FGX$ and $t \in \Couplings{F}(t_1,t_2)$, we define $\nabla(t_1,t_2)(t) = F\gamma(t)$, then this satisfies the conditions of  \Cref{lem:compositionality2}.\ref{prop:comp:nabla}. First, we have $F\gamma(t) \in \Couplings{FG}(t_1,t_2)$ because $FG\pi_i\big(F\gamma(t)\big) = F(G\pi_i \circ \gamma)(t) = F\tau_i(t) = t_i$. Moreover $\EvaluationFunctor{FG}d\big(F\gamma(t)\big) = \ev_F \ast \ev_G \circ F\big(Gd \circ \gamma(t)\big) = \ev_{F} \circ F\ev_G \circ F(Gd \circ \gamma)(t) =\ev_{F} \circ F(\EvaluationFunctor{G}d \circ \gamma)(t)  = \ev_{F} \circ F\Wasserstein{G}{d} (t) = \EvaluationFunctor{F}\Wasserstein{G}{d}(t)$ so \Cref{lem:compositionality2}.\ref{prop:comp:nabla} yields $\Wasserstein{FG}{d} \leq \Wasserstein{F}{(\Wasserstein{G}{d})}$.
\end{proof}
This criterion will sometimes turn out to be useful for our later results. Nevertheless it provides just a sufficient condition for compositionality as the following examples show.

\begin{example}[Compositionality for the Distribution Functor]
	\label{exa:comp}
	We consider the distribution functor (with finite support) $\DistributionsFinite$ with the evaluation function defined in \Cref{exa:probability-distribution-functor}. For any pseudometric space $(X,d)$ we have $\Wasserstein{\DistributionsFinite\DistributionsFinite}{d} = \Wasserstein{\DistributionsFinite}{\big(\Wasserstein{\DistributionsFinite}{d}\big)}$ by~\Cref{prop:comp} because optimal couplings always exist.
\end{example}

\begin{example}[Compositionality for the Finite Powerset Functor]
	\label{exa:powerset-compositional}
	We consider the finite powerset functor $\PowersetFinite$ with the evaluation function defined in \Cref{ex:evaluation-max}. We claim that for any pseudometric space $(X,d)$ we have $\Wasserstein{\PowersetFinite\PowersetFinite}{d} = \Wasserstein{\PowersetFinite}{\big(\Wasserstein{\PowersetFinite}{d}\big)}$ although $\PowersetFinite$-couplings do not always exist.\label{exa:comp:3} To verify this, we recall from \Cref{lem:compositionality2}.\ref{item:geqWasserstein} that 
	\begin{align}
		\Wasserstein{\PowersetFinite\PowersetFinite}{d} \geq \Wasserstein{\PowersetFinite}{\left(\Wasserstein{\PowersetFinite}{d}\right)}\label{eq:powfin0}
	\end{align} holds. We now show that we always have equality. Let $(X,d)$ be a pseudometric space and $T_1,T_2 \in \PowersetFinite\PowersetFinite X$. We distinguish three cases:
	\begin{enumerate}
		\item If $T_1=T_2=\emptyset$ we know by reflexivity that both sides of \eqref{eq:powfin0} are $0$.
		\item If $T_1=\emptyset\not=T_2$ or $T_1\not=\emptyset=T_2$ we know from \Cref{exa:hausdorff} that $\Couplings{\PowersetFinite}(T_1,T_2) = \emptyset$ and therefore $\Wasserstein{\PowersetFinite}{\big(\Wasserstein{\PowersetFinite}{d}\big)}(T_1,T_2) = \top$ and thus \eqref{eq:powfin0} is necessarily an equality because the left hand side can never exceed $\top$. 
		\item Let $T_1,T_2 \not=\emptyset$. We know from \Cref{exa:hausdorff} that we have an optimal coupling $T^* \in \Couplings{\PowersetFinite}(T_1,T_2)$, say $T^* = \set{(V_{j1}, V_{j2}) \in \PowersetFinite X \times \PowersetFinite X \mid j \in J}$ for a suitable index set $J$. Then, using the projections $\pi_i \colon (\PowersetFinite X)^2 \to \PowersetFinite X$, we have $T_i = \PowersetFinite \pi_i (T^*) = \pi_i[T^*] = \set{\pi_i\big((V_{j1}, V_{j2})\big)\ \big|\  j \in J} = \set{V_{ji} \mid j \in J}$. By optimality we thus have:
		\begin{align}
			\Wasserstein{\PowersetFinite}{\left(\Wasserstein{\PowersetFinite}{d}\right)}(T_1,T_2) = \EvaluationFunctor{\PowersetFinite}{\Wasserstein{\PowersetFinite}{d}}(T^*) = \max \Wasserstein{\PowersetFinite}{d}[T^*] = \max_{j \in J} \Wasserstein{\PowersetFinite}{d}(V_{j1}, V_{j2})\,. \label{eq:powfin1}
		\end{align}
	We will make another case distinction:
	\begin{enumerate}
		\item If there is an index $j' \in J$ such that $\Couplings{\PowersetFinite}(V_{j'1}, V_{j'2}) = \emptyset$, we have $\Wasserstein{\PowersetFinite}{d}(V_{j'1}, V_{j'2}) = \top $ and by \eqref{eq:powfin1} also $\Wasserstein{\PowersetFinite}{\big(\Wasserstein{\PowersetFinite}{d}\big)}(T_1,T_2) = \top$ which again shows that \eqref{eq:powfin0} is an equality. 
		
		\item Otherwise we can take optimal couplings $V_j^* \in \Couplings{\PowersetFinite}(V_{j1}, V_{j2})$ (see \Cref{exa:hausdorff}, \cpageref{exa:hausdorff}). Continuing \eqref{eq:powfin1} we have
		\begin{align}
		\Wasserstein{\PowersetFinite}{\left(\Wasserstein{\PowersetFinite}{d}\right)}(T_1,T_2) = \max_{j \in J} \EvaluationFunctor{\PowersetFinite}d(V_j^*) = \max_{j \in J} \max d[V_j^*] \label{eq:powfin2}
		\end{align}
		
		We define $T := \set{V_j^* \mid j \in J} \in \PowersetFinite\PowersetFinite(X \times X)$ and calculate for the projections $\pi_i \colon X \times X \to X$
		\begin{align*}
		\PowersetFinite\PowersetFinite \pi_i (T) = \PowersetFinite\pi_i[T] = \set{\PowersetFinite\pi_i(V_j^*) \mid j \in J} = \set{V_{ji} \mid j \in J} = T_i
		\end{align*}
		and thus $T \in \Couplings{\PowersetFinite\PowersetFinite}(T_1,T_2)$. Moreover we have
		\begin{align}
			\Wasserstein{\PowersetFinite\PowersetFinite}{d}(T_1,T_2)  &\leq \EvaluationFunctor{\PowersetFinite\PowersetFinite}d(T) =\max\Big(\big(\PowersetFinite \max\big)\big(\PowersetFinite\PowersetFinite d(T)\big)\Big) \nonumber\\
			&= \max \big( \max \big[ (\PowersetFinite d)[T] \big] \big) = \max \left( \max\big[\set{d[V_j^*]\mid j \in J}\big] \right) \nonumber\\
			&=\max\left(\set{\max d[V_j^*]\mid j \in J}\right) = \max_{j \in J}\max d[V_j^*] \label{eq:powfin3}\,.
		\end{align}
		Thus using this, \eqref{eq:powfin2} and \eqref{eq:powfin0} we conclude that 
		\begin{align*}
		\Wasserstein{\PowersetFinite\PowersetFinite}{d}(T_1,T_2) \leq \max_{j \in J}\max d[V_j^*] = \Wasserstein{\PowersetFinite}{\left(\Wasserstein{\PowersetFinite}{d}\right)}(T_1,T_2) \leq \Wasserstein{\PowersetFinite\PowersetFinite}{d}(T_1,T_2)
		\end{align*}
		which proves equality also in this last case. 
	\end{enumerate}
	\end{enumerate}
\end{example}

\noindent We conclude our study of compositionality for endofunctors with another example for which one again has to show compositionality separately. This result will later turn out to be helpful to obtain trace pseudometrics for nondeterministic automata. The proof follows closely the approach used in the previous example, hence it is omitted.

\begin{example}
	\label{exa:m2-comp}
	As in \Cref{exa:wasserstein-machine} (\cpageref{exa:wasserstein-machine}) we equip the machine functor with the evaluation function $\ev_{M_\two} \colon \two \times [0,1]^A\to [0,1]$, $(o,s) \mapsto c \cdot \ev_I(s)$ where $c \in \,]0,1]$ and $\ev_I$ is one of the evaluation functions for the input functor from \Cref{exa:wasserstein-input-functor}. Moreover, for the powerset functor we use the maximum as evaluation function (see \Cref{ex:evaluation-max}). 
	
	Although couplings for $M_\two$ do not always exist, one simply can adapt \cite[Ex.~5.6.9]{Ker16} the approach employed in \Cref{exa:powerset-compositional} to show that we have $\Wasserstein{\PowersetFinite M_\two}{d} = \Wasserstein{\PowersetFinite}{\left(\Wasserstein{M_\two}{d}\right)}$.
\end{example}

\subsection{Compositionality for Multifunctors}
We conclude the analysis of compositionality with a short explanation on how our theory extends to multifunctors. 

For $n \in \N$ we denote by $[n] := \{1,\dots,n\} \subseteq \N$ the set of all positive natural numbers less than or equal to $n$. Now let $n_i \in \N$ for all $i \in[n]$ and $F\colon \Set^n \to \Set$ and $G_i\colon \Set^{n_i} \to \Set$ (for $i \in [n]$) be multifunctors with evaluation functions $\ev_F\colon F({[0,\top]}^n) \to [0,\top]$ and $\ev_{G_i}\colon G_i([0,\top]^{n_i}) \to [0, \top]$. We define $N:=\sum_{i=1}^n n_i$ and define the functor 
\begin{align*}
H := F \circ \prod_{i=1}^nG_i = F \circ (G_1 \times G_2 \times \dots \times G_n) \colon \Set^N \to \Set
\end{align*}
Then we can define the evaluation function $\ev_H\colon H([0,\top]^N) \to [0,\top]$ by
\begin{align*}
\ev_H := \ev_F \circ F(\ev_{G_1}, \ev_{G_2}, \dots, \ev_{G_n})\,.
\end{align*}
In this setting, compositionality of the lifting means that whenever we have $N$ pseudometric spaces $(X_i,d_i)$ the pseudometric $\LiftedMetric{H}{(d_1,\dots, d_N)}$ is equal to
\begin{align*}
\LiftedMetric{F}{\left(\LiftedMetric{G_1}{(d_1,\dots,d_{n_1})}, \LiftedMetric{G_2}{(d_{n_1+1},\dots,d_{n_1+n_2})}\dots, \LiftedMetric{G_n}{(d_{N-n_n+1}, \dots, d_N)}\right)}\,.
\end{align*}
In the upcoming examples we will just use the Wasserstein lifting and we only have the following two cases:
\begin{enumerate}
	\item $n=1$, $n_1=2$ so that $F\colon \Set \to \Set$ is an endofunctor with evaluation function $\ev_F\colon F[0,\top] \to [0,\top]$ and $G \colon \Set^2 \to \Set$ is a bifunctor with evaluation function $\ev_G\colon G([0,\top],[0,\top]) \to [0,\top]$. Then we have $N=n_1=2$ and obtain the bifunctor $H = F \circ G \colon \Set^2 \to \Set$ with evaluation $\ev_H = \ev_F \circ F\ev_G \colon FG([0,1],[0,1]) \to [0,1]$. Compositionality means that for an two pseudometric spaces $(X_1,d_1)$, $(X_2,d_2)$ we have $\Wasserstein{H}{(d_1,d_2)} = \Wasserstein{F}{\big(\Wasserstein{G}{(d_1,d_2)}\big)}$.
	\item $n=2$, $n_1=n_2 = 1$ so that $F\colon \Set^2 \to \Set$ is a bifunctor with evaluation function $\ev_F \colon F([0,\top],[0,\top]) \to [0,\top]$ and $G_1,G_2\colon \Set \to \Set$ are endofunctors with evaluations $\ev_{G_i}\colon G_i[0,\top] \to [0,\top]$. Then we have $N=n_1+n_2 = 1 + 1=2$ and obtain the bifunctor $H = F \circ (G_1 \times G_2) \colon \Set^2 \to \Set$ with evaluation $\ev_H = \ev_F \circ F(\ev_{G_1}, \ev_{G_2})\colon F(G_1[0,\top],G_2[0,\top]) \to [0,\top]$. Compositionality means that for an two pseudometric spaces $(X_1,d_1)$, $(X_2,d_2)$ we have $\Wasserstein{H}{(d_1,d_2)} = \Wasserstein{F}{(\Wasserstein{G}{d_1},\Wasserstein{G_2}{d_2})}$.
\end{enumerate}

\noindent The results presented for endofunctors work analogously in the multifunctor case (the proofs can be transferred almost verbatim), so we do not explicitly present them here. Instead, we will use compositionality to obtain the machine bifunctor. 

\begin{example}[Machine Bifunctor]
	\label{exa:machinebifunctor}
	\index{machine bifunctor}
	Let $I = \_^A$ be the input functor, $\Id$ the identity endofunctor on $\Set$ and $P$ be the product bifunctor of \Cref{def:product-bifunctor}. The \emph{machine bifunctor} is the composition $M := P \circ (\Id \times I)$, i.e., the bifunctor $M\colon \Set^2 \to \Set$ with $M(B,X) := B \times X^A$. We compute the composed evaluation function which, of course, depends on the evaluation functions for $P$ and $I$ (for $\Id$ we always take $\id_{[0,\top]})$. Let $(o,s) \in [0,\top] \times [0,\top]^A$, then
	$\ev_M(o,s) = \ev_P \circ P(\id_{[0,\top]}, \ev_I)  (o,s) = \ev_P \circ (\id_{[0,\top]} \times \ev_I) (o,s) = \ev_P\big(o,\ev_I(s)\big)$.
	By instantiating $\ev_P$ and $\ev_I$ as in the table below
        (see also \Cref{exa:wasserstein-input-functor} and
        \Cref{lem:product-pseudometrics}), we obtain the corresponding
        evaluation functions $\ev_M\colon [0,\top] \times [0,\top]^A
        \to [0,\top]$. They are well-behaved since all involved
        functors preserve weak pullbacks and for $\Id$ and $I$ there
        are unique (thus optimal) couplings so we have
        compositionality by a multifunctor equivalent to
        \Cref{prop:comp:well-behaved}. \new{Again we introduce
          constants $c_1,c_2$ as weighting factors and we impose
          certain requirements to ensure well-behavedness. In
          particular we have to require that the distance does not
          exceed the upper bound $\top$.}
	\begin{center}\begin{tabular}{c|c|c|c}
	Parameters & $\ev_P(r_1,r_2)$ & $\ev_I(s)$ & $\ev_M(o,s)$\\
	\hline
	$c_1,c_2 \in\,]0,1]$ & $\max \set{c_1r_1,c_2r_2}$ & $\max\limits_{a \in A}s(a)$& $\max\Big\{c_1o,c_2\max\limits_{a \in A}s(a)\Big\}$\\
	$\begin{matrix}c_1,c_2 \in\,]0,1], \\c_1+c_2 \leq 1\end{matrix}$ & $c_1 x_1 + c_2 x_2$ & $|A|^{-1}\sum\limits_{a \in A}s(a)$ &$c_1o + c_2|A|^{-1} \sum\limits_{a \in A}s(a)$\\
	$\begin{matrix}c_1,c_2 \in \,]0,\infty[, \\\top=\infty\end{matrix}$ & $c_1 x_1 + c_2 x_2$ & $\sum\limits_{a \in A}s(a)$ &$c_1o + c_2 \sum\limits_{a \in A}s(a)$
	\end{tabular}\end{center}
	Now let $(B,d_B)$, $(X,d)$ be pseudometric spaces. For any $t_1,t_2 \in M(B,X)$ with $t_i = (b_i,s_i) \in B \times X^A$ the unique and therefore necessarily optimal coupling is $t := (b_1,b_2, \langle s_1,s_2 \rangle)$. We compute the Wasserstein distance
	\begin{align*}
		\Wasserstein{M}{(d_B,d)}(t_1,t_2) &= \EvaluationFunctor{M}(d_B,d) (t) = \ev_M \circ M(d_B,d) (t) \\
		&= \ev_M \circ \left(d_B \times d^A\right) \left(b_1,b_2,\langle s_1,s_2 \rangle \right) =\ev_M\big(d_B(b_1,b_2), d \circ \langle s_1, s_2\rangle\big)\,.
	\end{align*}
	We obtain in the first case $\Wasserstein{M}{(d_B,d)}(t_1,t_2) = \max\set{c_1d_B(b_1,b_2), c_2 \cdot \max_{a \in A} d\big(s_1(a),s_2(a)\big)}$, in the second case $\Wasserstein{M}{(d_B,d)}(t_1,t_2) ={c_1d_B(b_1,b_2) + {c_2}{|A|}^{-1} \sum_{a \in A} d\big(s_1(a),s_2(a)\big)}$
	and in the third case $\Wasserstein{M}{(d_B,d)}(t_1,t_2) ={c_1d_B(b_1,b_2) + {c_2} \sum_{a \in A} d\big(s_1(a),s_2(a)\big)}$. Of course one has to choose which of these pseudometrics fits into the respective context. While the first one selects either the distance of the output values or the maximal distance of the successors and neglects the other one, the latter two accumulate the distances. Depending on our maximal element, we have to make sure that we stay within the selected measuring interval $[0,\top]$ by proper scaling of the values.
\end{example}

Usually we will fix the first argument of the machine bifunctor (the set of outputs) of the machine bifunctor and just consider the machine endofunctor $M_B := M(B,\_)$  as in \Cref{exa:m2-comp,exa:wasserstein-machine} (\cpageref{exa:m2-comp,exa:wasserstein-machine}).  However, for the same reasons as explained before for the product bifunctor, we often need to lift it as a bifunctor and then fix the first component of the lifted bifunctor. One notable exception is the case where $B$ is endowed with the discrete metric. Then it is easy to show that the endofunctor lifting and the bifunctor lifting coincide, i.e., for all pseudometric spaces $(X,d)$ we have $\Wasserstein{M}{(d_B,d)} = \Wasserstein{M_B}{d}$ \cite[Ex.~5.6.11]{Ker16}.

Let us finish our short excursion to the theory of multifunctor compositionality with another example which shows how the machine bifunctor lifting helps to obtain suitable bisimilarity pseudometrics.

\begin{example}[Bisimilarity Pseudometric for Automata with Real Outputs]
	\label{exa:bisim-pseudo-pa}
	We consider the machine endofunctor with output set $[0,1]$, i.e., the functor $M_{[0,1]} = [0,1] \times \_^A$ which arises out of the machine bifunctor $M$ by fixing the first component to $[0,1]$. As maximal distance we set $\top=1$ and equip the machine bifunctor $M$ with the evaluation function $\ev_M \colon [0,1] \times [0,1]^A \to [0,1]$ where $\ev_M(o,s) = c_1o + {c_2}{|A|}^{-1} \sum_{a \in A} s(a)$ for $c_1,c_2 \in\,]0,1[$ such that $c_1 + c_2 \leq 1$ as in \Cref{exa:machinebifunctor}. Moreover, we recall that the carrier of the final $M_{[0,1]}$-coalgebra is $[0,1]^{A^*}$ \cite[Lem.~3]{JSS15}. 
	
	If we equip $[0,1]$ with the Euclidean metric $d_e$ and use our knowledge from \Cref{exa:machinebifunctor} we know that for any pseudometric $d$ on $[0,1]^{A^*}$ we obtain as Wasserstein pseudometric the function $\Wasserstein{F}{(d_e,d)}\colon \big([0,1] \times \big([0,1]^{A^*}\big)^A\big)^2 \to [0,1]$ with
	\begin{align*}
		\Wasserstein{F}{(d_e,d)}\big((r_1,s_1), (r_2, s_2)\big) = c_1|r_1-r_2| +  \frac{c_2}{|A|}\cdot\sum_{a \in A}d\big(s_1(a), s_2(a)\big)\,.
	\end{align*}
	We now want to obtain the final coalgebra for the endofunctor $\LiftedFunctor{M}_{([0,1],d_e)}=\LiftedFunctor{M}{\big(([0,1],d_e), \_\big)}$ on $\PMet$ which is a lifting of $M_{[0,1]}$. For this we use the fixed-point equation induced by \Cref{lem:lifting-coalg} (\cpageref{lem:lifting-coalg}). It is given by, for $p_1,p_2 \in [0,1]^{A^*}$, the equation
	\begin{align*}
		d(p_1,p_2) = c_1|p_1(\epsilon)- p_2(\epsilon)| + \frac{c_2}{|A|} \cdot \sum_{a \in A} d\big(\lambda w.p_1(aw), \lambda w.p_2(aw)\big)\,.
	\end{align*}
	As in the previous example a simple calculation shows that the function
	\begin{align*}
		&d_{[0,1]^{A^*}} \colon [0,1]^{A^*} \times [0,1]^{A^*} \to [0,1],\quad d_{[0,1]^{A^*}}(p_1,p_2) = c_1 \cdot \sum_{w \in A^*} \left(\frac{c_2}{|A|}\right)^{|w|} \left|p_1(w) - p_2(w)\right|
	\end{align*}
	is the least fixed-point of this equation so if we equip $[0,1]^{A^*}$ with this pseudometric we obtain the final $\LiftedFunctor{M}_{([0,1],d_e)}$-coalgebra. Thus for a probabilistic automaton $\alpha \colon X \to [0,1] \times X^A$ the bisimilarity pseudometric as given in  \Cref{def:bisim-pseudometric} (\cpageref{def:bisim-pseudometric}) is the function
	
	\begin{align*}
		\bd_\alpha \colon X \times X \to [0,1], \quad \bd_\alpha(x,y) = c_1 \cdot \sum_{w \in A^*} \left(\frac{c_2}{|A|}\right)^{|w|} \left|\final{x}_\alpha(w) - \final{y}_\alpha(w)\right|
	\end{align*} 
	where the unique map into the final coalgebra $\final{\cdot}_\alpha \colon X \to [0,1]^{A^*}$ maps each state to the function describing the output value of the automaton for each finite word when starting from the respective state.
\end{example}
\section{Lifting Natural Transformations and Monads}
\label{sec:monadlifting}
If we have a monad on $\Set$, we can of course use our framework to lift the endofunctor $T$ to a functor $\LiftedFunctor{T}$ on pseudometric spaces. A natural question that arises is, whether we also obtain a monad on pseudometric spaces, i.e., if the components of the unit and the multiplication are nonexpansive with respect to the lifted pseudometrics. In order to answer this question, we first take a closer look at sufficient conditions for lifting natural transformations.

\begin{theorem}[Lifting of a Natural Transformation]
	\label{prop:nt-lifting}
	Let $F$, $G$ be endofunctors on $\Set$ with evaluation functions
	$\ev_F$, $\ev_G$ and $\lambda\colon F \Rightarrow G$ be a natural
	transformation. The following two properties hold for the Kantorovich lifting.
	\begin{enumerate}
		\item If $d_e \circ (\ev_G \circ \lambda_{[0,\top]} \times \ev_G \circ
                  \lambda_{[0,\top]}) \le d_e \circ (\ev_F \times \ev_F)$
                  % $\ev_{G} \circ \lambda_{[0,\top]} \leq \ev_{F}$
                  then $\lambda_X$ is nonexpansive for all
                  pseudometric spaces $(X,d)$, i.e.,
                  $\Kantorovich{G}{d} \circ (\lambda_X \times
                  \lambda_X) \leq \Kantorovich{F}{d}$.\footnote{Note
                    that in an earlier version of the paper we stated
                    the requirement incorrectly as
                    $\ev_{G} \circ \lambda_{[0,\top]} \leq
                    \ev_{F}$.} \label{prop:nt-lifting:1}
		\item If $\ev_{G} \circ \lambda_{[0,\top]} = \ev_{F}$ then $\lambda_X$ is an isometry for all pseudometric spaces $(X,d)$, i.e., $\Kantorovich{G}{d} \circ (\lambda_X \times \lambda_X) = \Kantorovich{F}{d}$. \label{prop:nt-lifting:2}
	\end{enumerate}
	\noindent Moreover, similar properties hold for the Wasserstein lifting. 
	\begin{enumerate}
	\setcounter{enumi}{2}
		\item If $\ev_{G} \circ \lambda_{[0,\top]} \leq \ev_{F}$ then $\lambda_X$ is nonexpansive for all pseudometric space $(X,d)$, i.e., $\Wasserstein{G}{d} \circ (\lambda_X \times \lambda_X) \leq \Wasserstein{F}{d}$.\label{prop:nt-lifting:3}
		\item If $\ev_{G} \circ \lambda_{[0,\top]} = \ev_{F}$ and the Kantorovich-Rubinstein duality holds for $F$, i.e., $\Kantorovich{F}{d} = \Wasserstein{F}{d}$, then $\lambda_X$ is an isometry for all pseudometric spaces, i.e., $\Wasserstein{G}{d} \circ (\lambda_X \times \lambda_X) = \Wasserstein{F}{d}$.\label{prop:nt-lifting:4}
	\end{enumerate}
\end{theorem}
\begin{proof}
	Let $t_1,t_2 \in FX$.
	\begin{enumerate}
        \item By naturality of $\lambda$ we obtain for every
          $f \colon X \to [0,\top]$ the equality
          $\EvaluationFunctor{G}f\circ \lambda_X = \ev_G \circ Gf
          \circ \lambda_X = \ev_G \circ \lambda_{[0,\top]} \circ
          Ff$. Using this and the assumed inequality we can infer
                  \begin{eqnarray*}
                    && \Kantorovich{G}{d}\big(\lambda_X(t_1),
                    \lambda_X(t_2)\big) \\
                    & = &
                    \sup\{d_e\big(\EvaluationFunctor{G}f\big(\lambda_X(t_1)\big),
                    \EvaluationFunctor{G}f\big(\lambda_X(t_2)\big)\big)
                    \mid f \colon (X,d) \nonexpansiveTo
                    ([0,\top],d_e)\} \\
                    & = &
                    \sup\{d_e\big(\ev_G(\lambda_{[0,\top]}(Ff(t_1))),
                    \ev_G(\lambda_{[0,\top]}(Ff(t_2)))\big) \mid f
                    \colon (X,d) \nonexpansiveTo
                    ([0,\top],d_e)\} \\
                    & \leq &
                    \sup\{d_e\big(\ev_F(Ff(t_1)),\ev_F(Ff(t_2))\big)\mid
                    f \colon (X,d) \nonexpansiveTo ([0,\top],d_e)\} \\
                    & = &
                    \sup\{d_e\big(\EvaluationFunctor{F}f(t_1),\EvaluationFunctor{F}f(t_2)\big)\mid
                    f \colon (X,d) \nonexpansiveTo ([0,\top],d_e)\} \\
                    & = & \Kantorovich{F}{d}(t_1,t_2).
                  \end{eqnarray*}
		\item We obtain by naturality of $\lambda$
                  and $\ev_G \circ \lambda_{[0,\top]} = \ev_F$ for
                  every $f \colon X \to [0,\top]$ the equality
                  $\EvaluationFunctor{G}f \circ \lambda_X = \ev_G
                  \circ Gf \circ \lambda_X = \ev_G \circ
                  \lambda_{[0,\top]} \circ Ff = \ev_F \circ Ff =
                  \EvaluationFunctor{F}f$. Using this we compute
                  $\Kantorovich{G}{d}\big(\lambda_X(t_1),
                  \lambda_X(t_2)\big)
                  =\sup\{d_e\big(\EvaluationFunctor{G}f\big(\lambda_X(t_1)\big),\EvaluationFunctor{G}f\big(\lambda_X(t_2)\big)\big)
                  \mid f \colon (X,d) \nonexpansiveTo
                  ([0,\top],d_e)\}=\sup\{d_e\big(\EvaluationFunctor{F}f(t_1),\EvaluationFunctor{F}f(t_2)\big)\mid
                  f \colon (X,d) \nonexpansiveTo ([0,\top],d_e)\} =
                  \Kantorovich{F}{d}(t_1,t_2)$.
		\item Naturality of $\lambda$ yields the equations $\lambda_X \circ F\pi_i = G\pi_i \circ \lambda_{X \times X}$ and $\lambda_{[0,\top]} \circ Fd = Gd \circ \lambda_{X \times X}$  where $\pi_i \colon X \times X \to X$ are the projections of the product and $d\colon X \times X \to [0,\top]$ is a pseudometric on $X$. Using the first equality we can see that $\lambda_{X \times X}$ maps every coupling $t \in \Couplings{F}(t_1, t_2)$ to a coupling $\lambda_{X \times X}(t)\in \Couplings{G}\big(\lambda_X(t_1), \lambda_X(t_2)\big)$ because $G\pi_i\big(\lambda_{X \times X}(t)\big) = \lambda_X(F\pi_i(t)) = \lambda_X(t_i)$. Moreover, we can use our requirement ($\ev_G \circ \lambda_{[0,\top]} \leq \ev_F$) and the second equality to obtain $\EvaluationFunctor{G}d\big(\lambda_{X \times X}(t)\big) = \ev_{G} \circ Gd \circ \lambda_{X \times X}(t) = \ev_{G} \circ \lambda_{[0,\top]} \circ Fd (t)\leq \ev_{F} \circ Fd (t) = \EvaluationFunctor{F}d(t)$.
		With these preparations at hand we can conclude that $\Wasserstein{G}{d}\big(\lambda_X(t_1), \lambda_X(t_2)\big) = \inf\{\EvaluationFunctor{G}d(t')\mid t' \in \Couplings{G}\big(\lambda_X(t_1),\lambda_X(t_2)\big)\}\leq \inf\{\EvaluationFunctor{G}d\big(\lambda_{X \times X}(t)\big)\mid t \in \Couplings{F}(t_1,t_2)\}\leq \inf\{\EvaluationFunctor{F}d(t)\mid t \in \Couplings{F}(t_1,t_2)\} = \Wasserstein{F}{d}(t_1,t_2)$.		
		\item Using the previous two results and the fact that Wasserstein is an upper bound yields $\Kantorovich{F}{d} = \Kantorovich{G}{d} \circ (\lambda_X \times \lambda_X) \leq \Wasserstein{G}{d} \circ (\lambda_X\times\lambda_X) \leq \Wasserstein{F}{d}$ and since $\Kantorovich{F}{d} = \Wasserstein{F}{d}$ all these inequalities are equalities.\qedhere
	\end{enumerate}
\end{proof}

\noindent In the remainder of this paper we will call a natural transformation $\lambda$ nonexpansive [an isometry] if (and only if) each of its components are nonexpansive [isometries] and write $\overline{\lambda}$ for the resulting natural transformation from $\LiftedFunctor{F}$ to $\LiftedFunctor{G}$. Instead of checking nonexpansiveness separately for each component of a natural transformation, we can just check the above (in-)equalities involving the two evaluation functions. 

By applying these conditions on the unit and multiplication of a given monad, we can now provide sufficient criteria for a monad lifting.

\begin{corollary}[Lifting of a Monad]
	\label{cor:monad-lifting}
	Let $(T, \eta, \mu)$ be a $\Set$-monad and $\ev_T$ an evaluation function for $T$. Then the following holds.
	\begin{enumerate}
	\item If $\ev_T \circ \eta_{[0,\top]} \leq \id_{[0,\top]}$
          then $\eta$ is nonexpansive for the Wasserstein liftings. If
          $d_e \circ (\ev_T \circ \eta_{[0,\top]} \times \ev_T \circ
          \eta_{[0,\top]}) \le d_e$ the same holds for the
          Kantorovich lifting. In both cases we obtain the unit
          $\overline{\eta}\colon \LiftedFunctor{\Id}\Rightarrow
          \LiftedFunctor{T}$ in $\PMet$.
	\item If $\ev_T \circ \eta_{[0,\top]} = \id_{[0,\top]}$ then $\eta$ is an isometry for both liftings. 
	\item Let $\LiftedMetric{T}{d} = \Wasserstein{T}{d}$
        (Wasserstein case). If
        $\ev_T \circ \mu_{[0,\top]} \leq \ev_T * \ev_T$ and
        compositionality holds for $TT$, i.e.,
        $\LiftedMetric{T}{(\LiftedMetric{T}{d})} =
        \LiftedMetric{TT}{d}$, then $\mu$ is nonexpansive, i.e.,
        $\LiftedMetric{T}{d} \circ (\mu_X \times \mu_X) \leq
        \LiftedMetric{T}{(\LiftedMetric{T}{d})}$.

        The same holds whenever $d^T = \Kantorovich{T}{d}$
        (Kantorovich case) and
        $d_e \circ (\ev_T \circ \mu_{[0,\top]} \times
        \ev_T \circ \mu_{[0,\top]}) \le d_e\circ
        ((\ev_T * \ev_T)\times (\ev_T * \ev_T))$.

        This yields the multiplication
        $\overline{\mu}\colon \LiftedFunctor{T}\,\LiftedFunctor{T}
        \Rightarrow \LiftedFunctor{T}$ in $\PMet$.
	\end{enumerate}
\end{corollary}
\begin{proof}
	This is an immediate consequence of
        \Cref{prop:nt-lifting}. For the unit take $F=\Id$ with
        evaluation function $\ev_F = \id_{[0,\top]}$, hence
        $\Kantorovich{F}{d} = \Wasserstein{F}{d}=d$ and $G=T$,
        $\ev_G=\ev_T$, $\lambda = \eta \colon \Id \Rightarrow T$. For
        the multiplication take $F = TT$, $G=T$, $\ev_F = \ev_{TT} =
        \ev_T \circ T\ev_T$, $\ev_G = \ev_T$ and $\lambda = \mu$.

        Note also that $d^{TT} \le (d^T)^T$ always holds in the
        Kantorovich case (due to
        \ref{lem:compositionality2}(\ref{item:leqKantorovich})) giving
        us an identity natural transformation
        $\LiftedFunctor{T}\,\LiftedFunctor{T}\Rightarrow
        \LiftedFunctor{TT}$, making the requirement of compositionality
        unnecessary in that case.
\end{proof}

We conclude this section with two examples of liftable monads.

\begin{example}[Lifting of the Finite Powerset Monad]\label{exa:monad:powfin}
	We recall that the finite powerset functor $\PowersetFinite$ is part of a monad with unit $\eta$ consisting of the functions $\eta_X\colon X \to \PowersetFinite X$, $\eta_X(x) = \set{x}$ and multiplication given by $\mu_X \colon \PowersetFinite \PowersetFinite X \to \PowersetFinite X$, $\mu_X(S) = \cup S$. We check if our conditions for the Wasserstein lifting are satisfied. Given $r \in [0,\infty]$ we have $\ev_T \circ \eta_{[0,\infty]}(r) = \max\set{r} = r$ and for $\mathcal{S} \in \PowersetFinite(\PowersetFinite[0,\top])$ we have $\ev_T \circ \mu_{[0,1]}(\mathcal{S}) = \max \cup \mathcal{S} = \max \cup_{S \in \mathcal{S}} S$ and $\ev_T \circ T\ev_T (\mathcal{S}) = \max\left(\ev_T[\mathcal{S}]\right) = \max\set{\max S \mid S \in \mathcal{S}}$ and thus it is easy to see that both values coincide. Moreover, we recall from \Cref{exa:powerset-compositional} that we have compositionality for $\PowersetFinite\PowersetFinite$. Therefore, by \Cref{cor:monad-lifting} $\eta$ is an isometry and $\mu$ nonexpansive.
\end{example}

\begin{example}[Lifting of the Distribution Monad With Finite Support]
	\label{exa:monad:distr}
	It is known that the probability distribution functor $\DistributionsFinite$ is part of a monad: the unit $\eta$ consists of the functions $\eta_X\colon X \to \DistributionsFinite X$, $\eta_X(x) = \delta_{x}^X$ where $\delta_x^X$ is the Dirac distribution and the multiplication is given by $\mu_X \colon \DistributionsFinite\DistributionsFinite X \to \DistributionsFinite X$, $\mu_X(P) = \lambda x.\sum_{q \in \DistributionsFinite X}P(q) \cdot q(x)$. We consider its Wasserstein lifting. Since $[0,1] = \DistributionsFinite \two$ we can see that $\ev_{\DistributionsFinite} = \mu_\two$. Using this and the monad laws we have $\ev_{\DistributionsFinite} \circ \eta_{[0,1]} = \mu_\two \circ \eta_{\DistributionsFinite \two} = \id_{\DistributionsFinite X} = \id_{[0,1]}$ and also $\ev_{\DistributionsFinite} \circ \mu_{[0,1]} = \mu_\two \circ \mu_{\DistributionsFinite \two} = \mu_\two \circ \DistributionsFinite \mu_\two = \ev_{\DistributionsFinite} \circ \DistributionsFinite\ev_{\DistributionsFinite}$. Moreover, since we always have optimal couplings, we have compositionality for $\DistributionsFinite\DistributionsFinite$ by \Cref{prop:comp}. Thus by \Cref{cor:monad-lifting} $\eta$ is an isometry and $\mu$ nonexpansive.
\end{example}

\section{Trace Pseudometrics}
\label{sec:tracemetrics}
\begin{wrapfigure}[13]{r}{5.25cm}
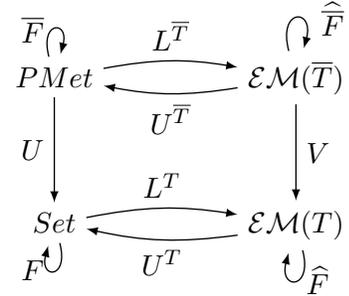

	\centering
	\begin{diagram}
		\matrix[matrix of math nodes, column sep=1.8cm, row sep=1.3cm] (m){
			\PMet & \EM(\LiftedFunctor{T})\\
			\Set & \EM(T)\\
		};
		
		\path (m-1-1) edge node[left]{$U$} (m-2-1);
		\path (m-1-2) edge node[right]{$V$} (m-2-2);
		\path (m-1-1) edge[bend left=10] node[above]{$L^{\LiftedFunctor{T}}$} (m-1-2);
		\path (m-1-2) edge[bend left=10] node[below]{$U^{\LiftedFunctor{T}}$} (m-1-1);
		
		\path (m-2-1) edge[bend left=10] node[above]{$L^T$} (m-2-2);
		\path (m-2-2) edge[bend left=10] node[below]{$U^T$} (m-2-1);
		\path (m-1-1) edge[loop above] node[left]{$\LiftedFunctor{F}$}
		(m-1-1);
		\path (m-2-1) edge[loop below] node[left]{$F$} 
		(m-2-1);
		\path (m-1-2) edge[loop above] node[right,xshift=.2cm]{$\widehat{\LiftedFunctor{F}}$} 
		(m-1-2);
		\path (m-2-2) edge[loop below] node[right]{$\widehat{F}$} 
		(m-2-2);
	\end{diagram}
	\caption{Trace pseudometrics via the generalized powerset construction}
	\label{fig:gen-pow-trace}
\end{wrapfigure}
Combining all our previous results we now want to use the generalized
powerset construction \cite{SBBR13} on $\PMet$ instead of $\Set$ to
obtain trace pseudometrics. The basic setup is summarized in the
(non-commutative) diagram in \Cref{fig:gen-pow-trace}. We quickly
recall that in the usual, qualitative setting (bottom part of
\Cref{fig:gen-pow-trace}) we have to start with a coalgebra $c\colon X
\to FTX$ where $F$ is an endofunctor with final coalgebra $z\colon Z
\to FZ$ and $(T, \eta, \mu)$ is a monad on $\Set$. Using a
distributive law $\lambda \colon TF \Rightarrow FT$ we can then consider the determinization $c^\sharp$ of  $c$ which is defined as\\
\begin{diagram}
		\matrix[matrix of math nodes, column sep=0.8cm, row sep=.5cm] (m){
			c^\sharp := \Big(TX & TFTX & FTTX & FTX\Big)\,.\\
		};
		\draw (m-1-1) edge node[above]{$Tc$}(m-1-2);
		\draw (m-1-2) edge node[above]{$\lambda_{TX}$}(m-1-3);
		\draw (m-1-3) edge node[above]{$F\mu_X$}(m-1-4);
\end{diagram}\\
With this determinized coalgebra we define two states $x,y \in X$ of the original coalgebra $c$ to be trace equivalent if and only if $\final{\eta_X(x)}_c = \final{\eta_X(y)}_c$ holds. 

The underlying reason why this technique works is that the
distributive law $\lambda$ yields a unique lifting (with respect to the canonical Eilenberg-Moore adjunction $L^T \dashv U^T$ of the functor $F$ to a functor $\widehat{F}$ on the Eilenberg-Moore category $\EM(T)$. The determinization of a coalgebra is nothing but the application of another lifting 
\begin{align*}
	L\colon \CoAlg{FT} \to \CoAlg{\widehat{F}}
\end{align*} 
of the free algebra functor $L^T$, i.e., $c^\sharp = U^TL(c)\colon TX \to FTX$ is the $F$-coalgebra underlying the $\widehat{F}$-coalgebra $L(c)\colon \mu_X \to \widehat{F}\mu_X$. 

In order to move to a quantitative setting (upper part of \Cref{fig:gen-pow-trace}) we need to require that both the functor $F$ and the monad $(T, \eta, \mu)$ can be lifted. Then clearly any $FT$-coalgebra $c\colon X \to FTX$ can be regarded as an $\LiftedFunctor{F}\,\LiftedFunctor{T}$-coalgebra $\overline{c}$ by
equipping the state space $X$ with the discrete metric $d$ assigning
$\top$ to non equal states (in this way, $\overline{c}$ is trivially
nonexpansive). Moreover, if we can ensure that the $\EM$ law $\lambda$
is nonexpansive, thus yielding a distributive law $\overline{\lambda} \colon \LiftedFunctor{T}\,\LiftedFunctor{F} \Rightarrow \LiftedFunctor{F}\,\LiftedFunctor{T}$, we can use exactly the same ideas as before. In particular, we can lift the lifted functor $\LiftedFunctor{F}$ to a functor 
$\widehat{\LiftedFunctor{F}} \colon \EM\big(\LiftedFunctor{T}) \to  \EM\big(\LiftedFunctor{T})$ 
on the Eilenberg-Moore category of the lifted monad. With this, we can lift the free algebra functor $L^{\LiftedFunctor{T}}$ to a functor 
\begin{align*}
	L' \colon \CoAlg{\LiftedFunctor{F}\,\LiftedFunctor{T}} \to \CoAlg{\,\widehat{\LiftedFunctor{F}}\,} 
\end{align*}
which allows us to determinize $\overline{c}$ (as $\LiftedFunctor{F}\,\LiftedFunctor{T}$-coalgebra in $\PMet$) to the $\widehat{\LiftedFunctor{F}}$-coalgebra  $L'\overline{c} \colon \LiftedFunctor{\mu}_X \to \widehat{\LiftedFunctor{F}}\LiftedFunctor{\mu}_X$ which is given by the underlying  $\LiftedFunctor{F}$-coalgebra 
\begin{align}
	\overline{c}^\sharp := U^TL'\overline{c} = \LiftedFunctor{F}\,\LiftedFunctor{\mu_X} \circ \LiftedFunctor{\lambda_{TX}} \circ \LiftedFunctor{T}\,{\overline{c}}\colon \LiftedFunctor{T}(X,d) \nonexpansiveTo \LiftedFunctor{F}\,\LiftedFunctor{T}(X,d)\,.
\end{align}
If we now equip $TX$ with the behavioral pseudometric $\bd_{{c}^\sharp} \colon (TX)^2 \to [0,\top]$ as in \Cref{def:bisim-pseudometric}, we can define the trace pseudometric on $X$ via the unit $\eta$ as follows.

\begin{definition}[Trace Pseudometric]
	\label{def:trace-pseudometric}
	Let $F$ be an endofunctor and $(T,\eta,\mu)$ be a monad on $\Set$. If $F$ has a final coalgebra $z \colon Z \to FZ$ in $\Set$, $F$ has a lifting $\LiftedFunctor{F} \colon \PMet \to \PMet$, $(T,\eta, \mu)$ has a lifting $(\LiftedFunctor{T},\LiftedFunctor{\eta},\LiftedFunctor{\mu})$, and there is an $\EM$-law $\lambda \colon TF \to FT$ which can be lifted to an $\EM$-law $\LiftedFunctor{\lambda}\colon \LiftedFunctor{T}\,\LiftedFunctor{F} \to \LiftedFunctor{F}\,\LiftedFunctor{T}$ then for any coalgebra $c \colon X \to FTX$ we define the \emph{trace pseudometric} to be 
	\begin{align*}
		\td_c:= \bd_{{c}^\sharp} \circ (\eta_X \times \eta_X) \colon X \times X \to [0,\top]
	\end{align*}
	where ${c}^\sharp = F\mu_X \circ \lambda_{TX} \circ Tc\colon TX \to FTX$ is the determinization of $c\colon X \to FTX$ and $bd_{{c}^\sharp} \colon (TX)^2 \to [0,\top]$ is the corresponding bisimilarity pseudometric (\Cref{def:bisim-pseudometric}, \cpageref{def:bisim-pseudometric}).
\end{definition}

In order to apply this definition to our two main examples (nondeterministic and probabilistic automata) the only missing thing is the lifting of the $\EM$-law to $\PMet$. We note that \Cref{prop:nt-lifting} (\cpageref{prop:nt-lifting}) not only provides sufficient conditions for monad liftings but also can be exploited to lift distributive laws. The additional commutativity requirements for $\Kl$-laws or $\EM$-laws trivially hold when all components are nonexpansive. For the Wasserstein lifting it suffices to require compositionality on the left hand side of the law and to check one inequality.

\begin{corollary}[Lifting of a Distributive Law]
	\label{cor:lifting-distr-law}
	Let $F,G$ be weak pullback preserving endofunctors on $\Set$ with well-behaved evaluation functions $\ev_F$, $\ev_G$ and $\lambda\colon FG \Rightarrow GF$ be a distributive law. If the evaluation functions satisfy $(\ev_G \ast \ev_F) \circ \lambda_{[0,\top]} \leq \ev_F \ast \ev_G$ and compositionality holds for $FG$, then $\lambda$ is nonexpansive for the Wasserstein lifting and hence $\overline{\lambda}\colon \LiftedFunctor{F}\,\LiftedFunctor{G} \Rightarrow \LiftedFunctor{G}\,\LiftedFunctor{F}$ is also a distributive law.
\end{corollary}
\begin{proof}
	We use the evaluation function $\ev_F \ast \ev_G$ for $FG$ and for $GF$ the evaluation function $\ev_{G} \ast \ev_F$. By \Cref{prop:nt-lifting} (\cpageref{prop:nt-lifting}) we know that $\Wasserstein{GF}{d} \circ (\lambda_X \times \lambda_X) \leq \Wasserstein{FG}{d}$ and by \Cref{lem:compositionality2}.\ref{item:geqWasserstein} (\cpageref{lem:compositionality2}) we have $\Wasserstein{G}{(\Wasserstein{F}{d})} \leq \Wasserstein{GF}{d}$. Plugging everything together we conclude that for every pseudometric space $(X,d)$ we have
	\begin{align*}
		\Wasserstein{G}{(\Wasserstein{F}{d})}  \circ (\lambda_X \times \lambda_X) \leq \Wasserstein{GF}{d} \circ (\lambda_X \times \lambda_X) \leq \Wasserstein{FG}{d} = \Wasserstein{F}{(\Wasserstein{G}{d})}
	\end{align*}
	which is the desired nonexpansiveness of $\lambda_X$.
\end{proof}

In the remainder of this section we will consider two examples where in both cases $G$ is the machine endofunctor $M_B = B \times \_^A$ (for  $B=\two$ and $B=[0,1]$). It is well-known \cite[Lem.~3]{JSS15} that for every output set $B$ the final coalgebra for $M_B$ is  
\begin{align*}
	z = \langle o_z, s_z \rangle \colon B^{A^*} \to  B \times \left(B^{A^*}\right)^A
\end{align*}
which maps any function $f \colon A^* \to B$ to the tuple $z(f)=\big(o_z(f),s_z(f)\big)$. The output value $o_z(f)$ is the value of $f$ on the empty word, i.e., $o_z(f) = f(\epsilon)$ and the successor function $s_z(f) \colon A \to B^{A^*}$ assigns to each letter $a \in A$ the function $s_z(f)(a)\colon A^* \to B$. Its value on a word $w \in A^*$ is equal to the value of $f$ on the word $aw$, formally $s_z(f)(a)(w) = f(aw)$. In order to lift the machine functor we have two possibilities:
\begin{enumerate}
	\item We can lift it  as an endofunctor obtaining an endofunctor $\LiftedFunctor{M_B}$ on $\PMet$. 
	\item We lift the machine bifunctor $M$ of \Cref{exa:machinebifunctor} to obtain a lifted bifunctor $\LiftedFunctor{M}\colon \PMet^2 \to \PMet$. Then we fix a pseudometric $d_B$ on the outputs $B$ and consider the induced endofunctor $\LiftedFunctor{M}_{(B,d_B)} := \LiftedFunctor{M}\big((B, d_B), \_\big)$.
\end{enumerate}
In the first case we can of course simply apply \Cref{cor:lifting-distr-law} from above but in the second case we have to prove nonexpansiveness of $\lambda$ separately. We will employ the first approach for nondeterministic automata (where $B= \two$) and the second one for probabilistic automata (where $B=[0,1])$. 

\subsection{Trace Pseudometric for Nondeterministic Automata}
We will now consider a known $\EM$-law $\lambda \colon  \PowersetFinite M_\two \Rightarrow M_\two \PowersetFinite$ for finitely branching nondeterministic automata \cite[p.~867]{JSS15}. Here, \Cref{cor:lifting-distr-law} is directly applicable using $F=\PowersetFinite$ and $G=M_\two = \two\times\_^A$.

\begin{lemma}[Distributive Law for Nondeterministic Automata]
	\label{exa:EMlaw:powfin:nonexpansive}
	Let $(\PowersetFinite,\eta,\mu)$ be the finite powerset monad from \Cref{exa:monad:powfin} with the maximum as evaluation function and $M_\two = \two \times \_^A$ be the deterministic automaton functor equipped with the evaluation function $\ev_{M_\two}\colon \two \times [0,\top]^A \to [0,\top]$, $\ev_{M_\two}(o,s) = c \cdot \max_{a \in A}s(a)$ with $c \in \,]0,1]$. We consider the $\EM$-law $\lambda \colon \PowersetFinite(\two \times \_^A) \Rightarrow \two \times \PowersetFinite(\_)^A$ on $\Set$ which is defined, for any set $X$, as $\lambda_X(S) = \langle o, s \rangle$  with $o(S) = 1$ if there is an $s' \in X^A$ such that $(1,s') \in S$ else $o(S) = 0$ and the successor functions $s(S) \colon A \to \Powerset X$, where $s(A) = \set{s'(a) \mid (o',s') \in S}$ for every $S \in \powerset{\two \times X^A}$. This law is nonexpansive. 
\end{lemma}
\begin{proof}
	In the notation of \Cref{cor:lifting-distr-law} we have $F = \two \times \_^A= M_\two$, and $G = \PowersetFinite$. The composed evaluation functions are $\ev_{\PowersetFinite} \ast \ev_{M_\two}\colon \PowersetFinite(\two \times [0,1]^A) \to[0,1]$ where for $S \in \PowersetFinite(\two \times [0,1]^A)$ 
	\begin{align*}
		\ev_{\PowersetFinite}\ast \ev_{M_\two} (S) &= \ev_{\PowersetFinite} \circ \PowersetFinite\ev_{M_\two} (S)= \max\set{\ev_{M_\two}(o,s) \mid (o,s) \in S} \\
		&=  \max\set{c \cdot \max_{a \in A} s(a) \mid (o,s) \in S} = c \cdot \max_{(o,s) \in S}\max_{a \in A}s(a)
	\end{align*}
	and $\ev_{M_\two} \ast \ev_{\PowersetFinite} \colon \two \times (\PowersetFinite [0,1])^A \to [0,1]$ where for $(o,s) \in \two \times (\PowersetFinite X)^A$
	\begin{align*}
		\ev_{M_\two} \ast  \ev_{\PowersetFinite} (S) &= \ev_{M_\two} \circ M_\two\ev_{\PowersetFinite}(o,s) = \ev_{M_\two}\big(o, \max \circ s\big) = c \cdot \max_{a \in A} \max s(a)
	\end{align*}
	
	\noindent As we have seen in \Cref{exa:m2-comp} (\cpageref{exa:m2-comp}) we have compositionality for $\PowersetFinite M_{\two}$ and the Wasserstein lifting. We want to apply \Cref{cor:lifting-distr-law} to show nonexpansiveness. For this we just have to check that $(\ev_{\PowersetFinite} \ast \ev_{M_\two}) \circ \lambda_{[0,1]} \leq \ev_{M_\two} \ast \ev_{\PowersetFinite}$ holds. For $S \in \PowersetFinite(\two \times [0,1]^A)$ we have $(\ev_{\PowersetFinite} \ast \ev_{M_\two})  \circ \lambda_{[0,1]}(S) = c \cdot \max_{a \in A}\max \set{s(a) \mid (o,s) \in S} = c \cdot \max_{a \in A}\max_{(o,s) \in S} s(a) = \ev_{M_\two} \ast \ev_{\PowersetFinite}(S)$ which concludes the proof.
\end{proof}

With this result at hand we can now define the trace pseudometric for finitely branching nondeterministic automata.

\begin{example}[Trace Pseudometric for Nondeterministic Automata]
	\label{exa:final-coalgebra-tracedistance-nfa}
	We consider the machine endofunctor $M_\two = \two \times \_^A$. As maximal distance we take $\top=1$ and as evaluation function we use $\ev_{M_\two} \colon [0,1] \times [0,1]^A$ with $\ev_{M_\two}(o,s) = c \cdot \max_{a \in A} s(a)$ for $0<c <1$ as in \Cref{exa:final-coalg-lifted-machine} (\cpageref{exa:final-coalg-lifted-machine}) and lift the functor using the Wasserstein lifting. 
	
	We now take a finitely branching nondeterministic automaton which is a coalgebra $\alpha \colon X \to \two \times (\PowersetFinite X)^A$. Its determinization is the powerset automaton $\alpha^\sharp\colon \PowersetFinite X \to \two \times (\PowersetFinite X)^A$ whose states are sets of states of the original automaton. We recall from \Cref{exa:bisim-pseudo-det-aut} (\cpageref{exa:bisim-pseudo-det-aut}) that the bisimilarity pseudometric is the function
	\begin{align*}
		\bd_{\alpha^\sharp} \colon \PowersetFinite X \times \PowersetFinite X \to [0,1], \quad \bd_{\alpha^\sharp}(S,T) = c^{\inf\set{n \in \N \mid \exists w \in A^n.\final{S}_{\alpha^\sharp}(w) \not = \final{T}_{\alpha^\sharp}(w)}}\,.
	\end{align*} 

	If we apply the construction of \Cref{def:trace-pseudometric} using the unit $\eta_X(x) = \set{x}$ of the powerset monad we obtain the trace pseudometric
	\begin{align*}
		\td_\alpha \colon X \times X \to [0,1], \quad \td_\alpha(x,y) = c^{\inf\set{n \in \N \mid \exists w \in A^n.\final{\set{x}}_{\alpha^\sharp}(w) \not = \final{\set{y}}_{\alpha^\sharp}(w)}}\,.
	\end{align*} 
	Thus the trace distance of states $x$ and $y$ of $\alpha$ is given by a word $w$ of minimal length which is contained in the language of the state $\set{x}$ of the determinization $\alpha^\sharp$ and \emph{not} contained in the language of the state $\set{y}$. Then, the distance is computed as $c^{|w|}$. 
\end{example}
\subsection{Trace Pseudometric for Probabilistic Automata}
Our next example of an $\EM$-law will be of the shape $\lambda \colon \DistributionsFinite M_{[0,1]} \to M_{[0,1]}\DistributionsFinite$. This is much more complicated because we need to consider multifunctors to obtain the correct lifting. 

\begin{lemma}[Distributive Law for Probabilistic Automata]
	\label{exa:EMlaw:distr}
	Let $(\DistributionsFinite,\eta,\mu)$ be the distribution monad with finite support from \Cref{exa:monad:distr} (\cpageref{exa:monad:distr}) and $M$ be the machine bifunctor from \Cref{exa:machinebifunctor} (\cpageref{exa:machinebifunctor}). There is a known\footnote{This law arises out of the so-called \emph{strength map} of the monad \cite[Lem.~4]{JSS15}. It is also straightforward to show this just using the properties of an $\EM$-law \cite{Ker16}.} $\EM$-law $\lambda \colon \DistributionsFinite([0,1] \times \_^A) \Rightarrow [0,1] \times \DistributionsFinite^A$ in $\Set$ where $\lambda_X = \langle o_X,s_X \rangle$ with
	\begin{align*}
		&o_X(P) = \sum_{r \in [0,1]} r \cdot P(r,X^A), \quad \text{and} \quad s_X(P) \colon A \to \DistributionsFinite X, s_X(P)(a)(x) = \sum_{{s' \in X^A,\, s'(a)=x}}\!\!\!\!\!\!\!\!\!\!P([0,1],s')
	\end{align*}
	for all sets $X$ and all distributions $P \colon [0,1] \times X^A \to [0,1]$. The endofunctors on both sides of this law can be seen a bifunctors $F,G$ for which one component is fixed. They arise by composition of the distribution functor, the identity functor and the machine bifunctor as follows. The bifunctor on the left hand side (the domain of $\lambda$) is $F= \DistributionsFinite \circ M$ and the bifunctor on the right hand side (the codomain of $\lambda$) is $G = M \circ (\Id \times \DistributionsFinite)$ and to get our law we need to fix the respective first parameter to be $[0,1]$.
	
	If we use the usual distribution functor evaluation function as given in \Cref{exa:probability-distribution-functor,exa:probability-distribution-functor2} (\cpageref{exa:probability-distribution-functor,exa:probability-distribution-functor2}), the identity function as evaluation function for $\Id$ and the discounted sum for the machine bifunctor as in \Cref{exa:machinebifunctor} (\cpageref{exa:machinebifunctor}), lift both bifunctors and then fix their first component to the metric space $([0,1],d_e)$ then the above $\EM$-law is nonexpansive. 
\end{lemma}

\begin{proof}
	Since all of the involved (bi)functors have optimal couplings, we have compositionality and the evaluation functions for the composed functors are
	\begin{align*}
		&\ev_F := \ev_{\DistributionsFinite} \circ \DistributionsFinite\ev_M\colon\DistributionsFinite([0,1]\times [0,1]^A) \to [0,1] \quad \text{and}\\
		&\ev_G := \ev_M \circ M(\id_{[0,1]}, \ev_{\DistributionsFinite}) \colon [0,1] \times (\DistributionsFinite X)^A \to [0,1]\,.
	\end{align*}
	For any set $X$ we define the function 
	\begin{align*}
		 \Lambda_X := \langle \mathbf{o}_1,\mathbf{o}_2, \mathbf{s}\rangle\colon \DistributionsFinite\big([0,1]^2\times (X\times X)^A\big)\to [0,1]^2\times \big(\DistributionsFinite (X\times X)\big)^A
	\end{align*}
	where for any $P \in \DistributionsFinite([0,1] \times [0,1] \times (X\times X)^A)$
	 \begin{align*}
	 &\mathbf{o}_1(P) = \sum_{r \in [0,1]} r \cdot P\big(r,[0,1],(X\times X)^A\big), \quad \mathbf{o}_2(P) = \sum_{r \in [0,1]} r \cdot P\big([0,1],r,(X\times X)^A\big), \text{ and}\\
	 &\mathbf{s}(P) \colon A \to \DistributionsFinite (X \times X), \quad \mathbf{s}(P)(a)(x,y) = \!\!\!\!\sum_{{s' \in (X \times X)^A,\ s'(a)=(x,y)}}\!\!\!\!\!\!\!\!\!\!\!P([0,1]^2,s')
	 \end{align*}
	 completely analogous to the definition of the components $\lambda_X$ of the distributive law in \Cref{exa:EMlaw:distr}. A lengthy but straightforward computation\footnote{It is presented in the appendix of the extended version of our CALCO '15 paper \cite{BBKK15}.} shows that for any $P_1,P_2 \in \DistributionsFinite\big([0,1] \times X^A\big)$ and any $P \in \Couplings{F}(P_1,P_2) \subseteq \DistributionsFinite\big([0,1] \times [0,1] \times (X\times X)^A\big)$ this function satisfies the requirements $\Lambda_X(P) \in \Couplings{G}\big(\lambda_X(P_1), \lambda_X(P_2)\big)$ and $\EvaluationFunctor{G}(d_B,d)\big(\Lambda_X(P)\big)\leq 	\EvaluationFunctor{F}(d_B,d)(P)$. Using these properties we can conclude that
	\begin{align*}
		\Wasserstein{G}{(d_B,d)}\big(\lambda_X(P_1), \lambda_X(P_2)\big) &= \inf\set{\EvaluationFunctor{G}(d_B,d)(P') \ \big|\ P' \in \Couplings{G}\big(\lambda_X(P_1),\lambda_X(P_2)\big)}\\
		&\leq \inf\set{\EvaluationFunctor{G}(d_B,d)\big(\Lambda_{X}(P)\big) \ \big|\ P \in \Couplings{F}(P_1,P_2)}\\
		&\leq \inf\set{\EvaluationFunctor{F}(d_B,d)(P) \ \big|\ t \in \Couplings{F}(P_1,P_2)} = \Wasserstein{F}{(d_B,d)}(P_1,P_2) 
	\end{align*}
	\noindent which, due to compositionality, proves the desired nonexpansiveness of $\lambda_X$. 
\end{proof}

Using this law, we can now define the trace pseudometric for probabilistic automata.

\begin{example}[Trace Pseudometric for Probabilistic Automata]
	\label{exa:final-coalgebra-tracedistance-pa}
	As in \Cref{exa:bisim-pseudo-pa} (\cpageref{exa:bisim-pseudo-pa}) we consider the machine functor $M_{[0,1]} = [0,1] \times \_^A$ which  arises out of the machine bifunctor $M$ by fixing the first component to $[0,1]$. As maximal distance we set $\top=1$ and equip the machine bifunctor $M$ with the evaluation function $\ev_M \colon [0,1] \times [0,1]^A \to [0, 1]$ where $\ev_M(o,s) = c_1o + {c_2}{|A|}^{-1} \sum_{a \in A} s(a)$ for $c_1,c_2 \in\,]0,1[$ such that $c_1 + c_2 \leq 1$ as in \Cref{exa:machinebifunctor}. We lift this bifunctor using the Wasserstein lifting and then fix its first component to $([0,1], d_e)$. 
	
	For a probabilistic automaton $\alpha \colon X \to [0,1] \to (\DistributionsFinite X)^A$ its determinization is the $M_{[0,1]}$-coalgebra $\alpha^\sharp\colon \DistributionsFinite X \to  [0,1] \to (\DistributionsFinite X)^A$ whose state space are distributions on the states of the original automaton. From \Cref{exa:bisim-pseudo-pa} we know that we obtain the following bisimilarity pseudometric
	\begin{align*}
		\bd_{\alpha^\sharp} \colon \DistributionsFinite X \times \DistributionsFinite X \to [0,1], \ \  \bd_{\alpha^\sharp}(p,q) = c_1 \cdot \sum_{w \in A^*} \left(\frac{c_2}{|A|}\right)^{|w|} \left|\final{p}_{\alpha^\sharp}(w) - \final{q}_{\alpha^\sharp}(w)\right|\,.
	\end{align*}
	If we apply the construction of \Cref{def:trace-pseudometric} using the unit $\eta_X(x) = \delta_x^X$ of the finite distribution monad we obtain the trace pseudometric
	\begin{align*}
		\td_\alpha \colon X \times X \to [0,1], \quad \td_\alpha(x,y) =  c_1 \cdot \sum_{w \in A^*} \left(\frac{c_2}{|A|}\right)^{|w|} \left|\final{\delta_x^X}_{\alpha^\sharp}(w) - \final{\delta_y^X}_{\alpha^\sharp}(w)\right|\,.
	\end{align*} 
	Thus the trace distance of states $x$ and $y$ of $\alpha$ is
        given by the distance of their Dirac distributions in the
        determinization. \new{We obtain a metric which is a weighted
          version of the total variation distance.}
\end{example}
\section{Conclusion, Related and Future Work}
\label{sec:conclusion}
In this paper we have demonstrated how a substantial part of the coalgebraic machinery for modeling and analyzing labeled transition systems can be extended from a qualitative to a quantitative setting. The crucial idea for this is the idea to lift a functor from the category $\Set$ of sets and functions to the category $\PMet$ of pseudometric spaces and nonexpansive functions. While all the remaining results require a bit of effort, they arise naturally once such lifting has been defined. The big advantage of our approach is that we try
\begin{itemize}
	\item to keep it as general as possible (by using coalgebra and not restricting to a specific class of transition systems) and
	\item to minimize the amount of additional information needed.
\end{itemize}
Instead of assuming that a transition system already comes equipped with some distance function on the state space, we give canonical definitions of bisimilarity and trace pseudometrics in the sense that they arise automatically out of the coalgebraic model. The only information we have to provide is the evaluation function which explains how we can evaluate the effect of applying the branching functor to real numbers as a single real number.

Whenever someone is interested in defining a new type of transition
system, he or she can now automatically derive canonical notions of
both behavioral equivalences and pseudometrics.

Since we suggested two possible liftings (Kantorovich and
Wasserstein), we should discuss their respective advantages and
disadvantages. 

The Wasserstein lifting is somewhat easier to understand due to the
transport plan analogy (Section~\ref{sec:motivation}) and behaves
better with respect to compositionality
(Section~\ref{sec:compositionality}). Furthermore, under some mild
conditions, the lifting preserves metrics. Also, there is an
interesting fibrational theory underlying the Wasserstein lifting
that we are currently working out.

The Kantorovich lifting has other advantages: first it can be
characterized via a universal property. Namely, it is the smallest
lifting that makes the evaluation function
$\ev_F\colon (F[0,\top],d_e^F)\to ([0,\top],d_e)$
nonexpansive. Furthermore, the Kantorovich (pseudo)metric seems to be useful for logics and could simplify the proof of a Hennessy-Milner
theorem for quantitative coalgebraic logics. If we assume that modal
logic formulas are denoted by $\phi$ and, for a given coalgebra
$c\colon X\to FX$, are associated with a function
$\llbracket \phi\rrbracket_c\colon X\to [0,\top]$, a quantitative
version of the Hennessy-Milner theorem reads as follows:
\[ d_c(x,y) = \sup_\phi d_e(\llbracket
\phi\rrbracket_c(x),\llbracket\phi\rrbracket_c(y))\,. \] If we
furthermore observe that bisimulation-invariant modal formulas
generalize to nonexpansive quantitative formulas, this formula is
similar in nature to the definition of the Kantorovich lifting.

Note that in several cases both liftings coincide, making the
decision easy. In all other cases the Wasserstein (pseudo)metric
is larger than the Kantorovich (pseudo)metric.

\subsection{Related Work}
The ideas for our framework are not only heavily influenced by transportation theory \cite{Vil09} but also by work on quantitative variants of (bisimulation) equivalence of probabilistic systems. In that context \name[Alessandro]{Giacalone}, \name[Chi-Chang]{Jou} and \name[Scott A.]{Smolka} observed in the nineties that probabilistic Larsen-Skou bisimulation \cite{LS89} is too strong and therefore introduced a metric based on the notion of $\epsilon$-bisimulations \cite{GJS90}. Such a bisimulation is a relaxation of the usual probabilistic bisimulation relation which allows matching the steps of another state with a probability that is not exactly same but can be at most $\epsilon$ apart with respect to the Euclidean metric on $[0,1]$. Based on this, two states are exactly $\epsilon$ apart if this is the smallest value such that the two states are $\epsilon$-bisimilar.

A second approach to behavioral distances is based on logics. Labeled Markov processes (\textlsc{LMP}) are generalizations of reactive probabilistic transition system to fairly arbitrary (namely analytic) state spaces which involve some measure theoretic results. Surprisingly, probabilistic bisimilarity for these systems can be expressed via a simple modal logic without negation \cite{DEP98} in the sense that two states of an \textlsc{LMP} are bisimilar if and only if they satisfy the same formulae. Using this logical framework \name[Josée]{Desharnais}, \name[Vineet]{Gupta}, \name[Radha]{Jagadeesan} and \name[Prakash]{Panangaden} defined a family of metrics between \textlsc{LMP}s \cite{DGJP04} via functional expressions, which can be understood as quantitative generalization of the logical formulae. If evaluated on a state of an \textlsc{LMP}, such a functional expression measures the extent (as real number between $0$ and $1$) to which a formula is satisfied in that state. Then, for any set of such functional expressions, the distance of two \textlsc{LMP}s is the supremum of all differences (with respect to the Euclidean distance on $[0,1]$) of these functional expressions.

A third, coalgebraic approach, which inspired us to develop our framework, is used by \name[Franck]{van Breugel} and \name[James]{Worrell} \cite{vBW05,vBW06}. As we have presented in the examples in this paper, they define both a discounted and an undiscounted pseudometric on probabilistic systems via a fixed point approach using the usual Kantorovich pseudometric for probability measures. Moreover, they show that this metric is related to the logical pseudometric by Desharnais et al. \cite{vBW05}. We quickly point out that metrics in a coalgebraic setting appeared already earlier in a paper by \name[Erik]{de Vink} and \name[Jan]{Rutten}. They used ultrametrics\footnote{An ultrametric is a reflexive and symmetric function $d \colon X^2 \to [0,1]$ satisfying the implication $d(x,y) \implies x = y$ and the strong triangle inequality $d(x,z) \leq \max \set{d(x,y), d(y,z)}$.} and the category of ultrametric spaces in order to define coalgebraic bisimulation for continuous probabilistic transition systems \cite{dVR97,dVR99}. However, they use it mainly as a technical tool to get a final coalgebra and not in order to study bisimilarity distances.

Not only the definition of distances for probabilistic systems and the study of their theoretical importance but also their efficient approximation or exact computation has been the focus of several recent research papers \cite{vBW01a,FPP04,TDZ11,CPP12}. In particular, \name[Di]{Chen}, \name[Franck]{van Breugel} and \name[James]{Worrell} proved that both the discounted and the undiscounted bisimilarity pseudometric for probabilistic systems can be computed exactly in polynomial time exploiting algorithms to solve linear programs \cite{CvBW12}. Taking some inspiration from this work, one year later \name[Giorgio]{Bacci}, \name[Giovanni]{Bacci}, \name[Kim]{Larsen} and \name[Radu]{Mardare} proposed an on-the-fly approach for the exact computation of bisimilarity distances \cite{BBLM13a} which they proved to be practically much more efficient than the earlier approximation algorithms.

Behavioral distances have not only been studied for probabilistic
systems but also for other types of transition systems. An example
which also appears in the main text is the branching distance for
metric transition systems \cite{dAFS09}. Moreover, a thorough
comparison of various different behavioral distances on labeled
transition systems has recently been carried out by
\name[Uli]{Fahrenberg}, \name[Axel]{Legay} and
\name[Claus]{Thrane}. They transfer \name[Rob]{van Glabbeek}'s
quantitative linear-time--branching time spectrum \cite{vG90} to a
quantitative setting \cite{FLT11,FL14}.

The lifting of a monad to the bicategory of $\mathcal{V}$-matrices
(where $\mathcal{V}$ is a quantale), similar to our Wasserstein
lifting, has earlier been studied in \cite{h:closed-objects}.  After
our original publication \cite{BBKK14} on lifting functors to $\PMet$,
there have been other suggestions to categorically characterize such
liftings in even more general settings
\cite{ks:codenisity-liftings-monads,bkv:extensions-v-cat}. Both
approaches are based on Kan extensions, where the work by
\name[Shin-ya]{Katsumata} and \name[Tetsuya]{Sato} generalizes the
Kantorovich lifting and the method by \name[Adriana]{Balan},
\name[Alexander]{Kurz} and \name[Ji\v{r}\'{\i}]{Velebil} is
reminiscent of the Wasserstein lifting.

\subsection{Future Work}
This paper proposes a paradigmatic shift from qualitative to quantitative behavior analysis and although many basic results are in place there is still a lot of work ahead. We will first discuss a few open questions whose answers (if they exist) might yield improvements of our current framework. Then we discuss further possible generalizations.

In light of the fact that we propose two different lifting approaches, the first apparent question is whether there are conditions that guarantee that these liftings coincide, i.e., such that the Kantorovich-Rubinstein duality holds. However, some preliminary attempts suggest that this is very difficult. The proof for the duality in the (arbitrary) probabilistic setting is domain specific and cannot easily be generalized.

Another valid question concerning the two different liftings is whether they can be captured by some \emph{universal properties}. Although we use a coalgebraic and thus category theoretic framework, our intuition comes from transportation and probability theory. It would be interesting to figure out whether there is some general category theoretic construction such that our liftings are two ends of this construction in a similar way as the Kleisli and Eilenberg-Moore categories of a monad are initial and final objects in a category of adjunctions. A possible source of inspiration for this line of work could be \name[Franck]{van Breugel}'s draft paper on the metric monad \cite{vB05} which describes a generalization of the (continuous) Giry monad in terms of universal properties involving monad morphisms.

In order to obtain trace pseudometrics we employed the generalized
powerset construction by \name[Alexandra]{Silva},
\name[Filippo]{Bonchi}, \name[Marcello]{Bonsangue} and
\name[Jan]{Rutten} \cite{SBBR13}, one of the two most prominent
coalgebraic approaches to traces. Another well-known approach,
suggested by \name[Ichiro]{Hasuo}, \name[Bart]{Jacobs} and
\name[Ana]{Sokolova} \cite{HJS07}, is to work in the Kleisli
category. Whether our framework can  be modified or extended to work
in that setting remains an interesting open question. A possible basis
for answering this could be the recent comparison between these two
approaches in the qualitative setting which was carried out by
\name[Bart]{Jacobs}, \name[Alexandra]{Silva} and \name[Ana]{Sokolova}
\cite{JSS15}. However, this is already quite complicated in the
qualitative case, and thus very likely to be even more involved in the
quantitative case. In particular, for this comparison one needs -- in
addition to the distributive law (\EM-law) -- a \Kl-law and a so-called extension natural transformation which connects the two distributive laws.

While our theory is already at a quite general level, there are several further possible generalizations. Among these, one could be to drop symmetry and study so-called \emph{directed pseudometrics} as is done in the case of metric transition systems \cite{dAFS09}. This would lay the foundation to study simulation distances from a coalgebraic perspective.

An even more general idea is to study certain reflexive functions $d \colon X^2 \to L$ as generalized metrics, where $L$ is a  (complete) lattice with possibly some additional structure. This could result in a theory in which one can model distances for conditional transition systems (\textlsc{CTS}). These systems were proposed by \name[{Ji\u{r}\'{\i}}]{Adámek}, \name[Filippo]{Bonchi}, \name[Mathias]{Hülsbusch}, \name[Barbara]{König}, \name[Stefan]{Milius} and \name[Alexandra]{Silva} \cite{ABH+12} and are similar to featured transition systems \cite{CHS+10}. Formally, a \textlsc{CTS} is a labeled transition system with the following semantics. Once the environment chooses a label (which represents a condition or a feature), all transitions with this label remain (but the label is dropped) and all the other transitions vanish. Then one is interested in the behavior of the resulting unlabeled transition system. If $A$ is the set of conditions, we could take the complete lattice (with respect to the subset ordering) $L=\Powerset A$ as codomain of our generalized metrics. A natural distance of two states $x,y$ could be the set of conditions for which they are not bisimilar. This is a further generalization of a distance recently proposed by \name[Joanne M.]{Atlee}, \name[Uli]{Fahrenberg} and \name[Axel]{Legay} for featured transition systems: their (simulation) distance only counts how many features prevent simulation \cite{AFL15}.

A less drastic, yet interesting generalization could result from replacing the Euclidean metric with a different metric which has recently been suggested by \name[Konstantinos]{Chatzikokolakis}, \name[Daniel]{Gebler}, \name[Catuscia]{Palamidessi} and \name[Lili]{Xu} for probabilistic systems \cite{CGPX14}. Our Kantorovich lifting is obviously based on the Euclidean metric and we could simply replace $d_e \colon [0,\top]^2 \to [0,\top]$ with a different metric in this definition and study the resulting liftings. However, it is unclear how to proceed for the Wasserstein lifting since the Euclidean metric only plays a role in \Cref{W2} of the well-behavedness of the respective evaluation function.

The use of so-called up-to techniques \cite{San98} can significantly reduce both memory consumption and running time for equivalence checks between two specific states of a labeled transition system. This was demonstrated recently by \name[Filippo]{Bonchi} and \name[Damien]{Pous} for the equivalence check of nondeterministic finite automata \cite{BP13}. Together with \name[Daniela]{Petri{\c s}an} and \name[Jurriaan]{Rot}, they discovered a fibrational basis for these techniques \cite{BPPR14,Rot15}. Based on some preliminary research, we know that our entire framework can be seen from a fibrational perspective as well, leading to yet another generalization. We hope that this fibrational view will lead to more general proofs and efficiently computable algorithms, possibly by using the aforementioned up-to techniques.

Another valid question is, if (and how) our techniques can be generalized to other categories than $\Set$. In order to talk about distances between states it is likely that this will require at least concrete categories.

As we have discussed above, for probabilistic systems it is possible to efficiently compute the lifted distance by using ideas from linear programming. While our framework provides a solid theoretical basis for reasoning about several behavioral distances, its algorithmic applicability as yet is somewhat limited. For this we would need
\begin{itemize}
	\item an efficient method to compute the lifted distances, and possibly also
	\item an efficient method to automatically derive the fixed point of \Cref{lem:lifting-coalg}.
\end{itemize}
However, it is most likely that the efficiency of these methods will to a great extent rely on the specific system under consideration.

Last but not least, the coalgebraic theory in $\Set$ has benefited a lot from the huge amount of examples. While we have already looked at several examples throughout this paper, there is yet a lot more examples that should be worked out explicitly. In particular, it would be interesting to see if (and how) we can recover other behavioral distances from the literature.

\section*{Acknowledgement}
The authors are grateful to \new{Larry Moss,} Franck van Breugel, Neil
Ghani and Daniela Petri{\c s}an for several precious suggestions and
inspiring discussions.  The first author acknowledges the support of
the MIUR PRIN project CINA and University of Padova project
ANCORE. The second author acknowledges the support by project
ANR-16-CE25-0011 REPAS. The third and fourth authors acknowledge the
support of DFG project BEMEGA.

\bibliographystyle{alpha}
\bibliography{coalg-beh-dist}
\end{document}